\documentclass[
times,
aps,
twocolumn,
showpacs,
superscriptaddress,
prx,
floatfix]{revtex4-1}

\usepackage{graphicx}
\usepackage{array}
\usepackage{color}
\usepackage{upgreek}
\usepackage{float}
\usepackage{enumerate}
\usepackage{enumitem}
\usepackage{lipsum}
\usepackage{booktabs}
\usepackage{url}
\usepackage{braket}

\usepackage{amsthm}
\usepackage{bm}% bold math
\usepackage[T1]{fontenc}
\usepackage{scrextend}
\usepackage{hyperref}
\usepackage[dvipsnames]{xcolor}
\definecolor{C1}{RGB}{52, 89, 149}
\definecolor{C2}{RGB}{251, 77, 61}
\definecolor{C3}{RGB}{3, 206, 164}
\definecolor{C4}{RGB}{202, 21, 81}
\hypersetup{colorlinks=true, linkcolor=C2, citecolor=C2, urlcolor=C2}
\usepackage{dsfont}
\usepackage{epstopdf}
\usepackage{tikz}
\usepackage{mathrsfs}
\usepackage[caption=false]{subfig}
\usepackage{fdsymbol}
\usepackage{thm-restate}

\usepackage{accents}

\newtheorem{mydef}{Definition}

\theoremstyle{remark}

\newcommand*{\ups}{\Upsilon}
\newcommand*{\nn}{\nonumber}

\newcommand*{\id}{\mathds{1}}

\newcommand*{\mc}{\mathcal}
\newcommand*{\dg}{\dagger}
\newcommand*{\ex}{\mathrm{e}}

\newcommand*{\pur}{\tr[\ups_{B    R_1}^2]}
\DeclareMathOperator{\tr}{tr}

\begin{document}
\title[]{An Operational Metric for Quantum Chaos and the Corresponding Spatiotemporal Entanglement Structure} 

\author{Neil Dowling}
\email[]{neil.dowling@monash.edu}
\address{School of Physics \& Astronomy, Monash University, Clayton, VIC 3800, Australia}

\author{Kavan Modi}
\email[]{kavan.modi@monash.edu}
\address{School of Physics \& Astronomy, Monash University, Clayton, VIC 3800, Australia}
\affiliation{Centre for Quantum Technology, Transport for New South Wales, Sydney, NSW 2000, Australia}

\date{\today}

\begin{abstract}
    Chaotic systems are highly sensitive to a small perturbation, and are ubiquitous throughout biological sciences, physical sciences and even social sciences. Taking this as the underlying principle, we construct an operational notion for quantum chaos. 
    Namely, we demand that the future state of a many-body, isolated quantum system is sensitive to past multitime operations on a small subpart of that system. By `sensitive', we mean that the resultant states from two different perturbations cannot easily be transformed into each other. That is, the pertinent quantity is the complexity of the effect of the perturbation within the final state.
    From this intuitive metric, which we call the Butterfly Flutter Fidelity, we use the language of multitime quantum processes to identify a series of operational conditions on chaos, in particular the scaling of the spatiotemporal entanglement. Our criteria already contain the routine notions, as well as the well-known diagnostics for quantum chaos. 
    This includes the Peres-Loschmidt Echo, Dynamical Entropy, Tripartite Mutual Information, and Local-Operator Entanglement. We hence present a unified framework for these existing diagnostics within a single structure. We also go on to quantify how several mechanisms lead to quantum chaos, such as evolution generated from random circuits. Our work paves the way to systematically study many-body dynamical phenomena like Many-Body Localization, measurement-induced phase transitions, and Floquet dynamics. 
\end{abstract}


\maketitle
% \tableofcontents 

\section{Introduction}
Chaos as a principle is rather direct; a butterfly flutters its wings, which leads to an effect much bigger than itself. In other words, something small leading to a very big effect. This effect arises in a vast array of fields, from economics and ecology to meteorology and astronomy, spanning disciplines and spatiotemporal scales.

Chaos at the microscale, on the other hand, is an exception. Quantum chaos is not well understood and lacks a universally accepted classification. There is a vast web of, often inconsistent, quantum chaos diagnostics in the literature~\cite{Kudler-Flam2020}, which leads to a muddy picture of what this concept actually means. In contrast, classically chaos is a relatively complete framework. If one perturbs the initial conditions of a chaotic dynamical system, they see an exponential deviation of trajectories in phase space, quantified by a Lyapunov exponent. Trying to naively extend this to quantum Hilbert space immediately falls short of a meaningful notion of chaos, as the unitarity of isolated quantum dynamics leads to a preservation of fidelity with time. How then, can there possibly be non-linear effects resulting from the linearity of Schr\"odinger's equation? We will see that the structure of entanglement holds the key to this conundrum.

Yet, much effort has been made to understand quantum chaos primarily as the cause of classical chaos~\cite{BerryTabor1977,Lindblad1986chaos,reichl2021transition,Haake2018-cs}, to identify the properties that an underlying quantum system requires in order to exhibit chaos in its semiclassical limit. An example of this is the empirical connection between random matrix theory and the Hamiltonians of classically chaotic systems~\cite{BerryTabor1977}. Recently, with experimental access to complex many-body quantum systems with no meaningful classical limit, and given progress in related problems such as the black hole information paradox~\cite{Hayden_Preskill_2007,Shenker_Stanford_2014} and the quantum foundations of statistical mechanics~\cite{Popescu2006,Gogolin,Rigol2016}, quantum chaos as a research program has seen renewed interest across a range of research communities. In this context, a complete structure of quantum chaos, independent of any classical limit, is highly desirable but remains absent. 

In this work we approach quantum chaos from an operational, and theory agnostic,
principle: \emph{Chaos is a deterministic phenomenon, where the future state has a strong sensitivity to a local perturbation in the past.} For quantum processes the key ingredient will turn out to be spatiotemporal entanglement. To get there, we first identify the underlying definition of chaos as a starting point, and build a quantum butterfly flutter process from this fundamental principle. With this, we construct a genuinely quantum measure for chaos, based solely on this statement, which we term the \emph{Butterfly Flutter Fidelity}. This relies on the intuition that it is the complexity induced by a perturbation in the resultant future pure state, rather than just orthogonality, that dictates a chaotic effect. We adapt this principle into the theory of quantum processes and exploit their multitime structures. Namely, we use a tool from quantum information theory -- process-state duality -- to determine a hierarchy of necessary conditions on meaningful notions of chaos in many-body systems. These conditions culminate into the novel, strong metric of the Butterfly Flutter Fidelity.

\begin{figure*}
  \includegraphics[width=0.99\textwidth]{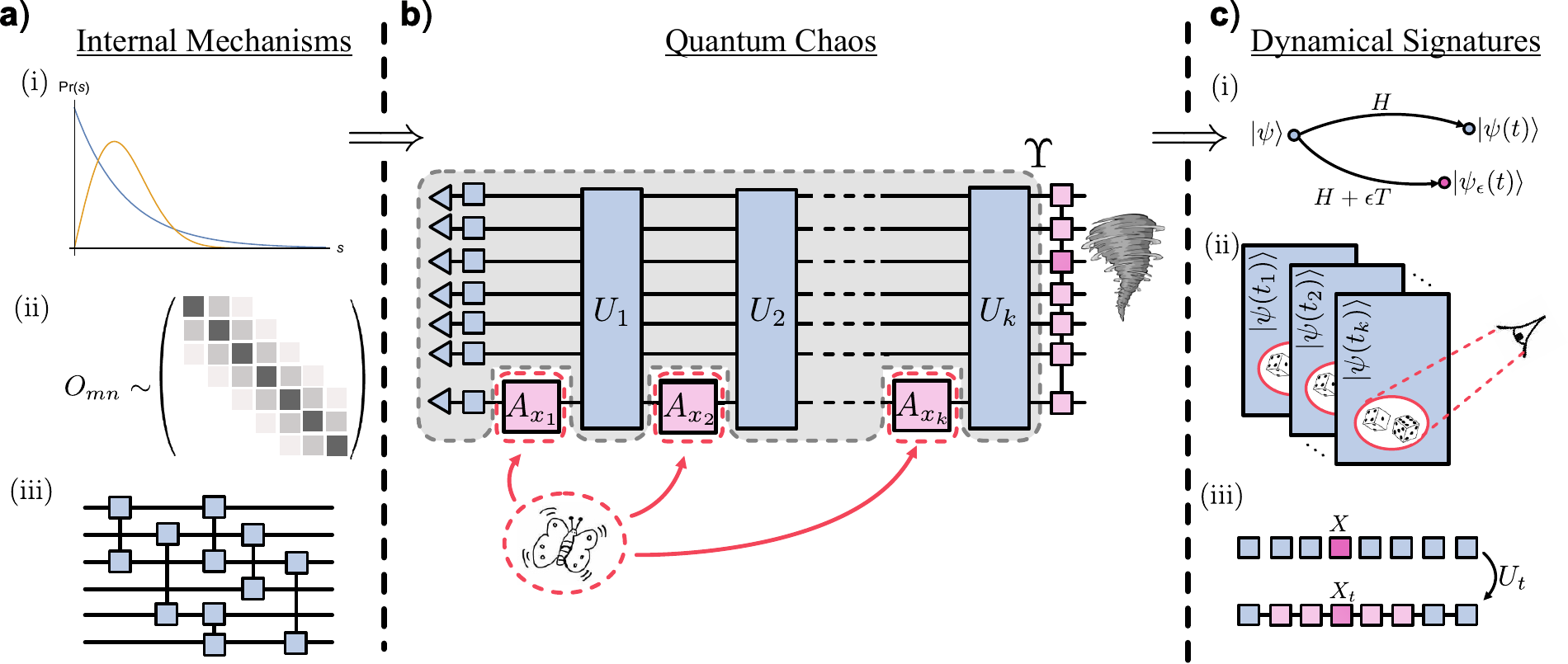}
  \caption{\footnotesize{A schematic of the causes, structure, and effects of quantum chaos. \textbf{a)} Internal mechanisms of chaos are the intrinsic properties of the dynamics that lead to chaotic effects. For example, properties of the Hamiltonian such as (i) level spacing statistics and (ii) the Eigenstate Thermalization Hypothesis (ETH), or properties of the quantum circuit describing the dynamics such as (iii) whether it forms a unitary design. \textbf{b)} In this work we will identify general quantum butterfly flutter protocol, and from this argue that chaos reduces to a hierarchy of conditions on the process describing the dynamics, including the volume-law spatiotemporal entanglement structure. This principle forms the stepping stone between causal mechanisms of chaos and observable diagnostics of chaos. We remark that we only conjecture that level spacing statistics and ETH (panels a.i and a.ii) lead to quantum chaos as formalized in this paper, and that these relationships form an interesting open question. \textbf{c)} Operational diagnostics for quantum chaos. Some popular probes include (i) The Peres-Loschmidt Echo, also known as Fidelity Decay or Loschmidt echo, which is the measure of the deviation between states, for evolution under a perturbed compared to an unperturbed Hamiltonian~\cite{Peres1984stab,Emerson2002}; (ii) The Dynamical Entropy, which quantifies how much information one gains asymptotically from repeatably measuring a subpart of a quantum system~\cite{Lindblad1986chaos,Pechukas1982,Slomczynski1994-ns,Alicki1994-dc}; and (iii) Local Operator Entanglement, which measures the complexity of the state representation of a time evolved Heisenberg operator~\cite{Prosen2007,Prosen2007a,Kos2020}. Another example which we analyze in this work (not shown) is the Tripartite Mutual Information, which measures entanglement properties of a state representation of a local input space of a channel together with a bipartition of the output space~\cite{Hosur2016}.  }
   }\label{fig:internal-external}
\end{figure*}

Fig.~\ref{fig:internal-external} breaks up the problem of quantum chaos into three broad components, laying out a review of the landscape of this multidisciplinary field and contextualizing our results. Panel (a) represents the mechanisms by which quantum chaos arises. Our contribution, depicted in Panel (b), is to identify a strong, operational criterion for quantum chaos through sensitivity in a future state, to the spatiotemporal quantum entanglement of the corresponding process. We propose that this intuitive metric bridges the gap between the mechanisms in Panel (a) and the signatures for chaos depicted in Panel (c). We provide explicit connections between several elements of these panels in this work, whose details we outline below.

Specifically, affirming the validity of our approach, we show that our criterion is stronger than and encompasses existing dynamical signatures of chaos. We show this explicitly for the Peres-Loschmidt Echo, Dynamical entropy, Tripartite Mutual Information, and Local-Operator Entanglement.\textsuperscript{\footnote{These are a selection of some of the most popularly accepted chaos probes, but this is by no means an exhaustive list. For example, some interesting alternative measures are entanglement spectrum statistics~\cite{Chamon2014-cb}, the Spectral Form Factor~\cite{Haake2018-cs}, and quantum coherence measures~\cite{Anand2021-yi}. Also see Refs.~\cite{HunterJones2018ChaosAR,Kudler-Flam2020} and references therein.}} That is, we identify the underlying structure leading to characteristic chaotic behavior of each of these popular chaos diagnostics. We offer a clear hierarchy of conditions of a chaotic effect, due to a `butterfly flutter', unifying a range of (apparently) inconsistent diagnostics.

Next, we show that there are several known mechanisms for quantum processes that lead to quantum chaos. In particular, we show that both Haar random unitary dynamics and random circuit dynamics -- which lead to approximate $t-$design states -- are highly likely to generate processes which satisfy our operational criterion for quantum chaos. Our results also open the possibility of systematically studying other internal mechanisms thought to generate quantum chaos, e.g. Wigner-Dyson statistics~\cite{BerryTabor1977}, or the Eigenstate Thermalization Hypothesis (ETH)~\cite{Deutsch1991, Srednicki,Rigol2008}.

Finally, Our approach is different from previous works that have usually relied on averages over Haar and/or thermal ensembles to draw connections between some previous signatures for quantum chaos~\cite{Yan2020,Roberts2017-en,Leone2021-ov}. We work solely within a deterministic, pure-state setting, identifying a series of conditions which stem from a sensitivity to past, local operations, without any need to average over operators or dynamics. Moreover, other metrics for quantum chaos also start from the notion of a kind of a butterfly effect, such as the out-of-time order correlator (OTOC)~\cite{Roberts2015shocks}. However, our sense of this intuitive idea is different, and does not necessarily suffer the same shortfalls as e.g. the OTOC which decays quickly even for some integrable systems~\cite{Pilatowsky2020,Xu2020,dowling2023scrambling}.

\begin{figure}[t]
  \includegraphics[width=0.49\textwidth]{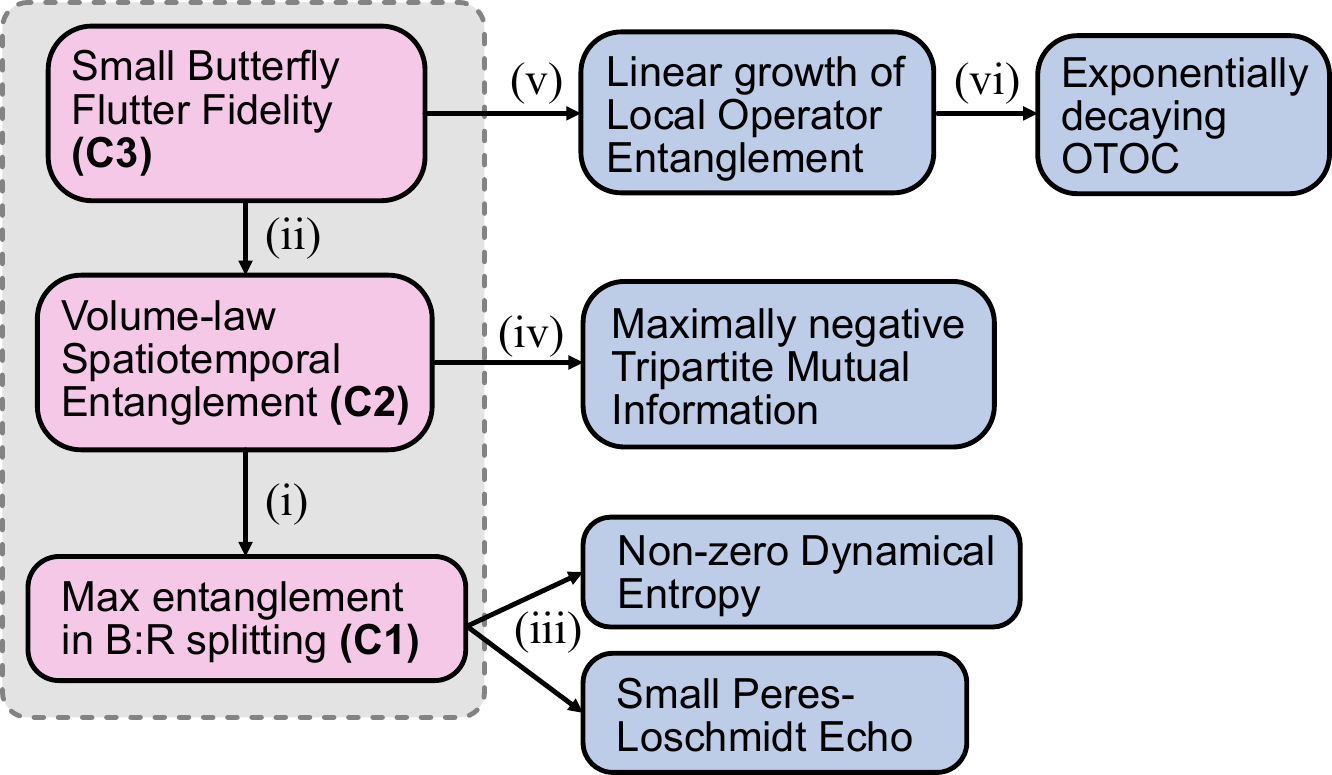}
  \caption{\footnotesize{A summary of the results of this work, where directed arrows mean implication. The shaded region with pink boxes is the hierarchy of conditions on quantum chaos as a sensitivity to perturbation proposed in this work, \ref{c1}-\ref{c3}. (i) Volume-law spatiotemporal entanglement of $\ket{\ups}$ is strictly stronger than maximum entanglement in the single bipartition $B:R$. (ii) A small Butterfly Flutter Fidelity (Def.~\ref{def:BFF}) necessarily implies the volume-law spatiotemporal entanglement of $\ket{\ups}$ (Prop.~\ref{thm:zeta-ST}), with equivalence when the initial state is area-law (Prop.~\ref{thm:state_chaos} ). (iii) The (trotterized) Peres-Loschmidt Echo constitutes the particular case of an asymptotically many-time, weak unitary butterfly flutter (Section~\ref{sec:LE_comb}), while an extensive Dynamical entropy is equivalent to an extensive entanglement scaling in the splitting $B:R$ (Prop.~\ref{thm:BR_ent} and Section~\ref{sec:Sdy}). (iv) For a single-time butterfly flutter, volume-law spatiotemporal entanglement  directly implies a (near) maximally negative Tripartite mutual information of the corresponding channel (Prop.~\ref{thm:tmi}). (v) For a single-time butterfly flutter, if the Butterfly Flutter Fidelity is small for any initial state, then for an operator entanglement complexity measure the Local Operator Entanglement grows linearly with time (Thm.~\ref{thm:LOE_BFF}). (vi) If the Local Operator Entanglement grows linearly with time, then general OTOCs necessarily decay exponentially~\cite{dowling2023scrambling}. }}
\label{fig:flow}
\end{figure}

\subsection*{Summary of Main Result}
We first give an informal explanation of the main innovation of this work. We use a simplified formalism and setup in order to convey the main ideas, with a more detailed exposition to be given later. 

Consider an isolated quantum system where a sequence of $k$ unitaries $A_{x_i}$ are applied on a local subspace $S$, such that the global system is defined on the Hilbert space $\mc{H}_S \otimes \mc{H}_E$. Later we will call this sequence a \emph{butterfly flutter}, and allow it to consist of an arbitrary sequence of rank-one instruments (Def.~\ref{def:butterfly}). The outgoing state after this protocol is 
\begin{equation} \label{eq:BFout}
    \ket{\ups_{R|\vec{x}}} = {A}_{x_k} {U}_k \cdots {A}_{x_2} {U}_2 {A}_{x_1} {U}_1\ket{\psi_{SE}}
\end{equation}
where $U_i$ represents global unitary evolution, either Floquet or according to a Hamiltonian for time $t_i$, and where $A_{x_i} \equiv (A_{x_i})_S \otimes \id_E$. The other choices of notation will become apparent in the following sections. 

Now we similarly introduce a strictly different set of $k$ unitaries, labeled by the list $\vec{y}$. We take these unitaries to be orthogonal to the first choice, in the Hilbert-Schmidt sense such that $\tr[A_{x_i}^\dg A_{y_i}   ]= 0$ for all $i \in [1,k]$. Note that we impose no such constraint on operations for different times, $A_{x_i}$ compared to $A_{x_j}$ with $i\neq j$. We later loosen this condition such that these can be collectively, approximately orthogonal operations. The outgoing state is defined analogously to Eq.~\eqref{eq:BFout}, with the same global dynamics $U_i$ and subsystem decomposition $\mc{H}_S \otimes \mc{H}_E$, but different unitary `perturbations'. We can then ask, how much do these two resultant states, $\ket{\ups_{R|\vec{x}}}$ and $\ket{\ups_{R|\vec{y}}}$, differ? 

This question is a direct translation to quantum mechanics, of the principle of chaos as a sensitivity perturbation. The task is to define exactly what we mean by this sensitivity. As mentioned in the introduction, fidelity is preserved under unitary evolution. Further, as we discuss in Section~\ref{sec:c1} and Appendix~\ref{sec:LB}, the fidelity cannot be the full story: most dynamics irrespective of integrability will lead to a small fidelity $| \langle \ups_{R|\vec{x}}  | \ups_{R|\vec{y}} \rangle|^2$. We will show that this orthogonality translates into an entropic condition on the underlying process for this perturbation protocol, namely that a genuinely chaotic system should necessarily have a volumetrically scaling spatiotemporal entanglement.

We instead strengthen this by defining a new metric to compare these states, which we call the \emph{Butterfly Flutter Fidelity} (Def.~\ref{def:BFF}). This compares how different the two final states are in a complexity sense, and measures the fidelity after what we call a correction unitary $V$
\begin{equation} \label{eq:state_chaos_prelim}
    \zeta :=  \underset{V \in \mc{R}}{\mathrm{sup}}  \left(| \langle \ups_{R|\vec{x}} | V | \ups_{R|\vec{y}} \rangle|^2 \right).
\end{equation}
This quantity is depicted graphically in Fig.~\ref{fig:zeta_cji} (a). 
Here, $\mc{R}$ is a restricted set of unitaries on $\mc{H}_S \otimes \mc{H}_E$, which for now can be considered to be the set of simple (low-depth) circuits. Intuitively, this measure~\eqref{eq:state_chaos_prelim} determines whether the orthogonality between $\ups_{R|\vec{x}}$ and $\ups_{R|\vec{y}}$ is complex or not. That is, is the sensitivity stemming from past perturbations (local unitaries) easily correctable? Based on our operational criteria for quantum chaos, we argue that the dynamics are chaotic if this not easily correctable -- when $\zeta \approx 0$ for an appropriately defined set of corrections $\mc{R}$, and for any choice of butterfly flutters. This notion of chaos then allows us to identify a connection with entanglement properties of the underlying process describing the `butterfly flutter' protocol. 

For example, one could choose two butterfly flutters as a sequence of $k$ Pauli $X$ gates on a single qubit of a many-body system at $k$ times, and the other to be a series of identity maps (do nothing). With free, global evolution occurring between each gate, the Butterfly Flutter Fidelity Eq.~\eqref{eq:state_chaos_prelim} would then indicate that the dynamics is chaotic if the fidelity between the final states is small, $\zeta \approx 0$, even after trying to align the two final states using any small depth, local circuit $V$. This quantity is the main focus of this work.

The rest of the paper is structured as the following: We present a review of the appropriate tools with which we need to analyze the Butterfly Flutter Fidelity Eq.~\eqref{eq:state_chaos_prelim} in Section~\ref{sec:processes}. This predominantly includes the theory of multitime quantum process~\cite{processtensor,processtensor2,milz2020quantum}, allowing us to describe all possible perturbations and resultant effects within a single quantum state. Then in Section~\ref{sec:mainresult} we present a set of increasingly stronger, necessary conditions on a dynamical process for which $\zeta \approx 0$ in Eq.~\eqref{eq:state_chaos_prelim}. These conditions are all motivated from the principle of chaos as a sensitivity to perturbation, and start with a minimal sense of what a large effect could be, stemming from a past, local perturbation. This main results section then culminates in the Butterfly Flutter Fidelity, as the strongest condition in this hierarchy. We conclude this section by comparing Butterfly Flutter Fidelity to the classical ideas of chaos, and detailing how one could in-principle measure it in experiment.

In Section~\ref{sec:prev_chaos}, we support the proposed conditions by showing how a range of previous dynamical signatures of chaos agree with them, as depicted in Fig.~\ref{fig:internal-external}. We summarize these connections in Fig.~\ref{fig:flow}, which serves as a summary of this work and the related work of Ref.~\cite{dowling2023scrambling}. Finally, in Section~\ref{sec:typ_process} we discuss mechanisms of chaos that lead to the operational effects we propose. In particular, we prove that random dynamics -- both fully Haar random and those generated by unitary designs -- typically lead to chaos.

\section{Tools: Multitime Quantum Processes and Spatiotemporal Entanglement} \label{sec:processes}
Many of the results of this work rely on the application of ideas from entanglement theory to multitime quantum processes, in order to interpret the overarching problem of chaos in isolated many-body systems. We here give only an overview of the relevant facets of this topic, and refer the reader Appendix~\ref{ap:processes} for more information, and to Refs.~\cite{OperationalQDynamics,milz2020quantum} for a more complete introduction to the process tensor framework.

Consider a finite dimensional quantum system. A quantum process is a quantum dynamical system under the effect of multitime interventions on some accessible local space $\mc{H}_S$. These interventions are described by instruments, which trace non-increasing quantum maps. The dynamics between interventions can then be dilated to a system-environment $\mc{H}_S \otimes \mc{H}_E$, such that the total isolated state on $\mc{H}_S \otimes \mc{H}_E$ evolves unitarily on this extended space. A $k-$step process tensor is the mathematical description of a such a process, encoding all possible spatiotemporal correlations in a single object; analogous to how a density matrix encodes single-time measurements.

 In this work we will generally consider rank-one instruments, such as unitary matrices and projective measurements (including the outgoing state). In this case, we are able to write down the full pure state on $\mc{H}_S \otimes \mc{H}_E$ at the end of this process, 
 \begin{equation} \label{eq:pure_process}
    \ket{\psi_{SE}^\prime} = {U}_k {A}_{x_k} {U}_{k-1} \dots {U}_1 {A}_{x_1} \ket{\psi_{SE}},
\end{equation}
where we have rewritten this as the conditional state of a subpart of process $\ket{\ups}$, and will explain exactly what this means below. $A_{x_i}$ can be arbitrary norm non-increasing operators, with $\sum_{x_i} A_{x_i}^\dg A_{x_i} = \id $. That is, anything that maps pure states to (possibly subnormalized) pure states. This includes e.g. unitary operators or projective measurements. We stress that $A_{x_i}$ are considered to act locally on $\mc{H}_S$, such that $A_{x_i} \equiv A_{x_i}^{(S)} \otimes \id^{(E)}$. 
As everything is pure here, there is no need to consider superoperators or density matrices, and left multiplication by matrices is a sufficient description (see Appendix~\ref{ap:processes} for the mixed-state extension of this). $\ket{\ups_{R|\vec{x}}}$ could be a sub-normalized pure state, for example if the instruments chosen to be a series of projective measurements 
\begin{equation} \label{eq:CondState}
    \ket{\psi_{SE}^\prime}  = \sqrt{p_{\vec{x}}} \ket{\ups_{R|\vec{x}}} .
\end{equation}
Here, $A_{x_i} = \ket{x_i} \bra{ x_i}$, $p_{\vec{x}}$ is the probability of observing this outcome, and where we have neglected the (unobservable) global phase. We will usually consider the (normalized) conditional state $\ket{\ups_{R|\vec{x}}}$ when investigating chaotic effects, as we will be concerned with the resultant state rather than the probability that it is produced. 

Rather than choosing a particular instrument $A_{x_i}$ for each intervention, one can instead feed in half of a maximally entangled state from an ancilla space, as shown in Fig.~\ref{fig:CJI}. This results in the pure state $\ket{\ups}$, encoding both any interventions on the multitime space in the past which we call $\mc{H}_B$, and the final pure state on the global, isolated system, on the space $\mc{H}_R$. This is the generalized Choi–Jamio\l kowski Isomorphism (CJI)~\cite{processtensor,processtensor2}, shown in Fig.~\ref{fig:CJI}. Alternatively to this ancilla-based construction, the pure process tensor can be defined succinctly as
\begin{equation} \label{eq:PurePT}
    \ket{\ups} := \ket{\mathtt{U}_k} * \dots * \ket{\mathtt{U}_1} * \ket{\psi_{SE} (t_1)},
\end{equation}
where $*$ is the Link product, corresponding to composition of maps within the Choi representation~\cite{Chiribella2008}, and is essentially a matrix product on the $\mc{H}_E$ space, and a tensor product on the $\mc{H}_S$ space. A ket of a rank-one instrument $A$ corresponds to the single-time Choi state  
\begin{equation} \label{eq:1CJI}
    \ket{A} := (A  \otimes \id)   \ket{\phi^+},
\end{equation}
by the usual single-time CJI: channel-state duality~\cite{OperationalQDynamics}.

Here, we have gathered the multitime Hilbert space where the full multitime instruments act on a space with the single label,
\begin{align} \label{eq:Bdef}
    \mc{H}_B \equiv \mc{H}^\mathrm{io}_{S(t_{k-1})} \dots \otimes \mc{H}^\mathrm{io}_{S(t_1)} \otimes \mc{H}^\mathrm{io}_{S(t_0)}
\end{align}
called the `butterfly' space $\mc{H}_B$, where $\mc{H}^\mathrm{io}_{S(t_j)} \equiv \mc{H}^\mathrm{i}_{S(t_j)} \otimes \mc{H}^\mathrm{o}_{S(t_j)}$. $\mc{H}^\mathrm{i}$ represents the input space to the process, while $\mc{H}^\mathrm{o}$ represents the output. The `remainder' space $\mc{H}_R$ -- the full final state on the system plus environment at the end of the protocol, where the `butterfly' does not act -- is
\begin{align}\label{eq:Rdef}
    \mc{H}_R \equiv \mc{H}^\mathrm{o}_{S(t_k)} \otimes \mc{H}^\mathrm{o}_{E(t_k)}.
\end{align}
All of these are clearly labeled in Fig.~\ref{fig:CJI}. It will become apparent in the following section why we name these spaces as such. 
\begin{figure}[t]
\centering
\includegraphics[width=0.48\textwidth]{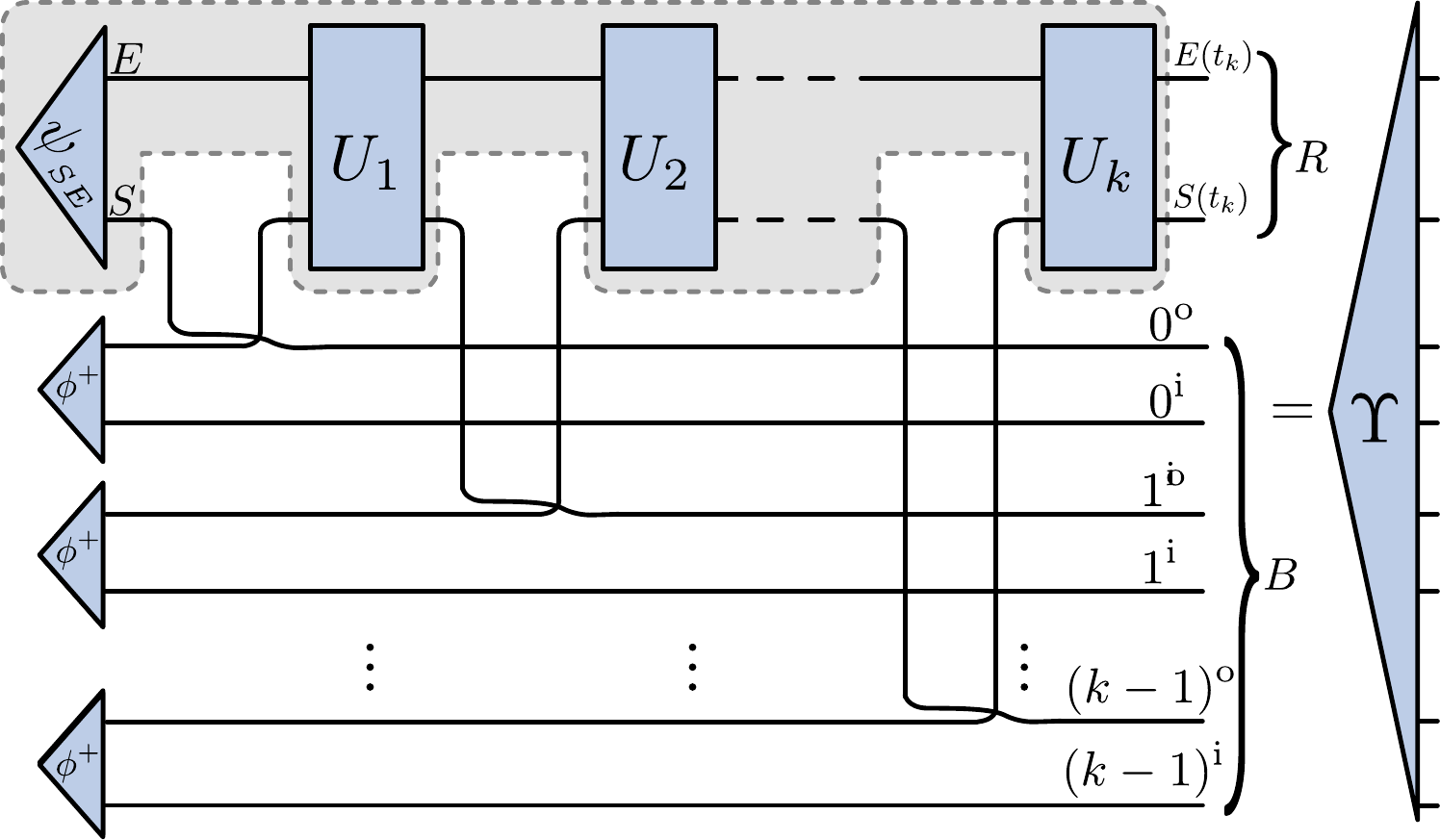}
\caption{\footnotesize{Tensor network diagram of the protocol producing the Choi state of a pure process tensor through the generalized Choi-Jamio\l{}kowski isomorphism~\cite{processtensor,watrous_2018}. This means that input indices are put on equal footing with output indices, through appending a maximally entangled ancilla system $\ket{\phi^+}$ at each time, and inserting half of this state into the process. The final output state of this protocol encodes all multitime spatiotemporal correlations: a pure process tensor. A multitime expectation value can then be computed in this representation by finding the Hilbert Schmidt inner product between this (normalized) Choi state and the (supernormalized) Choi state of a multitime instrument. The system $\mc{H}_S$ denotes the singletime space where instruments act, and the environment $\mc{H}_E$ the dilated space such that all dynamics are unitary. Here the independent Hilbert spaces are labeled such that $(\ell)^\mathtt{i}$ ($(\ell)^\mathtt{o}$) is the input (output) system space $\mc{H}_S$ at time $t_\ell$, showing that the final output $\ket{\Upsilon}$ corresponds to a $(2k+2)-$body pure quantum state. }} \label{fig:CJI}
\end{figure}

From Eq.~\eqref{eq:PurePT} we can determine the outgoing (possibly sub-normalized) state Eq.~\eqref{eq:CondState} from projections on this state,
\begin{equation}
    \ket{\psi^\prime_{SE}} = \braket{\vec{x}|\ups}.
\end{equation}
For independent instruments at each intervention time, we have that
\begin{equation} \label{eq:choi_instruments}
    \ket{\vec{x}} := \ket{x_k} \otimes \dots \otimes \ket{x_1}, 
\end{equation}
where each single-time state is constructed as in Eq.~\eqref{eq:1CJI}. Alternatively, one could trace over the final state on $\mc{H}_R$, and the reduced state on $\mc{H}_B$, $\ups_B$, is the process tensor~\cite{processtensor,processtensor2,milz2020quantum}, as we describe in Appendix~\ref{ap:processes}.

The key point here is that through the CJI we have reduced all possible correlations of a dynamical multitime experiment to a single quantum state, $\ket{\ups}$. This means that all the machinery from many-body physics is available to describe multitime effects. A subtle difference from the single-time case is that the normalization of these Choi states do not exactly correspond to the normalization of states and projections. Instruments are taken to be supernormalized, while processes have unit normalization and so constitute valid quantum states
\begin{align} \label{eq:normalization}
    &\braket{\ups|\ups}=1, \quad \text{and} \quad
    \braket{\vec{x}|\vec{x}} \leq d_S^{2k},
\end{align}
where the inequality is saturated for deterministic instruments: CPTP maps. This normalization ensures that one gets well defined probabilities in Eq.~\eqref{eq:CondState}.

% Therefore, we have a natural notion for sensitivity to
% an initial perturbation contained within the Butterfly
% Flutter Fidelity (22). We show in Section IV D that the
% Local-Operator Entanglement measures this single-time
% chaos, optimizing over any initial state. Further, it can
% be shown that out-of-time-order correlators generically
% probe this operator entanglement [27]. The hierarchy
% (C1)-(C3) gives a robust understanding of why these
% previous diagnostics measure chaos, in terms of a future
% sensitivity to past local operations.

Therefore, dynamical properties of a process such as: non-Markovianity~\cite{processtensor2, Costa2016,Markovorder1, Taranto2019FiniteMarkov,milz2020quantum}, temporal correlation function equilibration~\cite{Dowling2021,finitetime}, whether its measurement statistics can be described by a classical stochastic process~\cite{strasberg2019,Milz2020prx,Strasberg2022}, multipartite entanglement in time~\cite{Milz2021GME}, and other many-time properties~\cite{White2021}, can all be clearly defined in terms of properties of the quantum state $\ket{\ups}$. However, the spatiotemporal entanglement structure of this multitime object is largely unexplored, and we will show that this has vast implications for understanding quantum chaotic versus regular dynamics. 

Any pure quantum state $\ket{\psi}_{AB}$ on $\mc{H}_A \otimes \mc{H}_B$ can be decomposed across any bipartition $A:B$ via the Schmidt decomposition,
\begin{equation} \label{eq:schmidt}
    \ket{\psi}_{AB} = \sum_{i=1}^\chi \lambda_i \ket{\alpha_i}_{A} \ket{\beta_i}_{B},
\end{equation}
where $\braket{\alpha_i|\alpha_j} =\delta_{ij}=\braket{\beta_i|\beta_j}$. $\chi$ is called the \emph{bond dimension} or \emph{Schmidt rank}, dictating intuitively how much of the subsystems are entangled with each other. The bond dimension is equal to one if and only if the state is separable across $A:B$.

Using this decomposition \eqref{eq:schmidt}, one can successively increase the size of the subsystem $\mc{H}_A$, and determine how the bond dimension scales. We will deal with one spatial dimension systems when discussing spatiotemporal entanglement in this work, as characteristic entanglement scaling depends on the underlying geometry~\cite{Eisert_Cramer_Plenio_2010}. Our results should generalize in a straightforward way to higher spatial dimensions. If $\chi$ is bounded by $\min \{d_A,D\}$ for a constant $D<(d_A d_B)/2$ for any $\mc{H}_A$ with dimension $d_A$ up to half the total Hilbert space dimension, this is called area-law scaling. In this case for example, a one dimensional spin chain state may be written efficiently as Matrix Product State (MPS)~\cite{Fannes1992,Schuch2008MPS,Eisert_Cramer_Plenio_2010,Bridgeman_2017}. Despite being introduced in order to efficiently simulate the ground state of certain Hamiltonians, it was soon realized that a fundamental property of a state written as an MPS is revealed in the scaling of the bond dimension~\cite{Vidal20003}. On the other hand, if the bond dimension (approximately) scales extensively with the subsystem size, this is volume-law scaling. This directly implies a characteristically scaling entanglement entropy, 
\begin{equation}
    S(\rho_A ) \propto \log(d_A),
\end{equation}
where $S(\rho_A )$ is the von Neumann entropy of the reduced state $\rho_A$. Area-law can be defined formally as a bounded entanglement with scaling subsystem size, for all R\'enyi entropies~\cite{Schuch2008MPS}.  
Such scaling will synonymously be called \emph{entanglement structure} or \emph{entanglement scaling} throughout this work. We will show that this property within the pure process tensor $\ket{\ups}$ is intrinsically linked to the chaoticity of a quantum process.

We will now delve into our main result, interpreting the dynamical meaning behind spatiotemporal entanglement structure of quantum processes.

\section{Main Result: The Butterfly Flutter Fidelity and Spatiotemporal Entanglement Structure} \label{sec:mainresult}
\begin{figure*} 
  \includegraphics[width=0.97\textwidth]{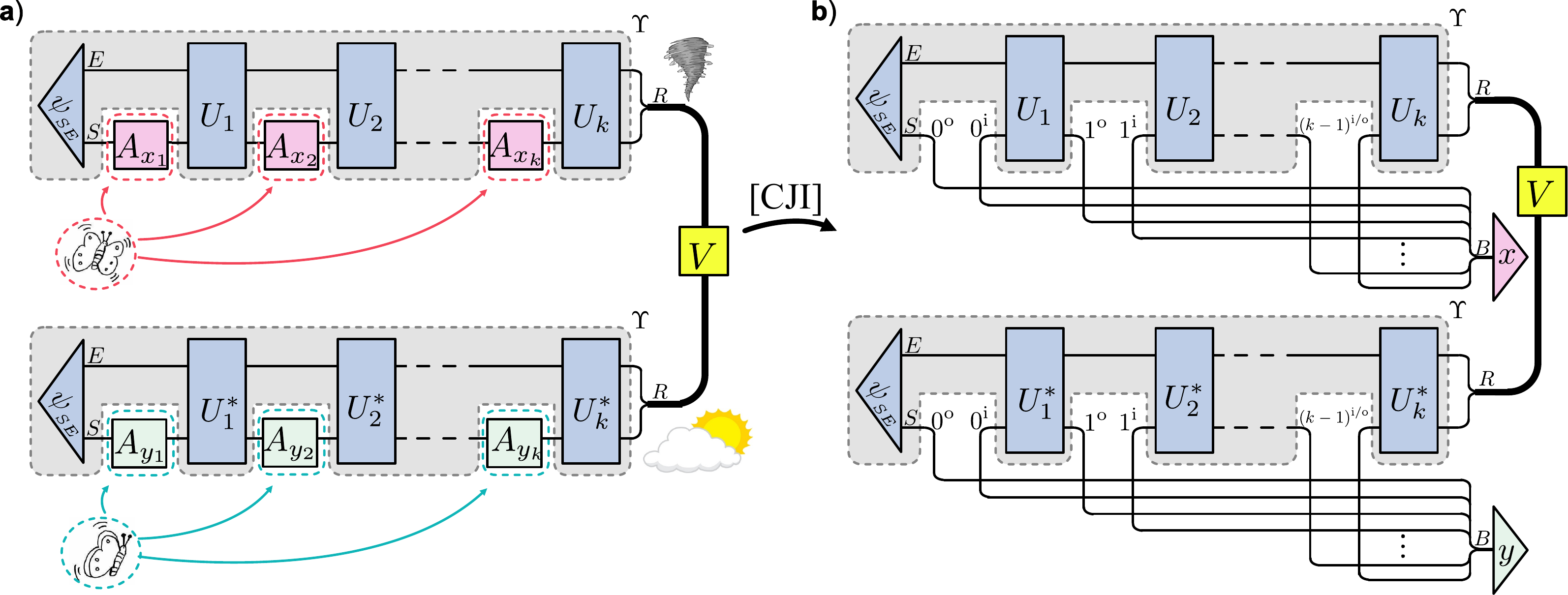}
  \caption{\footnotesize{Two equivalent representations of the Butterfly Flutter Fidelity, Def.~\ref{def:BFF}. \textbf{a)} The process representation of the Butterfly Flutter Fidelity. Two orthogonal sequences of instruments $\{ {A}_{x_i}\}$ and $\{{A}_{y_i}\}$ act at $k$ times on the system Hilbert space denoted by $\mc{H}_S$, of a time evolving state $\ket{\psi_{SE}}$. The final pure states on $\mc{H}_R=\mc{H}_S \otimes \mc{H}_E$ can be compared, with a simple correction unitary $V$ (partially) aligning the states, which enforces that the butterfly flutter's effect is complex on the final states. This is Eq.~\eqref{eq:state_chaos} in the text. \textbf{b)} Using the CJI, as described in Section~\ref{sec:processes} and Fig.~\ref{fig:CJI}, the process corresponding to the butterfly protocol can be mapped one-to-one to a quantum state $\ket{\ups} \in \mc{H}_B \otimes \mc{H}_R$. Then the Butterfly Flutter Fidelity \eqref{eq:state_chaos} corresponds to projecting onto the butterfly space $\mc{H}_B$ with two orthogonal projections $\bra{\vec{x}}$ and $\bra{\vec{y}}$, and comparing the resulting conditional states on $\mc{H}_R$. This allows us to interpret strong and complex effects due to the butterfly flutter, in terms of the entanglement properties of $\ket{\ups}$: a strong effect from entanglement in the bipartition $B:R$ (Proposition~\ref{thm:BR_ent}), and a complex effect from volume-law entanglement in the full state (Proposition~\ref{thm:state_chaos}).}}\label{fig:zeta_cji}
\end{figure*}
In the previous section we have defined a pure process tensor $\ket{\ups}$ which encodes an experiment where a local part $\mc{H}_S$ of a many-body quantum system is interacted with across $k$ times, together with the outgoing pure state on the space $\mc{H}_R$. The multitime intervention, which we call a butterfly flutter and is defined explicitly in Def.~\ref{def:butterfly}, is taken to act on the collective multitime `butterfly' space $\mc{H}_B$, Eq.~\eqref{eq:Bdef}. This formalism will allow us to identify necessary conditions stemming from the principle of quantum chaos as a sensitivity to perturbation, in terms of the properties of the state $\ket{\ups}$. Fig.~\ref{fig:zeta_cji} offers a graphical representation of the state-process duality which serves as a key tool of our analysis.  

We will now identify a series of conditions on the process $\ket{\ups}$, each stronger than the previous, such that if a process satisfies $\ref{c3}$ then it also necessarily satisfies $\ref{c2}$ and hence also $\ref{c1}$. We make an intuitive argument based on chaos as a sensitivity to perturbation, to argue for each condition. We will show in Section~\ref{sec:prev_chaos} that they are each related to previous signatures of chaos; see also Fig.~\ref{fig:internal-external}. 
\begin{enumerate}[label={\textbf{{(C\arabic*)}}}]
    \item \textit{(Perturbation orthogonalizes future state)} The final state on $R$ should be strongly sensitive to butterfly flutters on $B$.
    \item \textit{(Scrambling as volume-law entanglement)} Butterfly flutters on $B$ should affect a large portion of the final state on $R$.
    \item \textit{(Complexity of sensitivity)} Different butterfly flutters on $B$ should lead to different enough states on $R$, in a complexity sense. 
\end{enumerate}
From each of these, we will identify the properties of $\ket{\ups}$ that these conditions lead to. 
Of course, as written above these conditions are informal statements. We will spend the rest of this section making this precise, and restate this list at the end in full technical detail. 

\subsection{Sensitivity to Perturbation (C1)} \label{sec:c1}
Given a sequence of small interventions on a many-body system, what is the minimal effect on the final pure state such that it is sufficiently perturbed? As a minimal condition, we argue that a perturbation should orthogonalize this final state, in the usual sense of fidelity. We will show that this leads to a simple entropic condition on the pure process state $\ket{\ups}$. 

More technically, we first define explicitly what we mean by a perturbation which probes chaos. 
\begin{mydef} \label{def:butterfly}
    A butterfly flutter is multitime linear map with some outcome label $\vec{x}$, defined by $k$ rank$-1$ instruments $\{{A}_{x_1}, A_{x_2},\dots,A_{x_k} \}$, which maps a $k$ time pure process $\ket{\ups} \in \mc{H}_B \otimes \mc{H}_R$ to a normalized state,
    \begin{equation}
        \frac{\braket{\vec{x} | \ups }}{\sqrt{\braket{\ups|\vec{x}}\braket{\vec{x}|\ups}}}= \ket{\ups_{R|\vec{x}}}.
    \end{equation}
    Here, $\ket{\vec{x}} \in \mc{H}_B$ is the Choi state of the multitime instrument which defines the butterfly flutter, as in Eq.~\eqref{eq:choi_instruments}, and the (conditional) output state $\ket{\ups_{R|\vec{x}}}$ is defined below in Eq.~\eqref{eq:psiE}.
\end{mydef}
Note that butterfly flutters are distinct from the multitime instruments discussed in Section~\ref{sec:processes} only in that we take the normalized output from its action. This is important as we do not wish to consider the probability of a butterfly to occur, only its effect. $\ket{\ups_{R|\vec{x}}}$ is just the conditional pure state on the global $\mc{H}_S \otimes \mc{H}_E$ space. 

We can compare the two final conditional (pure) states after two distinct butterfly flutter protocols labeled by $\vec{x}$ and $\vec{y}$
\begin{equation} \label{eq:process_chaos}
    \mc{D}(\ket{\ups_{R | \vec{x}}} , \ket{\ups_{R | \vec{y}} } ).
\end{equation}
Here, $\mc{D}$ is some metric on pure quantum states, naturally taken to be the fidelity, and the label $\vec{w}=(w_1,\dots,w_k ) \in \{\vec{x}, \vec{y} \}$ denotes instruments acting at $k$ times, such that 
\begin{align} \label{eq:psiE}
    \ket{\ups_{R | \vec{w}} } := \frac{{A}_{w_k} {U}_k \cdots {A}_{w_2} {U}_2 {A}_{w_1} {U}_1\ket{\psi_{SE}}}{\sqrt{\bra{\psi_{SE}}{U}_1^\dg \cdots {A}_{w_k}^\dg {A}_{w_k} \cdots {U}_1\ket{\psi_{SE}}}}.
\end{align}
This is a bipartite quantum state after a butterfly flutter protocol, which may include a sequence of measurements and preparations on some local system labeled $\mc{H}_S$, recording the outcomes as $\vec{w}$. Alternatively, ${A}_{w_i}$ could be a unitary on some subspace, or any other quantum operation which could even be correlated across multiple times. Note that if two butterflies only consist of unitary maps, then the normalization in the denominator of Eq.~\eqref{eq:psiE} is simply equal to one. In the interest of identifying the general form of any quantum butterfly effect, we allow the perturbation to be any pure multitime instrument. 

Condition $\ref{c1}$ then means that 
\begin{equation} \label{eq:orthogg}
    |\braket{\ups_{R | \vec{x}} |\ups_{R | \vec{y}}}|^2 \approx 0,
\end{equation}
for any two orthogonal butterfly flutters $\ket{\vec{x}}$ and $\ket{\vec{y}}$. Our construction of dynamical quantum chaos then reduces to a static property of a process: given two non-deterministic projections on some small subsystem, how do the leftover states compare? $\ref{c1}$ states that for a chaotic process, butterflies need to have a large effect as in Eq.~\eqref{eq:orthogg}. 

We now ask what property of the many-time state $\ket{\ups}$ leads to the behavior Eq.~\eqref{eq:orthogg}? We summarize in Fig.~\ref{fig:zeta_cji} the butterfly flutter fidelity \eqref{eq:state_chaos} in the equivalent Choi and operator representations. We have done the conceptual heavy lifting in the setup of this problem, and so the following result is rather direct.
\begin{restatable}{prop}{BRent} \label{thm:BR_ent}
    If for any two orthogonal butterflies, one obtains (approximately) orthogonal final states on $\mc{H}_R$ if and only if $\ket{\ups}$ is (approximately) maximally entangled across the bipartition $B:R$. 
\end{restatable}
Proof for this and all further results in this Section can be found in Appendix~\ref{ap:proofs1}.

We note that the previous signature of Dynamical entropy turns out to be exactly the scaling of the entanglement of $\ket{\ups}$ in $B:R$ with times $k$, and the fidelity $   |\braket{\ups_{R | \vec{x}} |\ups_{R | \vec{y}}}|^2$ is a trotterized generalization of the Peres-Loschmidt Echo. We show this later in Section~\ref{sec:prev_chaos} with a detailed exposition on the relation between our conditions \ref{c1}-\ref{c3} and previous signatures (see also Fig.~\ref{fig:internal-external}). Proposition~\ref{thm:BR_ent} then gives a novel connection between these two previously well-studied metrics of chaos.

\subsection{Scrambling as Spatiotemporal Entanglement (C2)}
The condition in the previous section cannot be a complete notion of quantum chaos. In fact, most systems will look ``chaotic'' according to the prescription $\ref{c1}$. For example, a circuit dynamics consisting solely SWAP gates, without any interactions, leads exactly to Eq.~\eqref{eq:orthogg} being satisfied. In this case, the `orthogonality' of the butterflies is transferred to some large environment, and a new pure state is accessed on the butterfly space with each step. The orthogonalization resulting from a butterfly flutter resides entirely in a small subspace of $\mc{H}_R$, yet could be misconstrued as a strong global effect. We look at this example in more detail in Appendix~\ref{sec:LB}, and name such dynamics as a Lindblad-Bernoulli Shift.\textsuperscript{\footnote{This is in reference to its introduction in Ref.~\cite{Lindblad1986chaos}, and the chaotic classical analogy of the Bernoulli shift.}} As a further example, it can be shown analytically that Free fermions lead to a (maximal) linearly growing Dynamical entropy of a process~\cite{Cotler2018}, which by Proposition~\ref{thm:BR_ent} means Eq.~\eqref{eq:orthogg} is also true.

We therefore now introduce a notion of scrambling to the entropic measure from the previous section. Instead of just specifying that the entanglement in the splitting $B:R$ of the purified process $\ket{\ups}$ is volume-law with increase $k$, we extend this to incorporate that the butterfly flutter's effect spreads non-locally. We do this by including a subpart of the the butterfly space together with a subpart of the final pure state when looking at an entanglement bipartition of the process. In particular, $\ref{c2}$ means that
\begin{equation} \label{eq:vol_ST}
    S(\ups_{B_1 R_1} ) \propto \log(d_{B_1} d_{R_1}),
\end{equation}
where $S$ here indicates von Neumann entropy, $R_1$ and $R_2$ are a bipartition of the final state, $\mc{H}_R =: \mc{H}_{R_1} \otimes \mc{H}_{R_2}$, and similarly $\mc{H}_B =: \mc{H}_{B_1} \otimes \mc{H}_{B_2}$. We will generally consider bipartitions $R_1:R_2$ such that $d_{B R_1} < d_{R_2}$. Eq.~\eqref{eq:vol_ST} means that the entanglement of $\ket{\ups}$ in the arbitrary splitting $B_1 R_1: B_2 R_2$ needs to be volume-law. Often we choose $S(\ups_{B R_1} ) $ to investigate the spatial scrambling of the interventions from the entire space $\mc{H}_B$. In this case, spatiotemporal entanglement of the process $\ket{\ups}$ serves as a multitime generalization of `strong scrambling' in terms of the tripartite mutual information~\cite{Hosur2016}; see Section~\ref{sec:TMI}.

There are two subtle considerations to take into account here. For one, explicitly defining a `volume-law' compared to an `area-law' entanglement scaling requires specifying the underlying geometry. For the $B:R$ entanglement of \ref{c1}, there is a natural one dimensional scaling through increasing the number of times $k$ on which $\mc{H}_B$ is defined (and suitably redefining $\ket{\ups}$ in each case). When discussing spatiotemporal entanglement, we require a (varying) bipartition of the the spatial part of the process on $\mc{H}_R$, as well as the temporal part $\mc{H}_B$. For \ref{c2} we therefore restrict ourselves to systems of one spatial dimension, but note that one could likely generalize these results to higher dimensions. In addition, if the dynamics are chosen to be local, as is often the physically relevant situation, the space $\mc{H}_{R_1}$ should be chosen to be causally connected to the space $B$, i.e. well within the Lieb Robinson `lightcone' of the past space $\mc{H}_B$~\cite{lieb_finite_1972,chen2023speed}. This ensures that the operations on $\mc{H}_B$ may be possibly correlated with $\mc{H}_{R_1}$. This is immediately clear in circuit models of dynamics where the light-cone is exact~\cite{PhysRevX.8.021013,PhysRevX.7.031016}. Beyond this, it would be interesting to investigate this further with precise Lieb-Robinson bounds, along the lines of Ref.~\cite{Swingle2016LR}.

 Eq.~\eqref{eq:vol_ST} then means that rather than the butterfly flutter only affecting some localized part of the final pure state on $R$ (as in the Lindblad-Bernoulli Shift, App.~\ref{sec:LB}), leading to a high entanglement in the splitting $B:R$, Eq.~\eqref{eq:vol_ST} means that its effect spreads globally. This is what we call volume-law spatiotemporal entanglement. We will further argue that this is the essence of quantum chaos: that large effects from small, past operations correspond exactly to an extensive entropy scaling with increasing size of $R_1$, for a given $B$ (and possible bipartition $B_1$).

An example tensor network for computing this quantity is given in Fig.~\ref{fig:vol_law}, for a one dimensional lattice system. Here, $\lambda_i$ represents the Schmidt coefficient across the bipartition $B:R$, while the (yellow) circles represent bonds within a MPS representation of the final spatial pure state on $R$. A volume-law spatiotemporal entanglement then means a maximum bond dimension across any of these yellow circle bonds within $R$, when the $B$ subsystem can be connected to any of the components of $R$ within this tensor network (represented by grayed-out bonds). 

\begin{figure} 
  \includegraphics[width=0.22\textwidth]{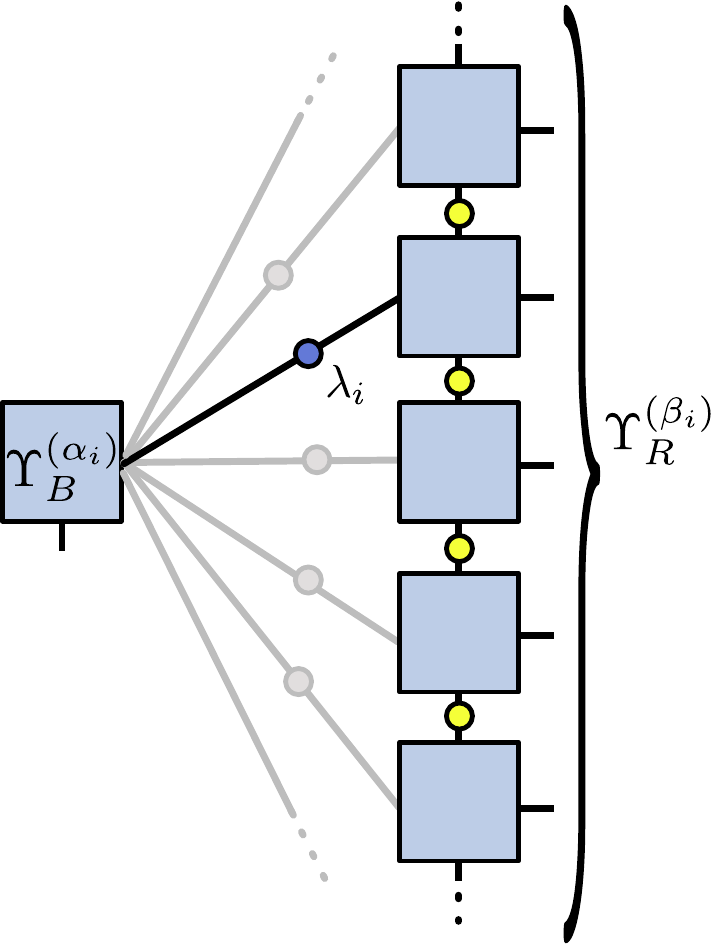}
  \caption{\footnotesize{A spatiotemporal tensor network of the process representing the butterfly protocol. Condition \ref{c2} states that a chaotic process will have maximal Schmidt rank across a decomposition across any cut (represented by colored circles) - i.e. that this network has a maximally bond dimension. These cuts are restricted to be only within the lightcone of the butterfly flutter in the final state on $\mc{H}_R$. The grayed-out lines represent other possible bonds in the choice of the space $\mc{H}_{R_1}$ in Eq.~\eqref{eq:state_chaos}, which should all pertain to a maximal bond dimension (volume-law) tensor network. The tensors (blue squares) on the right hand side have a local dimension of $d_{   R_1}$, and additionally should have maximal bond dimension between them for a chaotic process. }}\label{fig:vol_law}
\end{figure}

But is this spatiotemporal entanglement detectable? If one measures the final state fidelity as in Eq.~\eqref{eq:orthogg}, such that the butterfly flutters $\ket{\vec{x}}$ and $\ket{\vec{y}}$ include part of the final state, $R_1$, we will see that these butterfly flutters typically distinguish between area-law and volume-law spatiotemporal entanglement in $\ket{\ups}$. We will do this by choosing random and unitary butterfly flutters on $\mc{H}_B \otimes  \mc{H}_{R_1}$, and determining the fidelity \eqref{eq:state_chaos} for these. 
% We give an explicit construction of this in Appendix~\ref{ap:basis}.

In the following, $\mathbb{P}_{a \sim \mu}$ and $\mathbb{E}_{a \sim \mu}$ will mean, respectively, the probability and expectation value of the sampling of a random variable $a$ over the measure $\mu$. $\mu=\mathbb{H}$ denotes the Haar measure - the unique, unitarily invariant measure on Hilbert space. More details on randomness in Hilbert space are given in Section~\ref{sec:typ_process}. 

\begin{restatable}{thm}{randButt} \label{thm:ran_butterfly}
    \textbf{(Random Butterflies are Likely to Detect Spatiotemporal Entanglement)} For a Haar random choice of orthogonal butterflies $\mc{X}=\{\ket{\vec{x}},\ket{\vec{y}} \}$ across the combined space $\mc{H}_B \otimes  \mc{H}_{R_1}$ for any choice of space $\mc{H}_{R_1}$ the fidelity of the final state is likely to be sensitive to the volume-law property of $\ket{\ups}$. In particular, for $\delta >0$, 
        \begin{align} 
            \mathbb{P}_{\mc{X} \sim {\mathbb{H}}}\left\{  |\braket{\ups_{R_2 | \vec{x}} |\ups_{R_2 | \vec{y}}}|^2 \geq {\delta} \right\} &
            % \leq \frac{d_{B    {R_1}}^2\left(\pur -\frac{1}{d_{B    {R_1}}} \right) }{\delta(d_{B    {R_1}}^2 -1)} \nn \\
            \lesssim  \frac{\pur -1/(d_{B    {R_1}}) }{\delta},\label{eq:thm1a}
        \end{align}
    where $d_{B    {R_1}} = d_{B}     d_{R_1}$. This inequality is slightly approximated for large $d_{B    {R_1}}^2$, such that $d_{B    {R_1}}^2 \pm 1 \approx d_{B    {R_1}}^2$.
\end{restatable}
A proof for this can be found in Appendix~\ref{ap:proofs1}. This result is also valid for sampling a random butterfly from a $2-$design rather than fully Haar random, which can be done efficiently in practice.  

This constitutes a concrete connection between the fidelity between final states in Eq.~\eqref{eq:orthogg}, and spatiotemporal entanglement of the pure process $\ket{\ups}$. The key point is that for volume-law entanglement of the process $\ket{\ups}$, the purity of the reduced state on $\mc{H}_B \otimes  \mc{H}_{R_1}$ is inversely proportional to the size of the subsystem,
\begin{equation}
    \tr[\ups_{B    {R_1}}^2 ] \sim \mc{O}(\frac{1}{d_{B    {R_1}}}).
\end{equation}
For a Choi state which is truly volume-law -- rather than just maximally entangled across some specific splitting -- this is the case for \emph{any} choice of $\mc{H}_{R_1}$ up to causality considerations. So for volume-law, the right hand side of Eq.~\eqref{eq:thm1a} is close to zero for almost any small $\delta > 0$. Therefore, for most random unitary butterflies, ${\zeta}$ in Eq.~\eqref{eq:state_chaos} will likely be small for volume-law processes. 

% On the other hand, for area-law process Choi states the purity will not scale with subsystem size, as the rank is bounded by some constant $D$. Part B of the theorem certifies this. Eq.~\eqref{eq:thm1b} says that for a small bond dimension $D$, but a relatively large butterfly dimension $\dim(\mc{H}_{BR}) \gg d_{B    {R_1}} \gg D$, almost certainly $\tilde{\zeta}$ is relatively large compared to the inverse of the butterfly size $1/d_{B    {R_1}}$.

Note that framing chaos in terms of the entanglement properties of $\ket{\ups}$ is independent of the instrument, i.e., the butterfly flutter represented by $\ket{\vec{x}}$. This allows for testing of this principle against any previous or new heuristic of quantum chaos. It also implies that the manifestation of quantum chaos may be tested for strong or weak butterflies, and many-time or few-time, which turns out to be the distinction between the Peres-Loschmidt Echo and Local-Operator Entanglement, as we show in Section~\ref{sec:prev_chaos}.

One might now want to know if volume-law spatiotemporally entangled processes exist; if the condition \ref{c2} is too strong. In fact, from concentration of measure results, it is known that most processes generated from Haar random dynamics are locally exponentially close to the completely noisy process~\cite{FigueroaRomero_Modi_Pollock_2019},
\begin{equation} \label{eq:typ_process}
    \tr_{R}[\ket{\ups^{(\mathbb{H})}}\bra{\ups^{(\mathbb{H})}}] \approx \frac{\id}{d_B}, \text{ for } d_B \ll d_R
\end{equation}
and polynomially close for dynamics sampled from an $\epsilon-$approximate $t-$design~\cite{FigueroaRomero_Pollock_Modi_2021}. Such a process also has volume-law spatiotemporal entanglement, as we prove in Section~\ref{sec:typ_process}.

We now move to our final condition on quantum chaos.

\subsection{Complexity of Sensitivity to Perturbation (C3)}

We now introduce a final, strictly stronger measure of chaos, based on the notion of how far the final states are, after two distinct butterfly flutters. This is not just a fidelity measure like we have so far considered, but rather the fidelity after a restricted correction. 
\begin{mydef} \label{def:BFF}
    The {Butterfly Flutter Fidelity} takes values between $0 \leq \zeta \leq 1 $, and is defined as 
\begin{align}
    \zeta(\ups)  :=&  \underset{{V \in \mc{R}, \braket{\vec{x}|\vec{y}}=0}}{\mathrm{sup}}  \left(| \langle \ups_{R|\vec{x}} | V | \ups_{R|\vec{y}} \rangle|^2 \right) \label{eq:state_chaos} 
    % \\
    % =&   \underset{{{   C},V, \braket{\vec{x}|\vec{y}}=0 }}{\mathrm{sup}}\left(\frac{|\langle \vec{y} | \tr_{   C}(\ups_{B   C} V) |\vec{x} \rangle |^2}{{\langle \vec{x} |\ups_B|\vec{x} \rangle \langle \vec{y} |\ups_B|y \rangle} } \right)\nn .
\end{align}
Here, $V$ is a unitary operation on the full (spatial) Hilbert space $\mc{H}_R$, and is restricted to some low-complexity set $\mc{R}\subset \mathbb{U}(d_R)$.
\end{mydef}

Note that often in Eq.~\eqref{eq:state_chaos} we will instead choose a particular pair of butterfly flutters $\ket{\vec{x}}$ and $\ket{\vec{y}}$, or otherwise average over some set of them. This is order to perform analytic calculations or to draw comparisons with other quantities, and the interpretation of a sensitivity to perturbation holds true without an optimization over all possible butterfly flutters satisfying $\braket{\vec{x}|\vec{y}}=0$.

Intuitively, the Butterfly Flutter Fidelity measures how difficult it is to convert one resultant state $\ups_{R|\vec{x}}$ to the other $\ups_{R|\vec{y}}$. In other words, it measures how easily correctable the effect of a past butterfly flutter is. We leave open the exact measure of the complexity with which the `correction' unitary $V$ is restricted. Possible choices include specifying $V$ to be: a constant depth local circuit, independent of the system size or time evolution in the process; a local circuit with depth proportional to the size of the system $d_R$ but independent of the time of evolution; a unitary with an appropriately defined restricted Nielsen Complexity~\cite{Nielsen2006}; or an MPO of restricted (constant) bond dimension. Of course, many of these measures are related. It would be an interesting avenue of future research to investigate this quantity in more detail and for different models. For the rest of this work, we will generally take $V$ such that it can be represented by an MPO with a restricted bond dimension, part of the set $\mc{R}_{\text{MPO}}$. Therefore, a process will be chaotic according to \ref{c3} if it is not possible to efficiently correct the effects of a past butterfly flutter. We note that the Butterfly Flutter Fidelity reduces to simply the fidelity, as in \ref{c1} and Eq.~\eqref{eq:orthogg}, when $V$ is restricted to the identity $\mc{R}=\{ \id \}$.

We will now show that this is a strictly stronger condition than volume-law spatiotemporal entanglement, that $\ref{c3} \implies \ref{c2} $. 
\begin{restatable}{prop}{zetaST} \label{thm:zeta-ST}
    If the Butterfly Flutter Fidelity \eqref{eq:state_chaos} is small, $\zeta \approx 0$, then the process $\ket{\ups}$ has volume-law spatiotemporal entanglement. 
\end{restatable}
A proof for this can be found in Appendix~\ref{ap:proofs1}.
The question remains of just how strong the condition \ref{c3} is. That is, when is there volume-law spatiotemporal entanglement in a process \ref{c2}, but the effects of a butterfly flutter are easily correctable? In fact, the only case where \ref{c2} and \ref{c3} are not equivalent is if the process has a volume-law initial state. 

\begin{restatable}{prop}{stateChaos} \label{thm:state_chaos}
     If the Butterfly Flutter Fidelity \eqref{eq:state_chaos} is not small (non-chaotic), $\zeta \approx 1$, but the process $\ket{\ups}$ has volume-law spatiotemporal entanglement, the process can be written as a simple dynamics with a volume-law entangled initial state.
\end{restatable}
A proof for this can be found in Appendix~\ref{ap:proofs1}.
What this result means is that in the particular case where a process is regular according to \ref{c3}, but chaotic according to \ref{c2}, then all the volume-law entanglement is attributed to the initial state. The dynamics part of the process can be considered to have area-law entanglement.

In the setup we have suggested to classify chaos in quantum systems, one interacts locally with a quantum system across multiple times and examines the effect on the final, global pure state. In this situation, the above result Prop.~\ref{thm:state_chaos} means that in terms of the entanglement properties of the corresponding process $\ket{\ups}$, one cannot distinguish between a process that first prepares a volume-law spatial entanglement state from a process that genuinely creates volume-law spatiotemporal entanglement from the dynamics. One way to interpret this is that complex spatial entanglement in itself is chaotic. We refer to this as quantum \emph{state chaos}: for a volume-law entangled state, performing an operation on a small part of a large state instantaneously has a highly non-local and strong effect on the remainder of the state. 

This also follows from the fact a multipartite quantum state is also a quantum channel, through teleportation. This is a purely quantum effect, and thus there is no classical analog to quantum state chaos. Volume-law spatiotemporal entanglement is equivalent to chaos in the sense of a strong, non-local sensitivity to perturbations, regardless of whether these perturbations occur simultaneously to the effect (state chaos), or in the past with the effect stemming from dynamics (as measured by the Butterfly Flutter Fidelity \ref{c3}). However, in the traditional dynamical sense, the butterfly flutter fidelity measures the chaoticity of the dynamics and so can be seen as equivalent to the quantum butterfly effect: the operationally meaningful notion of quantum chaos.

\subsection{Sensitivity to initial perturbations}
The above operational understanding for quantum chaos readily resolves a fundamental question. Namely, are quantum chaotic systems sensitive to an initial perturbation? 

The usual argument against a quantum sensitivity to perturbation is that the distance (or fidelity) between two initial states, $\epsilon = |\braket{\psi |\phi}|$, is preserved with unitary time evolution
\begin{equation} \label{eq:ignorant_butterfly}
    |\braket{\psi | U_t^\dg U_t |\phi}| = |\braket{\psi|\phi}| = \epsilon.
\end{equation}
This precludes a straightforward notion of exponential (or otherwise) deviation with respect to $\epsilon$.

\ref{c3} includes a rather direct and intuitive notion of sensitivity to initial conditions. Consider a single-time butterfly flutter protocol, with a perturbative operations, $X$ and $Y$ on initial state $\ket{\psi}$. Eq.~\eqref{eq:state_chaos} then reduces to a sensitivity of the resultant state after this initial operation
\begin{equation}
     \zeta(\ups) =\underset{{V\in \mc{R}}}{\mathrm{sup}}  \left(| \langle \psi | X^\dg U^\dg_t V U_t Y |\psi  \rangle|^2 \right).
\end{equation}
Here we have assumed the local perturbations $X\equiv (X_S \otimes \id_E)$ are unitary for simplicity, and so the final states are normalized. Instead of comparing the final state fidelity given an initially perturbed state as in Eq.~\eqref{eq:ignorant_butterfly}, the single-time flutter corresponds to how difficult it is to correct the resultant state from a local perturbation. This notion of difficulty encompasses the complexity inherent to quantum mechanics but admits the classical analog of sensitivity to perturbation.

\subsection{Determining The Butterfly Flutter Protocol in the Laboratory} \label{sec:ancilla}
The quantum butterfly protocol Eq.~\eqref{eq:state_chaos} is a fidelity of final pure states, which apparently requires a backwards-in-time global evolution to compute. In this section we show that by appending a quantum ancilla space to the protocol, one can compute $\zeta$ through only forward-in-time evolution. 

Consider the same setup as the butterfly flutter protocol in Eq.~\eqref{eq:state_chaos}, with appended qubit ancilla space $\mc{H}_A$, with combined initial state
\begin{equation}
    \ket{\psi_{SEA}} = \ket{\psi}_{SE} \otimes \ket{+}_A.
\end{equation}
Then, for a butterfly flutter defined by the unitary instruments $A_{x_1}, A_{x_2}, \dots, A_{x_k}$ acting on $\mc{H}_S$, define an instrument at time $t_i$ on the full $SEA $ space as 
\begin{equation}
    A^\prime_i :=  \id_E \otimes \left(A_{x_i} \otimes \ket{0}\bra{0} +A_{y_i} \otimes \ket{1}\bra{1} \right).
\end{equation}
We also define an additional (controlled) instrument which encodes the correction unitary $V$,
\begin{equation}
    V^\prime:= \id_{SE} \otimes \ket{0}\bra{0} + V \otimes  \ket{1}\bra{1}.
\end{equation}
Then the final state of the reduced state of the ancilla qubit at the end of the forward-in-time evolution of the butterfly protocol is 
\begin{equation} \label{eq:rhoPrime}
    \rho= \tr_{SE}[\mc{V}^\prime \mc{U}_{k} \mc{A}^\prime_k  \dots \mc{U}_1 \mc{A}^\prime_1  (\ket{\psi_{SEA}} \bra{\psi_{SEA}})]
\end{equation}
where as is standard throughout this work, calligraphic script letters correspond to superoperators, $\mc{A}_i^\prime (\cdot )\equiv A^\prime_i(\cdot )A_i^{\prime \dg}$, $\mc{V}^\prime (\cdot )\equiv V^\prime(\cdot )V^{\prime \dg}$, and $\mc{U}_i^{\prime } (\cdot )\equiv U^\prime_i(\cdot )U^{\prime \dg}_i$ is global unitary $\mc{H}_S \otimes \mc{H}_E$ evolution. This protocol is presented in Fig.~\ref{fig:ancilla}. Then it is easy to check that the squared off-diagonal elements of the final state of the ancilla qubit gives exactly the Butterfly Flutter Fidelity \eqref{eq:state_chaos}, for a given choice of correction unitary $V$. Further, the off-diagonal elements of a density matrix are easily measurable, 
\begin{equation}
    \zeta = |\bra{0} \rho \ket{1}|^2 = |\frac{1}{2}\braket{\sigma_x + i \sigma_y }|^2.
\end{equation}
This can be directly generalized to larger ancilla spaces, if one wanted try a set of different butterfly flutters.

\begin{figure}[t]
  \includegraphics[width=0.47\textwidth]{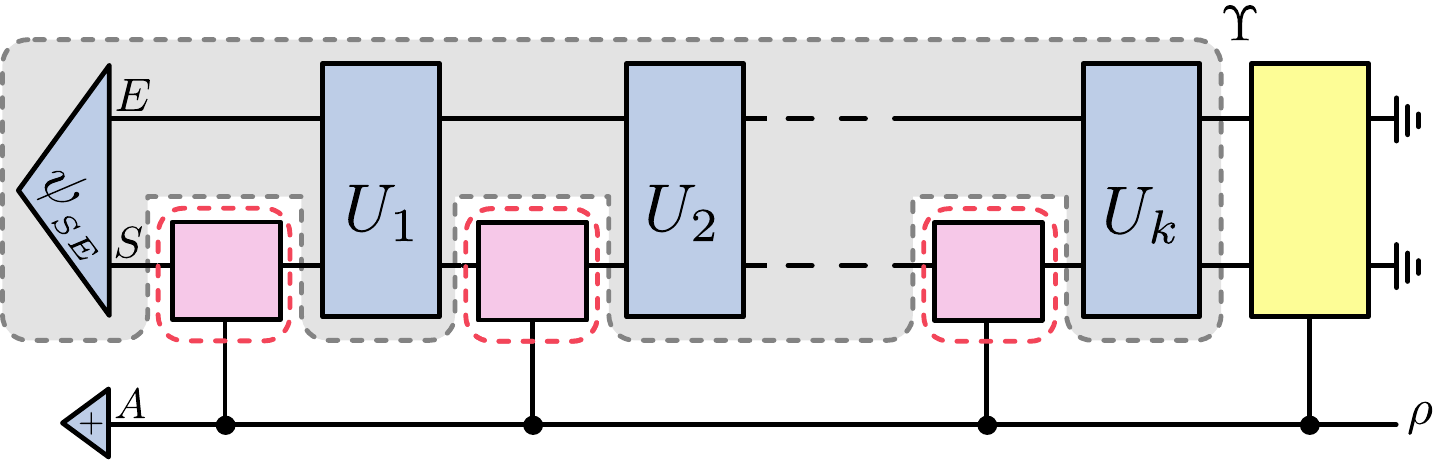}
  \caption{\footnotesize{The forward-in-time protocol for measuring the Butterfly Flutter Fidelity. Here the pink controlled operations correspond to the butterfly flutters $\ket{\vec{x}}$ for ancilla qubit equal to $\ket{0}\bra{0}$, and $\ket{\vec{y}}$ for ancilla qubit equal to $\ket{1}\bra{1}$. The (yellow) global operation at the end (rightmost) is then similarly controlled to be either the identity map, or the correction unitary $V$ to be optimized over (see discussion around Eq~\eqref{eq:state_chaos}). The Butterfly Flutter Fidelity is then stored in the coherences of the final state of the ancilla, $\rho$, while information of decoherence effects is encoded in the diagonals.}} 
  \label{fig:ancilla}
\end{figure}

We will finish this section by making some remarks about the limitations of a realistic experimental setup. In a non-isolated situation where the initial state may be mixed and where evolution may not be unitary, then Eq.~\eqref{eq:rhoPrime} takes the form
\begin{equation} \label{eq:rhoPrime1}
    \rho = \tr_{SE}[\mc{V}^\prime \mc{A}^\prime_k \mc{L}_{k-1} \dots \mc{A}^\prime_1 \mc{L}_1 (\rho_i)].
\end{equation}
Here, $\rho_i$ is some arbitrary initial state, and $\mc{L}_i$ represents open quantum system evolution which may include arbitrary decoherence effects ($\mc{L}_i$ is generally a CPTP map). In practice, Eqs.~\eqref{eq:rhoPrime} and \eqref{eq:rhoPrime1} require an identical protocol from a (hypothetical) experimenter. In practice, one cannot easily tell whether the scaling of the butterfly flutter fidelity is according to Eq.~\eqref{eq:rhoPrime1} or the perfectly isolated Eq.~\eqref{eq:rhoPrime}. This is a problem faced with other measures of the quantum case (such as the OTOC), and indeed even classically it is difficult to discern between noise and dynamical chaos. 

In the BFF protocol, one can check the unitarity of the dynamics by checking the purity of the final total state. This requires access to two copies of the final state to perform a swap test. One alternative to this setup is to perform the butterfly flutter protocol with the correction unitary coming from the set of all possible unitaries, $\mc{R} = \mathbb{U}(d_R)$ in Eq.~\eqref{eq:state_chaos}. This means that the correction unitary $V$ will align the resultant states $\ups_{R|\vec{x}}$ and $\ups_{R|\vec{y}}$ to give $\zeta = 1$ if the dynamics are unitary (following Eq.~\eqref{eq:rhoPrime}). If the dynamics were not unitary (Eq.~\eqref{eq:rhoPrime1}), this should not be possible and so in this case we have $\zeta < 1 $. Of course, it is highly expensive and non-trivial to implement an optimization over all possible unitaries in Eq.~\eqref{eq:state_chaos}.

This protocol allows one to perform a forward-in-time experiment to determine the butterfly flutter fidelity. This requires a perfect control over the system-environment space (the $R$ space), in order to implement the correction unitary $V$, and a perfectly isolated ancilla space which is not itself influenced by decoherence effects or other uncontrolled dynamics. However, the correction unitary itself is in principle easy to implement by construction. The predominant difficulty is how exactly to perform the maximization over $V \in \mc{R}$ in Eq.~\eqref{eq:state_chaos}. It would be interesting to determine an efficient algorithm that could approximate this optimization.

\subsection{Summary and Discussion}
We now restate the hierarchy of conditions for quantum chaos. For two butterfly flutters, with Choi states $\ket{\vec{x}}$ and $\ket{\vec{y}}$, we call a process $\ket{\ups}$ chaotic if: 
\begin{enumerate}[label={\textbf{{(C\arabic*)}}}]
    \item \label{c1}\textit{(Perturbation orthogonalizes future state)} The final state on $R$ should be strongly sensitive to butterflies on $B$: 
    \begin{equation}
        |\braket{\ups_{R|\vec{x}}  |\ups_{R|\vec{y}}}|^2 \approx 0,
    \end{equation}
    or equivalently (Prop.~\ref{thm:BR_ent})
    \begin{equation}
        S(\ups_B) \sim \log(d_B).
    \end{equation}.
    \item \label{c2}\textit{(Scrambling as volume-law entanglement)} Butterflies on $B$ should affect a large portion of the final state on $R$: 
    \begin{equation}
        S(\ups_{B_1 R_1}) \sim \log(d_{B_1} d_{R_1}),
    \end{equation}
    for appropriate choices of $\mc{H}_{R_1} \subset \mc{H}_R $ and $\mc{H}_{B_1} \subseteq \mc{H}_B $.
    \item \label{c3}\textit{(Complexity of sensitivity)} Different butterflies on $B$ should lead to different enough states on $R$, as measured by the Butterfly Flutter Fidelity
    \begin{equation}
        \zeta(\ups) = \underset{{V\in \mc{R},\braket{\vec{x}|\vec{y}}=0}}{\mathrm{sup}} |\bra{\ups_{R|\vec{x}}} V  \ket{\ups_{R|\vec{y}}}|^2 \approx 0
    \end{equation}
    for some defined set of bounded-complexity unitaries $\mc{R}$.
\end{enumerate}
The operational criteria for quantum chaos impose several restrictions on the spatiotemporal correlation content of a process. \ref{c1} and \ref{c2} require that $\ups$ is volume entangled, while \ref{c3} further requires that the process itself must be able to dynamically generate volume-law spatiotemporal entanglement. Importantly, these criteria directly led to a universal operational metric for quantum chaos in Eq.~\eqref{eq:state_chaos}, which we showed to be accessible in a laboratory setting.

We then used these ideas, especially \ref{c3}, to show how quantum processes are also sensitive to initial conditions much like their classical counterparts. This opens up the possibility of operationally defining quantum Lyapunov exponents to further close the gap between the theories of classical and quantum chaos. Finally, \ref{c3} has the same flavor as the \textit{complexity=volume} conjecture due to Susskind~\cite{susskind}, also see~\cite{bouland} for the more definite version of the same conjecture. Namely, the operational metric for quantum chaos is concerned with the complexity of the correction unitary in Eq.~\eqref{eq:state_chaos}. Our results therefore hint that quantum chaos may be key to understanding this conjecture, fitting with the common belief that black holes are maximally chaotic quantum systems~\cite{Sekino_Susskind_2008,Shenker2014}. On the other hand, the tools presented in Ref.~\cite{bouland} are likely applicable to the case of quantum chaos.

We show in Section~\ref{sec:loe} that the previous dynamical signature of the Local-Operator Entanglement measures this single-time sensitivity, optimizing over any initial state. Further, it can be shown that out-of-time-order correlators generically probe this operator entanglement~\cite{dowling2023scrambling}. The hierarchy \ref{c1}-\ref{c3} gives a robust understanding of why these previous diagnostics measure chaos, in terms of a future sensitivity to past local operations.

\section{Connection to Previous Signatures}\label{sec:prev_chaos}
Our construction so far has involved a first-principles proposition of a series of conditions that mean chaos as a sensitivity to perturbation in quantum systems. We will now show how these conditions \ref{c1}-\ref{c3} compare to previous dynamical signatures of chaos (see diagram of this connection in Fig.~\ref{fig:internal-external}). The Peres-Loschmidt Echo corresponds to \ref{c1} in the many-time limit and for weak butterflies, while Dynamical Entropy is exactly the entanglement scaling of $\ket{\ups}$ in the splitting $B:R$, and so is in some sense equivalent to the Peres-Loschmidt Echo scaling according to Proposition~\ref{thm:BR_ent}. The tripartite mutual information measures spatiotemporal entanglement for a single-time butterfly, and so \ref{c2} can be seen as a multitime generalization of this measure. Finally, the local-operator space entanglement measures the required entanglement complexity of the correction unitary $V$, such that $\zeta(\ups) = 1$ for any initial state. See Fig.~\ref{fig:flow} for a summary of these connections. In this section we will explain these diagnostics and show each of these connections in turn. Our first principles construction is supported by, and contains a range of previous notions of quantum chaos from recent years, all within a single intuitive framework.

\subsection{Peres-Loschmidt Echo} \label{sec:LE_comb}
The Peres-Loschmidt Echo measures the sensitivity of an isolated quantum system to a weak perturbation to the dynamics~\cite{Peres1984stab}.\textsuperscript{\footnote{This quantity is alternatively called fidelity decay or Loschmidt Echo in the literature. We take the middle ground here, to give credit to the historical role of Peres~\cite{Peres1984stab}, while remaining familiar to those who recognize this quantity as the latter.}} It is equal to the deviation in fidelity between the same initial states evolving unitarily according to some Hamiltonian compared to a perturbed Hamiltonian,
\begin{equation}
    |\braket{\psi_t | \psi^\epsilon_t}|^2=|\braket{\psi| e^{iHt} e^{-it(H+\epsilon T)} |\psi}|^2.
\end{equation}
This equivalently measures the distance from the initial state, when a state evolves forward in time, then evolves backwards in time according to imperfect evolution. Exponential decay with time is regarded heuristically to mean quantum chaos. In practice~\cite{Emerson2002,Poulin2004}, one often needs to discretise the dynamics in order to realize the perturbation to the Hamiltonian, $T$. To do so, one can use the  Trotter approximation of the perturbed evolution,
\begin{align}
    e^{-it(H+\epsilon T)} &\approx (e^{iH \delta t}e^{i \epsilon T \delta t})^k, \label{eq:trotter}\\
    &=:(U_{\delta t} W_\epsilon  )^k \nn
\end{align}
where $k \delta t = t$, which is valid for large $k$ and small $\delta t$. Then, up to  Trotter error~\cite{Childs2021}, the Peres-Loschmidt Echo corresponds to the fidelity between two final states, given the application of $k$ identity channels, compared to $k$ unitaries which are $\epsilon-$close to the identity (see Fig.~\ref{fig:combs_probes} a)). From this, we can already see that the Peres-Loschmidt Echo falls into the category of a fidelity between resultant states given two past butterfly flutters as in \ref{c1}, Eq.~\eqref{eq:orthogg}. 

In addition to the trotterization, the key difference between our condition \ref{c1} and the Peres-Loschmidt Echo is that instead of optimal butterflies, we specify the two many-time butterfly flutters to be projections which are $(k\epsilon)-$close. These two projections are, respectively, the Choi states of a sequence of $k$ weak unitaries and a sequence of $k$ identity maps, such that
\begin{align}
    |\braket{\vec{x}|\vec{y}}|:=& |\langle W_{\epsilon}^{\otimes k} | \id^{\otimes k} \rangle| \nn \\
    =& |\langle W_{\epsilon} | \id \rangle|^k  \label{eq:le_decay}\\
    =& (1-\epsilon)^k d_S^{2k}. \nn
\end{align}
where we recall that $d_B=d^{2k}$. Then for a typical volume-law process, consisting of random dynamics as described around Eq.~\eqref{eq:typ_process} and further explored in Section~\ref{sec:typ_process}, under the action of any two butterflies of appropriate size we have that $\ups_B \sim \id/d_B $, and so typically
\begin{equation}
    \begin{split} \label{eq:LE_exp}    |\braket{\ups_{R|\vec{x}}|\ups_{R|\vec{y}}}|^2 &\approx  \frac{|\langle \vec{x}|\id/d_B|\vec{y} \rangle |^2}{{\langle \vec{x}|\id/d_B |\vec{x}\rangle \langle \vec{y} | \id/d_B |\vec{y} \rangle} } \\ 
        &= (1-\epsilon)^{2k} d_B^2 (1/d_B^2) \\
        &= (1-\epsilon)^{2k}  \\
        &\approx \ex^{- 2 k \epsilon } , \text{ for small $\epsilon$,}\\
        &\approx 0, \text{ for large $k$}.
    \end{split}
\end{equation}
In the first line we have used the Schmidt decomposition, as in Fig.~\ref{fig:zeta_cji} \textbf{b)} and Eq.~\eqref{eq:schmidt1}. 
% Note that the normalization in the denominator of the general $\zeta$ in Eq.~\eqref{eq:state_chaos} is equal to one for unitary butterflies. This is ensured by the causality conditions of a process Eq.~\eqref{eq:causality_conditions}. 
For an area-law Choi state, this fidelity will be larger, and tend to scale as the leading-order Schmidt coefficient. 

\begin{figure}[t]
  \includegraphics[width=0.49\textwidth]{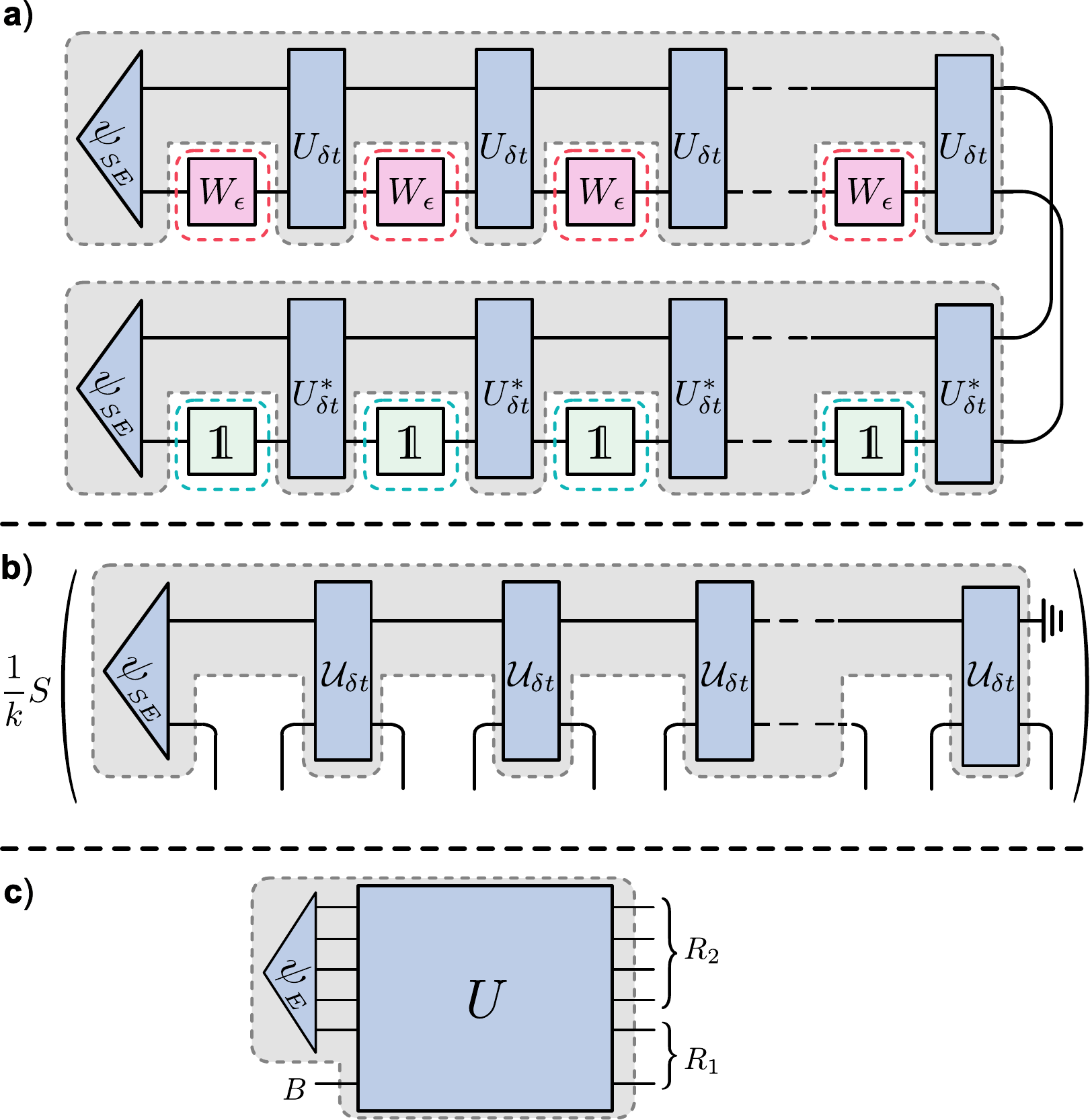}
  \caption{\footnotesize{The quantum process construction for the dynamical signatures: (a) Peres-Loschmidt Echo, (b) Dynamical Entropy, and (c) Tripartite Mutual Information. Note that the diagrams of $\mathbf{a})$ and $\mathbf{c})$ are in the pure state representation, such that the initial state is a state-vector (ket), and boxes correspond to matrices (Latin script), while $\mathbf{b})$ is in the superoperator representation, such that the initial state is a vectorized density matrix, and boxes represent quantum channels in the Liouville superoperator representation (calligraphic script). See Section~\ref{sec:processes} and Ref.~\cite{OperationalQDynamics} for further details.}}
  \label{fig:combs_probes}
\end{figure}

For a given  Trotter error, time evolution corresponds to increasing $k$, for a constant $\delta t$ and $\epsilon$. Therefore in Eq.~\eqref{eq:LE_exp} we can see how exponential decay with time stems from the property of entanglement structure of the Choi state $\ket{\ups}$. The choice of temporally local, weak unitaries is key to this exponential time decay with time.

We have shown that the Peres-Loschmidt echo can be characterized through weak, many-time butterfly flutters under the first condition \ref{c1}. It should be noted that this is the weakest condition which we argue is necessary for quantum chaos. In particular, the Peres-Loschmidt Echo has no extra ingredient of an correction unitary $V$ acting on the final states as in \ref{c3}. This distinction means that while the Peres-Loschmidt Echo probes a butterfly having a strong effect, it does not probe the delocalization of this effect; the scrambling. This will be become apparent in Appendix~\ref{sec:LB} where we investigate an example of a regular dynamics which is apparently chaotic according to the Peres-Loschmidt Echo. 

From Proposition~\ref{thm:BR_ent}, we saw that the Butterfly Flutter Fidelity for $V = \id$ is small if and only if the entanglement $S(\ups_B)$ is extensive. We will now see that the quantum Dynamical Entropy exactly measures this quantity asymptotically with number of perturbations $k$, given a novel connection to the Peres-Loschmidt Echo.

\subsection{Dynamical Entropy} \label{sec:Sdy}
The quantum Dynamical Entropy was originally introduced as the quantum generalization of the Kolmogorov-Sinai entropy, which quantifies the asymptotic gain of information when a classical system is repeatably measured~\cite{Lindblad1986chaos,Pechukas1982,Slomczynski1994-ns,Alicki1994-dc}. It measures the long-term unpredictability of a dynamics, with positivity indicating chaoticity in the classical case. Quantum mechanically, measurement necessarily perturbs a system, and comes with its own inherent unpredictability. One can account for the entropy due to a measuring device compared to the process itself~\cite{Slomczynski1994-ns}, but a more elegant solution is to define this quantity in a device independent way~\cite{Lindblad1979-tm,Lindblad1986chaos,Cotler2018}. Indeed, classically, Kolmogorov Sinai entropy is the entropy rate of a \emph{stochastic process}, so the natural language of the quantum version of this requires a description of quantum stochastic processes~\cite{Milz2020kolmogorovextension,milz2020quantum}: precisely the process tensor formalism detailed in Section~\ref{sec:processes}.

Formally, dynamical entropy is defined as the asymptotic gain in information when a additional (measurement) steps are added to a quantum process, 
\begin{equation} \label{eq:Sdy}
    S_{\mathrm{Dy}}(\Upsilon) := \underset{k\to \infty}{\lim}\frac{1}{k} S(\Upsilon_{B_k}),
\end{equation}
where $\Upsilon_{B_k} =\tr_R[\ups]$ is a marginal process on $k$ time steps, meaning a process with a given dynamics, measured every $\delta t$ seconds. We do not need to specify what measurement, as the process tensor encodes any possible measurement protocol, fulfilling precisely the role of a spatiotemporal density matrix. See Section~\ref{sec:processes} for details. For such an asymptotic quantity to be non-zero, this strictly requires an infinite dimensional environment. Poincare recurrence would render any finite isolated system to have finite total entropy in the asymptotic limit. As we consider unitary dynamics on an isolated, finite dimensional quantum system, we will not take the asymptotic limit precisely. Instead, we define the $k^\mathrm{th}$ Dynamical Entropy 
\begin{equation} \label{eq:Sdyk}
    S_{\mathrm{Dy}}^{(k)}(\Upsilon) := \frac{1}{k} S(\Upsilon_{B_k}),
\end{equation}
where $k$ is taken to be large, but small enough such that $d_B = d_S^{2k} \ll d_E \approx d_R$. The expression Eq.~\eqref{eq:Sdyk} is represented graphically in Fig.~\ref{fig:combs_probes} b). From this definition and Proposition~\ref{thm:BR_ent}, we can directly see that a non-zero $S_{\mathrm{Dy}}^{(k)}$ is sufficient for volume-law entanglement of $\ket{\ups}$. This is given that we scale the number of interventions $k$, recalling that the butterfly space scales as $k$ (see Eq.~\eqref{eq:Bdef}).
\begin{restatable}{prop}{Sdy} \label{thm:sdy}
    If the Dynamical Entropy is non-zero, then the process $\ket{\ups}$ is volume-law entangled in the splitting $B:R$ for all times.
\end{restatable}
This is apparent from definitions, and a proof in supplied in Appendix~\ref{ap:proofs2}.

This also approximately holds true if instead the $k-$dynamical entropy is considered. What is important is that the dynamical entropy generally exhibits distinct behavior for area- versus volume-law temporally entangled processes. This simple result shows how closely the construction of Dynamical Entropy agrees with first condition \ref{c1} derived in this work, despite arriving at it from a starkly different viewpoint -- that of the quantum version of the butterfly effect. 

For example, the $k-$dynamical entropy of a typical process, Eq.~\eqref{eq:typ_process}, is on average maximal,
\begin{align}
    S_{\mathrm{Dy}}^{(k)} (\ups^{(\mathbb{H})}) &= \frac{S(\frac{\id}{d_S^{2k}})}{k} \nn \\
    &= \frac{  \log ( d_S^{2k})}{k} = 2\log(d_S).
\end{align}
A more precise typicality bound can be found from Theorem~\ref{thm:rand_process_zeta}.

Moreover, one can see that summing over a full basis of butterfly flutters give a quantity proportional to this entanglement.
\begin{restatable}{prop}{pesin} \label{thm:q_pesin}
    Consider a full basis of local unitary butterfly flutters, $\mc{X} = \{ A_{w_1},A_{w_2},\dots,A_{w_k}\}^{d_B}$, where the number of operations at each time in the set is $\# w_i = d_S$ (see Appendix~\ref{ap:basis} for an example construction of this). Then the following relation holds,
    \begin{equation}
        S^{(2)}(\ups_B) = -\log \Big( \frac{1}{d_B^2} \big( \sum_{\vec{x}\neq \vec{y} \in \mc{X}}^{d_B^{2}-d_B}  |\bra{\vec{x}} \ups_B \ket{\vec{y}} |^2  - d_B \big) \Big) \label{eq:pesin}
    \end{equation} 
    where $S^{(2)}(\ups_B)$ is the quantum $2-$R\'enyi entropy. 
\end{restatable}
This is proved in Appendix~\ref{ap:proofs2}.
 There is a large body of literature arguing that under certain conditions both the Peres-Loschmidt Echo~\cite{Peres1984stab,Jalabert2001,Cucchietti2002,Emerson2002} and OTOCs~\cite{Maldacena_Shenker_Stanford_2016,Foini2019,Parker2019} decay exponentially across some time regimes for chaotic systems. Given the close ties between the Butterfly Flutter Fidelity and other metrics which we describe in this work, it is not unreasonable to speculate that the Butterfly Flutter Fidelity exhibits a similar behavior. Eq.~\eqref{eq:pesin} then forms a relation between Dynamical Entropy, and these conjectured quantum Lyapunov exponents. This is suggestive of a kind of \emph{Quantum Pesin's Theorem},\textsuperscript{\footnote{The classical Pesin's theorem states that the Kolmogorov-Sinai entropy is a lower bound of the sum of the positive Lyapunov exponents of a classical dynamical system~\cite{Pesin1977}.}} although more needs to be done to understand how and when the Butterfly Flutter Fidelity produces an exponential decay, and to refine the notion of quantum Dynamical Entropy. 

To our knowledge the exact connection of dynamical entropy to quantum chaos as a sensitivity to perturbation has not yet been explored in the literature; it has only been proposed as a generalization of the classical equivalent, Kolmogorov-Sinai entropy. Here we can essentially \emph{derive} dynamical entropy, starting from our principle \ref{c1}, and connecting it to the Peres-Loschmidt Echo and other notions of chaos as a sensitivity to perturbation.

Due to the classical equivalences between Lyapunov exponents and Kolmogorov Sinai entropy, one might be tempted to conflate quantum chaos with a non-zero dynamical entropy. This, however, only accounts for the weakest of the three conditions \ref{c1}-\ref{c3}. Equivalently, it only allows $V$ in the Butterfly Flutter Fidelity Eq.~\eqref{eq:state_chaos} to be strictly equal to the identity. As we have already discussed in Section~\ref{sec:mainresult}, this is an insufficient charaterization. For example, Free Fermion dynamics generally exhibit an extensive Dynamical Entropy~\cite{Cotler2018}, as does a dynamics consisting of SWAP gates, as we detail in Appendix~\ref{sec:LB} (both valid up to finite dimension constraints). We therefore move onto the more robust conditions of quantum chaos, based around spatiotemporal entanglement \ref{c2} and the Butterfly Flutter Fidelity \ref{c3}.

\subsection{Tripartite Mutual Information} \label{sec:TMI}
Here we will show that the spatiotemporal entanglement \ref{c2}, in the single time case, corresponds to the tripartite mutual information signature of chaos as introduced in Ref.~\cite{Hosur2016}, sometimes termed `strong scrambling'. 

The tripartite mutual information is a measure between a subsystem of the input to a quantum channel, and some bipartition of the output. Considering a single time butterfly flutter, as in Fig.~\ref{fig:combs_probes} c), in our language this corresponds to 
\begin{equation} \label{eq:tripartite}
    I_3(B:R_1:R_2):=I(B : R_1) + I(B : R_2) - I(B : R)
\end{equation}
recalling that $\mc{H}_R = \mc{H}_{R_1} \otimes \mc{H}_{R_2}$, and where $I(A:B)$ is the quantum mutual information, defined as
\begin{equation}
    I(A : B) := S(\rho_A) + S(\rho_B) - S(\rho_{AB}).
\end{equation}
Note that this single-time butterfly flutter protocol process corresponds exactly to the setup from Ref.~\cite{Hosur2016} when the initial state is separable across $S:E$. This is represented in Fig.~\ref{fig:combs_probes} c).
When the tripartite information \eqref{eq:tripartite} is near-minimal, it is argued that the channel is strongly scrambling. This quantity has been connected to an average of an infinite-temperature OTOC over a complete basis of operators~\cite{Hosur2016}; see also Proposition~\ref{thm:q_pesin} for a similar result. We can in fact show directly that volume-law spatiotemporal entanglement implies strong scrambling.
\begin{restatable}{prop}{TMI} \label{thm:tmi}
        If the single-intervention process $\ket{\ups}$ is volume-law spatiotemporally entangled in the splitting $BR_1:R_2$, then the corresponding channel is strongly scrambling, i.e.  $I_3(B:R_1:R_2) \approx -2 \log(d_B) $. 
\end{restatable}
This is proved in Appendix~\ref{ap:proofs2}.

Given this connection, we can see that one could easily generalize the Tripartite Mutual Information signature of chaos to a many-time butterfly space rather than single-time, together with some bipartition of the final pure state on $\mc{H}_R$. This may offer new insight into the sensitivity of many-body systems to multitime interventions.

We will now go on to discuss connections of the Butterfly Flutter Fidelity, as in \ref{c3}, with previous signatures.

\subsection{Local Operator Entanglement and OTOCs} \label{sec:loe}
Consider an initially local operator that evolves in time according to the Heisenberg picture,
\begin{equation} \label{eq:HeisOp}
    X_t = U_t^\dg X U_t.
\end{equation}
One can compute the Choi state of this object by acting it on one half of a maximally entangled state on a doubled space, using the CJI as described in Section~\ref{sec:processes},
\begin{equation}
    \ket{X_t}=: X_t \otimes \id \ket{\phi^+}.
\end{equation}
The entanglement of this object across some spatial bipartition is known as the Local-Operator Entanglement, and its scaling in time is considered to be a signature of chaos~\cite{Prosen2007,Prosen2007a,Kos2020}. In particular, if it scales linearly with time then the dynamics cannot be efficiently classically simulated, and linear scaling is conjectured to be characteristic of non-integrability~\cite{Prosen2007a,Prosen2009,Muth2011,Dubail_2017,Jonay2018,Alba2019,Kos2020,Alba2021}.

\begin{restatable}{thm}{LOEbff} \label{thm:LOE_BFF}
    Consider the Butterfly Flutter Fidelity \eqref{eq:state_chaos}, choosing the set of correction unitaries without volume-law entanglement, $\mc{R} = \mc{R}_{\mathrm{MPO}}$, and the single time butterfly flutters chosen to be the identity matrix $\id$ and local unitary $X$. Then if for any initial state, $\zeta(\ups) \approx 1$, then also
    \begin{equation}
        S(\ket{X_t}) \sim \mc{O}(\log(t)),
    \end{equation}
    characteristic of (interacting) integrable dynamics.  
\end{restatable}
This is proved in Appendix~\ref{ap:proofs2}.
By the contrapositve statement of Theorem~\ref{thm:LOE_BFF}, we can see that if the Local Operator Entanglement scales linearly, the Butterfly Flutter Fidelity is small, i.e. chaotic according to our prescription \ref{c3}. 

The Local-Operator Entanglement is intimately related to the OTOC. In Ref.~\cite{dowling2023scrambling} it is shown the OTOC serves as a probe of Local-Operator Entanglement, with exponential scaling of the OTOC being a strictly necessary condition for linear (chaotic) Local-Operator Entanglement growth. We suspect that there may be strong connections between a multitime generalization of the OTOC~\cite{Roberts2017-en}, a kind of multi-point operator entanglement, and the volume-law spatiotemporal entanglement structure as in \ref{c2}. We leave this for future work.

\subsection{Discussion: Chaos and Many-body Phenomena} \label{sec:discussion1}
Throughout this section, we have shown how the three conditions \ref{c1}-\ref{c3} encapsulate some of the most common quantum chaos diagnostics studied in recent years (summarized in Fig.~\ref{fig:flow}). In contrast to these other approaches, we have started with a highly intuitive principle of chaos as a sensitivity to perturbation, without appealing to classical limits (which may not be well defined in many quantum systems) or heuristic observations. This leads to a rather direct notion of which signatures are stronger than others, and a framework with which to analyze the chaoticity of a system.

We now make a few comments about how our formalism compares to others often considered in many-body physics. The setup for butterfly flutters, Def.~\ref{def:butterfly}, strongly resembles Floquet systems~\cite{Heyl_Hauke_Zoller_2019}. A Floquet Hamiltonian is a periodic time-dependent Hamiltonian, such as that produced by introducing a periodic kick to an otherwise time independent Hamiltonian. If, in our construction, we replace the (small) butterfly flutters with strong, global unitaries, they no longer function as a \emph{small} perturbation to the process. Instead, they may change the qualitative behavior of process, possibly creating chaos or order. For example, the quantum rotor is clearly a regular system, whereas the quantum kicked rotor is chaotic for strong enough kicks~\cite{SANTHANAM20221}. The key difference here, is that a system-wide strong unitary does not classify as a butterfly flutter (as in Def.~\ref{def:butterfly}), as it is neither weak nor localized. Considering \ref{c2}, a strong global `perturbation' can change the entanglement structure of a process. Likewise, a strong global butterfly acting on an already chaotic dynamics will likely remain chaotic. 

However, this is not always the case. If a butterfly acts strongly and locally on the whole system-environment, we expect that it can break the volume-law spread of entanglement. This may render a process area-law (or sub-volume-law) and hence regular according to \ref{c2}. It would be interesting to determine the entanglement structure of systems exhibiting Many-Body Localization (MBL)~\cite{ALET2018,Geraedts2016}, and measurement-induced phase transitions~\cite{Skinner2019,Zhang2020}. MBL systems are known to be resistant to perturbation, the opposite of chaotic according to the principles underlying \ref{c1}-\ref{c3}. While these two phenomena were previously surprising, the framework presented here offers a novel path to systematically studying the mechanisms behind them. Such topics would be interesting to investigate in more detail in future work.

\section{Mechanisms for Chaos} \label{sec:typ_process}
So far in Section~\ref{sec:mainresult} we have proposed a hierarchy of conditions on quantum chaos, inspired by the principle of chaos as a sensitivity to perturbation. This culminated in the metric of the Butterfly Flutter Fidelity, closely connected to the spatiotemporal entanglement of the corresponding process $\ket{\ups}$. Then in Section~\ref{sec:prev_chaos} we have shown how this connects to and encompasses a range of existing dynamical signatures. Looking at the summary of this work, Fig.~\ref{fig:internal-external}, we have yet to discuss the mechanisms of chaos on the left of this figure; the underlying properties of the dynamics that lead to chaotic phenomena in a quantum system. 

We will now analyze two broad classes of dynamics, and show through these that randomness typically leads to chaos. Consider dynamics which is globally random. More formally, we independently sample unitary matrices from the Haar probability measure $U_i \sim {\mathbb{H}}$ between each intervention in the butterfly flutter protocol \eqref{eq:psiE}. ${\mathbb{H}}$ is the unique, unitarily invariant measure, meaning that if any ensemble $\{ U_i \}$ is distributed according to the Haar measure, then so is $\{W U_i \}$ and $\{ U_i W\}$ for any unitary $W$. Considering such random unitaries allows one to derive strong concentration of measure bounds. One such famous example for quantum states says that small subsystems of large random pure states are exponentially likely to be maximally mixed~\cite{Popescu2006}. Similarly, processes sampled through Haar random evolution between inventions are highly likely to look like the completely noisy process, given a large environment dimension~\cite{FigueroaRomero_Modi_Pollock_2019,FigueroaRomero_Pollock_Modi_2021}; see Appendix~\ref{ap:pedro}. By a completely noisy process, we mean that any measurements result in equal weights, corresponding to the identity matrix Choi state as in Eq.~\eqref{eq:typ_process}.

However, strictly Haar random evolution is not entirely physical, with the full, exponentially large Hilbert space not practically accessible; a `convenient illusion'~\cite{Poulin2011}. On the other hand, quantum circuits with finite depth represent a far more reasonable model for realistic dynamics. Moreover, one can simulate randomness up to the first $t$ moments using unitary design circuits. An $\epsilon-$approximate $t-$design can formally be defined such that 
\begin{equation}
    \mc{D}\left( \mathbb{E}_{\mu_{t_\epsilon}}\left\{ (U^\dg)^{\otimes s} (X) U^{\otimes s}   \right\}- \mathbb{E}_{\mathbb{H}}\left\{ (U^\dg)^{\otimes s} (X) U^{\otimes s}   \right\}\right) \leq \epsilon, \nn
\end{equation}
for all $s\leq t$, some appropriate metric $\mc{D}$, and any observable $X \in \mc{H}^{\otimes s}$. In words, the $s-$fold channel of a $t-$design needs to approximately agree with perfectly Haar random sampling. Such design circuits therefore simulate full unitary randomness, but are much more akin to real physical models. For example, an $\epsilon-$approximate $2-$design can be generated efficiently from two-qubit gates only in polynomial time~\cite{Winter2017}. This is equivalent to a model of two different two-body interactions occurring randomly in a system~\cite{FigueroaRomero_Pollock_Modi_2021}.

We will now give concentration of measure bounds both for unitary designs, and for full Haar random evolution. We will see that sampling from these random dynamics is highly likely to produce a process with volume-law spatiotemporal entanglement, as in \ref{c2}.

\begin{restatable}{thm}{buttTyp} \textbf{(Most Processes are Chaotic)} \label{thm:rand_process_zeta}
    Consider a pure process $\ket{\ups}$ generated by random dynamics, either entirely Haar random denoted by $\mathbb{H}$ or according to an $\epsilon-$approximate $t-$design denoted by $\mu_{\epsilon-t}$. Then for any $R_1 \subset R$ such that $d_{R_1} \approx d_S$, and for any $\delta>0$ and $0<m<t/4$, 
        \begin{align} 
            \mathbb{P}_{U_i \sim \mu}\Big\{  \log(d_{B R_1}) - S^{(2)}(\ups_{B R_1})  \geq  \mc{J}_\mu(\delta ) \Big\} \leq \mc{G}_{\mu}(\delta).\label{eq:statechaos_typicality}
        \end{align}
        Where for a process generated from independent Haar-random evolution,
           \begin{align}
            &\mc{J}_{\mathbb{H}}(\delta) = \log(d_{B R_1} (\mc{B}+\delta) + 1 )  \approx d_{BR_1}(\frac{1}{d_{R}}+\delta),\text{ and} \nn \\
            &\mc{G}_{\mathbb{H}}(\delta) = \exp [-\mc{C} \delta^2]\approx \exp [-\frac{(k+1) d_R }{8d_B} \delta^2], \label{eq:haar_typ}
        \end{align}
        while for that generated from an $\epsilon-$approximate unitary $t-$design 
        \begin{align}
            &\mc{J}_{\mu_{\epsilon-t}}(\delta) = \log(d_{B R_1} \delta + 1 ) \approx d_{B R_1} \delta,\text{ and} \nn\\
            &\mc{G}_{\mu_{\epsilon-t}}(\delta) = \frac{\mc{F}(d_B,d_R,m,t,\epsilon )}{\delta^{m}}\label{eq:t_typ}
        \end{align}
        The exact definition of $\mc{B}$, $\mc{C}$, and $\mc{F}$  are provided in Eqs.~\eqref{eq:betadef},~\eqref{eq:cdef} and~\eqref{eq:fdef} respectively. The approximations in Eqs.~\eqref{eq:haar_typ}-\eqref{eq:t_typ} are valid for $d_R \gg d_B \gg 1$, and for small $\delta$.
\end{restatable}
The proof of this theorem builds on results from Refs.~\cite{FigueroaRomero_Modi_Pollock_2019,FigueroaRomero_Pollock_Modi_2021}, and can be found in Appendix~\ref{ap:pedro}. The result Eq.~\eqref{eq:statechaos_typicality} states that random dynamics are likely to lead to a volume-law spatiotemporal entanglement, according to a small butterfly flutter in comparison to the total isolated system. In particular, for Haar random dynamics, Eq.~\eqref{eq:haar_typ} indicates an exponentially small probability that a single sampling deviates strongly from maximal entanglement in the splitting $BR_1:R_2$. Further, this result is valid for any choice of $\mc{H}_{R_1}$, given that it is small enough in comparison to the full system $\mc{H}_R$. This directly implies that random dynamics typically have volume-law spatiotemporal entanglement.

Note that the bounds given here is for the independently sampled evolution between butterfly times, but we note that techniques in Ref.~\cite{FigueroaRomero_Modi_Pollock_2019} can be used to prove similar bound for repeated dynamics, i.e. a single sample of a unitary evolution matrix that describes all dynamics between interventions. 

Similarly, random circuits yield a related bound in terms of how well they approximate a unitary design. In this case, Eq.~\eqref{eq:t_typ} is a polynomially small bound, and in practice can be optimized over the parameter $m$. The key point is that both of these probability bounds are small for $d_E \gg d_{B    R_1}$.

While these concentration of measure bounds are for the spatiotemporal entanglement of $\ket{\ups}$, similar bounds can also be proved for other dynamical signatures that derive from this, such as those considered in Section~\ref{sec:prev_chaos}. For example, Dynamical Entropy is likely to be extensive according to this result. This is immediate to see from Theorem~\ref{thm:rand_process_zeta} by choosing $\mc{H}_{R_1}$ to be empty. This therefore means that repeated measurements of a process generated from random evolution give almost maximal information. That is, one typically only sees approximately maximally noisy measurement results.

We have shown that Haar random evolution, as well as that generated by $\epsilon-$approximate $t-$designs, constitute mechanisms that are highly like to produce chaos. This is clearly not the only internal mechanism that causes chaotic phenomena;  c. f. Fig.~\ref{fig:internal-external} \textbf{a)}. The next step will be to understand how a continuous quantum evolution, defined by time independent Hamiltonians, can lead to chaos.

For example, the so-called Wigner-Dyson level spacing distribution is often conflated with quantum chaos~\cite{BerryTabor1977,Haake2018-cs}. This is the empirical observation that if one computes the distribution between next-neighbor energy levels, it follows a characteristic form when the semiclassical limit of the Hamiltonian is chaotic. An interesting connection may be found in entanglement spectra, which can be connected to a sense of irreversibility of the dynamics~\cite{Chamon2014-cb}. Another example is the eigenstate thermalization hypothesis (ETH), which proposes that certain `physical' observables look thermal according to individual eigenstates of certain Hamiltonians. Often one calls such Hamiltonians chaotic, and the ETH leads to a deterministic (pure-state) foundation of statistical mechanics results. 

It would be interesting to determine how (if) these mechanisms lead to volume-law spatiotemporal entanglement within a process, to prove that they are mechanisms of chaos as in Fig.~\ref{fig:internal-external}. Indeed, such a connection would firmly cement quantum chaos as a foundational, deterministic principle underlying statistical mechanics, in perfect analogy with the classical case. Volume-law entanglement of eigenstates is already a key feature of the strong ETH. In addition, for a wide range of specific Hamiltonian classes, Ref.~\cite{Bianchi_Hackl_Kieburg_Rigol_Vidmar_2022} determine that volume-law entanglement is highly typical. In this context, a key question will be how (many-body) quantum scars play into this, i.e. when some eigenstates of an apparently chaotic Hamiltonian do not satisfy the ETH. Such eigenstates can have different entanglement scaling~\cite{Serbyn_2021,Moudgalya_2022}. 

Finally, the typicality bounds presented here have foundational implications regarding the prevalence of Markovianity in nature, which we now discuss in our concluding remarks.

\section{Conclusions}
Starting from a theory-independent notion of chaos as the butterfly effect, in this work we have identified a series of conditions on quantum chaos (Section~\ref{sec:mainresult}), with the strongest being measured by the Butterfly Flutter Fidelity, shown that these proposed conditions generalize and hence unify a range of previous diagnostics (Section~\ref{sec:prev_chaos}), and shown how a number of mechanisms that lead to quantum chaos (Section~\ref{sec:typ_process}). This framework is summarized in Fig.~\ref{fig:internal-external}. 

The results of Refs.~\cite{FigueroaRomero_Modi_Pollock_2019,FigueroaRomero_Pollock_Modi_2021} state that processes generated from random dynamics are highly likely to be almost Markovian, for large enough systems. Paradoxically, Thm.~\ref{thm:rand_process_zeta} states that perturbations in such processes have a strong impact in the environment. That is, most random processes are chaotic. To make sense of this, note that Markovianity is with respect to a restricted measurement space, often taken to be small. Then, when a process is highly chaotic, a butterfly impacts the future pure state in such a strong and non-local way, that for any small subsystem it looks entirely noisy and hence Markovian on this future measurement space. 
Given that in nature chaos is the rule, not the exception, this helps address the fundamental question of why Markovian phenomena are so prevalent in nature~\cite{Dowling2021,finitetime,FigueroaRomero_Modi_Pollock_2019,FigueroaRomero_Pollock_Modi_2021,Strasberg2022}: chaotic processes on large systems look Markovian with respect to interventions on a much smaller subsystem. We anticipate that this may be a key factor in understanding the emergence of thermalization from underlying quantum theory; in particular the necessary loss of memory in the process of thermalization. It would interesting to investigate this further in a future work. 

This is related to Refs.~\cite{Bremner_Mora_Winter_2009,Gross_Flammia_Eisert_2009}, where it is shown that states which are too entangled - which is most states in the full Hilbert space - are not useful for measurement-based quantum computation. For such states which are too entangled, one can replace the local statistics with `coin flipping' -- purely classical stochasticity. It is, however, very difficult to produce large, highly entangled states. Usefulness is not necessarily proportional to the resources required to create a state. Our results in Section~\ref{sec:typ_process} are a spatiotemporal version of this. Most processes are so chaotic, that future measurements statistics constitute purely classical noise. What is needed, then, to have complex, quantum non-Markovian phenomena? We propose that it is `between order and chaos' where these interesting processes lie~\cite{Crutchfield2012}. This would correspond to processes with sub-volume-law (logarithmic) spatiotemporal entanglement scaling. This is intrinsically tied to criticality in the spatial setting, and the Multiscale Entanglement Renormalization Ansatz (MERA) tensor network~\cite{Vidal2007}. Current research explores a process tensor ansatz, inspired by MERA, structurally exhibiting long-range (polynomially decaying) temporal correlations~\cite{TeMERA2023}.

A relevant problem which we have not tackled in this work is the question of how (if) classical chaos emerges from quantum chaos in some limit. While historically this was the main motivation for understanding quantum chaos~\cite{BerryTabor1977,Haake2018-cs,reichl2021transition}, here we have developed a genuinely quantum notion of chaos, of interest for the wide range of phenomena and modern experiments in many-body physics with no classical analogue. It is therefore an open question how exactly to connect this to the classical picture. Modern notions of the transition to classicality may be integral to understanding this, such as quantum Darwinism~\cite{Zurek2003} or classical stochasticity arising from quantum theory~\cite{strasberg2019,Milz2020prx,Strasberg2022}. Related to this is Ref.~\cite{Leone2021quantumchaosis}, where it is shown that circuits generated solely by Clifford gates, or doped with only a few non-Clifford gates, are not chaotic according to a signature based off a generalized OTOC. It would be interesting to check what kind of entanglement structure a (doped) Clifford circuit has, that is whether this statement is consistent with the structure of chaos we have revealed in this work. This would have implications regarding whether any chaotic quantum process, satisfying the strongest condition \ref{c3}, can be simulated classically.

% The quantum butterfly effect Eq.~\eqref{eq:state_chaos} was constructed in terms of a fidelity between final states. However, there may be another, more appropriate metric. Given the realization of Result~\ref{thm:qchaos}, one could concoct any such distance or diagnostic which probes the spatiotemporal entanglement structure of a process. One possibly interesting  metric is that from Ref.~\cite{Nico-Katz_Bose_2022}. There they define a class of geometric measures to the closest MPS of some constant bond dimension $D$. In the process setting, from Result~\ref{thm:qchaos} an equivalent concept would then correspond to the distance to the closest non-chaotic process. The utility of the principle of quantum chaos presented in this work is that it allows the understanding of the underlying cause of \emph{any} chaos heuristic, and as a test of its validity. Beyond this, as briefly discussed in Section~\ref{sec:discussion1}, using the framework stemming from Result~\ref{thm:qchaos}, one can reinterpret and approach many problems and phenomena in many-body physics, such as Many-Body Localization~\cite{ALET2018,Geraedts2016}, many-body quantum scars~\cite{Moudgalya_2022}, and measurement induced phase transitions~\cite{Skinner2019,Zhang2020}, to name a few.

It is difficult to directly convert from classical to quantum chaos, due to the linearity of isolated quantum mechanics. The novelty of our approach is that it treats chaos itself as a primitive concept, independent of whether we adopt a classical or quantum formalism. Classically, this reduces to a non-linearity of the dynamics in phase space. On the quantum side of things, we have shown that spatiotemporal entanglement structure directly satisfies this principle: perturb a small part of a system in the past, and see a complex, non-local effect in the future. From this realization, we have shown that previous diagnostics fit perfectly within this framework. Further, one can systematically compare our framework with any other quantum chaos diagnostic, and use the new metrics to tackle a wide range of relevant problems in the field of many-body physics.

\begin{acknowledgments}
% We thank the anonymous referees for comments leading to an improved manuscript.
We thank S. Singh for several technical discussions. KM thanks A. Gorecka for conversations that distilled some key ideas. ND is supported by an Australian Government Research Training Program Scholarship and the Monash Graduate Excellence Scholarship. KM acknowledges support from the Australian Research Council Future Fellowship FT160100073, Discovery Projects DP210100597 and DP220101793, and the International Quantum U Tech Accelerator award by the US Air Force Research Laboratory.
\end{acknowledgments}

% \bibliography{bib-file}

\begin{thebibliography}{112}%
\makeatletter
\providecommand \@ifxundefined [1]{%
 \@ifx{#1\undefined}
}%
\providecommand \@ifnum [1]{%
 \ifnum #1\expandafter \@firstoftwo
 \else \expandafter \@secondoftwo
 \fi
}%
\providecommand \@ifx [1]{%
 \ifx #1\expandafter \@firstoftwo
 \else \expandafter \@secondoftwo
 \fi
}%
\providecommand \natexlab [1]{#1}%
\providecommand \enquote  [1]{``#1''}%
\providecommand \bibnamefont  [1]{#1}%
\providecommand \bibfnamefont [1]{#1}%
\providecommand \citenamefont [1]{#1}%
\providecommand \href@noop [0]{\@secondoftwo}%
\providecommand \href [0]{\begingroup \@sanitize@url \@href}%
\providecommand \@href[1]{\@@startlink{#1}\@@href}%
\providecommand \@@href[1]{\endgroup#1\@@endlink}%
\providecommand \@sanitize@url [0]{\catcode `\\12\catcode `\$12\catcode
  `\&12\catcode `\#12\catcode `\^12\catcode `\_12\catcode `\%12\relax}%
\providecommand \@@startlink[1]{}%
\providecommand \@@endlink[0]{}%
\providecommand \url  [0]{\begingroup\@sanitize@url \@url }%
\providecommand \@url [1]{\endgroup\@href {#1}{\urlprefix }}%
\providecommand \urlprefix  [0]{URL }%
\providecommand \Eprint [0]{\href }%
\providecommand \doibase [0]{http://dx.doi.org/}%
\providecommand \selectlanguage [0]{\@gobble}%
\providecommand \bibinfo  [0]{\@secondoftwo}%
\providecommand \bibfield  [0]{\@secondoftwo}%
\providecommand \translation [1]{[#1]}%
\providecommand \BibitemOpen [0]{}%
\providecommand \bibitemStop [0]{}%
\providecommand \bibitemNoStop [0]{.\EOS\space}%
\providecommand \EOS [0]{\spacefactor3000\relax}%
\providecommand \BibitemShut  [1]{\csname bibitem#1\endcsname}%
\let\auto@bib@innerbib\@empty
%</preamble>
\bibitem [{\citenamefont {Kudler-Flam}\ \emph {et~al.}(2020)\citenamefont
  {Kudler-Flam}, \citenamefont {Nie},\ and\ \citenamefont
  {Ryu}}]{Kudler-Flam2020}%
  \BibitemOpen
  \bibfield  {author} {\bibinfo {author} {\bibfnamefont {J.}~\bibnamefont
  {Kudler-Flam}}, \bibinfo {author} {\bibfnamefont {L.}~\bibnamefont {Nie}}, \
  and\ \bibinfo {author} {\bibfnamefont {S.}~\bibnamefont {Ryu}},\ }\href
  {\doibase 10.1007/JHEP01(2020)175} {\bibfield  {journal} {\bibinfo  {journal}
  {Journal of High Energy Physics}\ }\textbf {\bibinfo {volume} {2020}},\
  \bibinfo {pages} {175} (\bibinfo {year} {2020})}\BibitemShut {NoStop}%
\bibitem [{\citenamefont {Berry}\ \emph {et~al.}(1977)\citenamefont {Berry},
  \citenamefont {Tabor},\ and\ \citenamefont {Ziman}}]{BerryTabor1977}%
  \BibitemOpen
  \bibfield  {author} {\bibinfo {author} {\bibfnamefont {M.~V.}\ \bibnamefont
  {Berry}}, \bibinfo {author} {\bibfnamefont {M.}~\bibnamefont {Tabor}}, \ and\
  \bibinfo {author} {\bibfnamefont {J.~M.}\ \bibnamefont {Ziman}},\ }\href
  {\doibase 10.1098/rspa.1977.0140} {\bibfield  {journal} {\bibinfo  {journal}
  {Proceedings of the Royal Society of London. A. Mathematical and Physical
  Sciences}\ }\textbf {\bibinfo {volume} {356}},\ \bibinfo {pages} {375}
  (\bibinfo {year} {1977})}\BibitemShut {NoStop}%
\bibitem [{\citenamefont {Lindblad}(1986)}]{Lindblad1986chaos}%
  \BibitemOpen
  \bibfield  {author} {\bibinfo {author} {\bibfnamefont {G.}~\bibnamefont
  {Lindblad}},\ }in\ \href@noop {} {\emph {\bibinfo {booktitle} {Fundamental
  Aspects of Quantum Theory}}},\ \bibinfo {editor} {edited by\ \bibinfo
  {editor} {\bibfnamefont {V.}~\bibnamefont {Gorini}}\ and\ \bibinfo {editor}
  {\bibfnamefont {A.}~\bibnamefont {Frigerio}}}\ (\bibinfo  {publisher} {Plenum
  Press},\ \bibinfo {address} {New York},\ \bibinfo {year} {1986})\ p.\
  \bibinfo {pages} {199}\BibitemShut {NoStop}%
\bibitem [{\citenamefont {Reichl}(2021)}]{reichl2021transition}%
  \BibitemOpen
  \bibfield  {author} {\bibinfo {author} {\bibfnamefont {L.}~\bibnamefont
  {Reichl}},\ }\href@noop {} {\emph {\bibinfo {title} {The Transition to Chaos:
  Conservative Classical and Quantum Systems}}},\ Vol.\ \bibinfo {volume}
  {200}\ (\bibinfo  {publisher} {Springer Nature},\ \bibinfo {year}
  {2021})\BibitemShut {NoStop}%
\bibitem [{\citenamefont {Haake}\ \emph {et~al.}(2018)\citenamefont {Haake},
  \citenamefont {Gnutzmann},\ and\ \citenamefont {Ku{\'s}}}]{Haake2018-cs}%
  \BibitemOpen
  \bibfield  {author} {\bibinfo {author} {\bibfnamefont {F.}~\bibnamefont
  {Haake}}, \bibinfo {author} {\bibfnamefont {S.}~\bibnamefont {Gnutzmann}}, \
  and\ \bibinfo {author} {\bibfnamefont {M.}~\bibnamefont {Ku{\'s}}},\
  }\href@noop {} {\emph {\bibinfo {title} {Quantum Signatures of Chaos}}}\
  (\bibinfo  {publisher} {Springer, Cham},\ \bibinfo {year} {2018})\BibitemShut
  {NoStop}%
\bibitem [{\citenamefont {Hayden}\ and\ \citenamefont
  {Preskill}(2007)}]{Hayden_Preskill_2007}%
  \BibitemOpen
  \bibfield  {author} {\bibinfo {author} {\bibfnamefont {P.}~\bibnamefont
  {Hayden}}\ and\ \bibinfo {author} {\bibfnamefont {J.}~\bibnamefont
  {Preskill}},\ }\href {\doibase 10.1088/1126-6708/2007/09/120} {\bibfield
  {journal} {\bibinfo  {journal} {Journal of High Energy Physics}\ }\textbf
  {\bibinfo {volume} {2007}},\ \bibinfo {pages} {120} (\bibinfo {year}
  {2007})}\BibitemShut {NoStop}%
\bibitem [{\citenamefont {Shenker}\ and\ \citenamefont
  {Stanford}(2014{\natexlab{a}})}]{Shenker_Stanford_2014}%
  \BibitemOpen
  \bibfield  {author} {\bibinfo {author} {\bibfnamefont {S.~H.}\ \bibnamefont
  {Shenker}}\ and\ \bibinfo {author} {\bibfnamefont {D.}~\bibnamefont
  {Stanford}},\ }\href {\doibase 10.1007/JHEP03(2014)067} {\bibfield  {journal}
  {\bibinfo  {journal} {Journal of High Energy Physics}\ }\textbf {\bibinfo
  {volume} {2014}},\ \bibinfo {pages} {67} (\bibinfo {year}
  {2014}{\natexlab{a}})}\BibitemShut {NoStop}%
\bibitem [{\citenamefont {Popescu}\ \emph {et~al.}(2006)\citenamefont
  {Popescu}, \citenamefont {Short},\ and\ \citenamefont
  {Winter}}]{Popescu2006}%
  \BibitemOpen
  \bibfield  {author} {\bibinfo {author} {\bibfnamefont {S.}~\bibnamefont
  {Popescu}}, \bibinfo {author} {\bibfnamefont {A.~J.}\ \bibnamefont {Short}},
  \ and\ \bibinfo {author} {\bibfnamefont {A.}~\bibnamefont {Winter}},\ }\href
  {\doibase 10.1038/nphys444} {\bibfield  {journal} {\bibinfo  {journal} {Nat.
  Phys.}\ }\textbf {\bibinfo {volume} {2}},\ \bibinfo {pages} {754} (\bibinfo
  {year} {2006})}\BibitemShut {NoStop}%
\bibitem [{\citenamefont {Gogolin}\ and\ \citenamefont
  {Eisert}(2016)}]{Gogolin}%
  \BibitemOpen
  \bibfield  {author} {\bibinfo {author} {\bibfnamefont {C.}~\bibnamefont
  {Gogolin}}\ and\ \bibinfo {author} {\bibfnamefont {J.}~\bibnamefont
  {Eisert}},\ }\href {\doibase 10.1088/0034-4885/79/5/056001} {\bibfield
  {journal} {\bibinfo  {journal} {Rep. Prog. Phys.}\ }\textbf {\bibinfo
  {volume} {79}},\ \bibinfo {pages} {056001} (\bibinfo {year}
  {2016})}\BibitemShut {NoStop}%
\bibitem [{\citenamefont {D'Alessio}\ \emph {et~al.}(2016)\citenamefont
  {D'Alessio}, \citenamefont {Kafri}, \citenamefont {Polkovnikov},\ and\
  \citenamefont {Rigol}}]{Rigol2016}%
  \BibitemOpen
  \bibfield  {author} {\bibinfo {author} {\bibfnamefont {L.}~\bibnamefont
  {D'Alessio}}, \bibinfo {author} {\bibfnamefont {Y.}~\bibnamefont {Kafri}},
  \bibinfo {author} {\bibfnamefont {A.}~\bibnamefont {Polkovnikov}}, \ and\
  \bibinfo {author} {\bibfnamefont {M.}~\bibnamefont {Rigol}},\ }\href
  {\doibase 10.1080/00018732.2016.1198134} {\bibfield  {journal} {\bibinfo
  {journal} {Adv. Phys.}\ }\textbf {\bibinfo {volume} {65}},\ \bibinfo {pages}
  {239} (\bibinfo {year} {2016})}\BibitemShut {NoStop}%
\bibitem [{\citenamefont {Peres}(1984)}]{Peres1984stab}%
  \BibitemOpen
  \bibfield  {author} {\bibinfo {author} {\bibfnamefont {A.}~\bibnamefont
  {Peres}},\ }\href {\doibase 10.1103/PhysRevA.30.1610} {\bibfield  {journal}
  {\bibinfo  {journal} {Phys. Rev. A}\ }\textbf {\bibinfo {volume} {30}},\
  \bibinfo {pages} {1610} (\bibinfo {year} {1984})}\BibitemShut {NoStop}%
\bibitem [{\citenamefont {Emerson}\ \emph {et~al.}(2002)\citenamefont
  {Emerson}, \citenamefont {Weinstein}, \citenamefont {Lloyd},\ and\
  \citenamefont {Cory}}]{Emerson2002}%
  \BibitemOpen
  \bibfield  {author} {\bibinfo {author} {\bibfnamefont {J.}~\bibnamefont
  {Emerson}}, \bibinfo {author} {\bibfnamefont {Y.~S.}\ \bibnamefont
  {Weinstein}}, \bibinfo {author} {\bibfnamefont {S.}~\bibnamefont {Lloyd}}, \
  and\ \bibinfo {author} {\bibfnamefont {D.~G.}\ \bibnamefont {Cory}},\ }\href
  {\doibase 10.1103/PhysRevLett.89.284102} {\bibfield  {journal} {\bibinfo
  {journal} {Phys. Rev. Lett.}\ }\textbf {\bibinfo {volume} {89}},\ \bibinfo
  {pages} {284102} (\bibinfo {year} {2002})}\BibitemShut {NoStop}%
\bibitem [{\citenamefont {Pechukas}(1982)}]{Pechukas1982}%
  \BibitemOpen
  \bibfield  {author} {\bibinfo {author} {\bibfnamefont {P.}~\bibnamefont
  {Pechukas}},\ }\href {\doibase 10.1021/j100209a019} {\bibfield  {journal}
  {\bibinfo  {journal} {The Journal of Physical Chemistry}\ }\textbf {\bibinfo
  {volume} {86}},\ \bibinfo {pages} {2239} (\bibinfo {year}
  {1982})}\BibitemShut {NoStop}%
\bibitem [{\citenamefont {S{\l}omczy{\'n}ski}\ and\ \citenamefont
  {{\.Z}yczkowski}(1994)}]{Slomczynski1994-ns}%
  \BibitemOpen
  \bibfield  {author} {\bibinfo {author} {\bibfnamefont {W.}~\bibnamefont
  {S{\l}omczy{\'n}ski}}\ and\ \bibinfo {author} {\bibfnamefont
  {K.}~\bibnamefont {{\.Z}yczkowski}},\ }\href {\doibase 10.1063/1.530704}
  {\bibfield  {journal} {\bibinfo  {journal} {J. Math. Phys.}\ }\textbf
  {\bibinfo {volume} {35}},\ \bibinfo {pages} {5674} (\bibinfo {year}
  {1994})}\BibitemShut {NoStop}%
\bibitem [{\citenamefont {Alicki}\ and\ \citenamefont
  {Fannes}(1994)}]{Alicki1994-dc}%
  \BibitemOpen
  \bibfield  {author} {\bibinfo {author} {\bibfnamefont {R.}~\bibnamefont
  {Alicki}}\ and\ \bibinfo {author} {\bibfnamefont {M.}~\bibnamefont
  {Fannes}},\ }\href {\doibase 10.1007/BF00761125} {\bibfield  {journal}
  {\bibinfo  {journal} {Lett. Math. Phys.}\ }\textbf {\bibinfo {volume} {32}},\
  \bibinfo {pages} {75} (\bibinfo {year} {1994})}\BibitemShut {NoStop}%
\bibitem [{\citenamefont {Prosen}\ and\ \citenamefont {\ifmmode \check{Z}\else
  \v{Z}\fi{}nidari\ifmmode~\check{c}\else \v{c}\fi{}}(2007)}]{Prosen2007}%
  \BibitemOpen
  \bibfield  {author} {\bibinfo {author} {\bibfnamefont {T.~c.~v.}\
  \bibnamefont {Prosen}}\ and\ \bibinfo {author} {\bibfnamefont
  {M.}~\bibnamefont {\ifmmode \check{Z}\else
  \v{Z}\fi{}nidari\ifmmode~\check{c}\else \v{c}\fi{}}},\ }\href {\doibase
  10.1103/PhysRevE.75.015202} {\bibfield  {journal} {\bibinfo  {journal} {Phys.
  Rev. E}\ }\textbf {\bibinfo {volume} {75}},\ \bibinfo {pages} {015202(R)}
  (\bibinfo {year} {2007})}\BibitemShut {NoStop}%
\bibitem [{\citenamefont {Prosen}\ and\ \citenamefont
  {Pi\ifmmode~\check{z}\else \v{z}\fi{}orn}(2007)}]{Prosen2007a}%
  \BibitemOpen
  \bibfield  {author} {\bibinfo {author} {\bibfnamefont {T.~c.~v.}\
  \bibnamefont {Prosen}}\ and\ \bibinfo {author} {\bibfnamefont
  {I.}~\bibnamefont {Pi\ifmmode~\check{z}\else \v{z}\fi{}orn}},\ }\href
  {\doibase 10.1103/PhysRevA.76.032316} {\bibfield  {journal} {\bibinfo
  {journal} {Phys. Rev. A}\ }\textbf {\bibinfo {volume} {76}},\ \bibinfo
  {pages} {032316} (\bibinfo {year} {2007})}\BibitemShut {NoStop}%
\bibitem [{\citenamefont {Bertini}\ \emph {et~al.}(2020)\citenamefont
  {Bertini}, \citenamefont {Kos},\ and\ \citenamefont {Prosen}}]{Kos2020}%
  \BibitemOpen
  \bibfield  {author} {\bibinfo {author} {\bibfnamefont {B.}~\bibnamefont
  {Bertini}}, \bibinfo {author} {\bibfnamefont {P.}~\bibnamefont {Kos}}, \ and\
  \bibinfo {author} {\bibfnamefont {T.}~\bibnamefont {Prosen}},\ }\href
  {\doibase 10.21468/SciPostPhys.8.4.067} {\bibfield  {journal} {\bibinfo
  {journal} {SciPost Phys.}\ }\textbf {\bibinfo {volume} {8}},\ \bibinfo
  {pages} {067} (\bibinfo {year} {2020})}\BibitemShut {NoStop}%
\bibitem [{\citenamefont {Hosur}\ \emph {et~al.}(2016)\citenamefont {Hosur},
  \citenamefont {Qi}, \citenamefont {Roberts},\ and\ \citenamefont
  {Yoshida}}]{Hosur2016}%
  \BibitemOpen
  \bibfield  {author} {\bibinfo {author} {\bibfnamefont {P.}~\bibnamefont
  {Hosur}}, \bibinfo {author} {\bibfnamefont {X.-L.}\ \bibnamefont {Qi}},
  \bibinfo {author} {\bibfnamefont {D.~A.}\ \bibnamefont {Roberts}}, \ and\
  \bibinfo {author} {\bibfnamefont {B.}~\bibnamefont {Yoshida}},\ }\href
  {\doibase 10.1007/JHEP02(2016)004} {\bibfield  {journal} {\bibinfo  {journal}
  {Journal of High Energy Physics}\ }\textbf {\bibinfo {volume} {2016}},\
  \bibinfo {pages} {4} (\bibinfo {year} {2016})}\BibitemShut {NoStop}%
\bibitem [{Note1()}]{Note1}%
  \BibitemOpen
  \bibinfo {note} {These are a selection of some of the most popularly accepted
  chaos probes, but this is by no means an exhaustive list. For example, some
  interesting alternative measures are entanglement spectrum statistics~\cite
  {Chamon2014-cb}, the Spectral Form Factor~\cite {Haake2018-cs}, and quantum
  coherence measures~\cite {Anand2021-yi}. Also see Refs.~\cite
  {HunterJones2018ChaosAR,Kudler-Flam2020} and references therein.}\BibitemShut
  {Stop}%
\bibitem [{\citenamefont {Deutsch}(1991)}]{Deutsch1991}%
  \BibitemOpen
  \bibfield  {author} {\bibinfo {author} {\bibfnamefont {J.~M.}\ \bibnamefont
  {Deutsch}},\ }\href {\doibase 10.1103/PhysRevA.43.2046} {\bibfield  {journal}
  {\bibinfo  {journal} {Phys. Rev. A}\ }\textbf {\bibinfo {volume} {43}},\
  \bibinfo {pages} {2046} (\bibinfo {year} {1991})}\BibitemShut {NoStop}%
\bibitem [{\citenamefont {Srednicki}(1994)}]{Srednicki}%
  \BibitemOpen
  \bibfield  {author} {\bibinfo {author} {\bibfnamefont {M.}~\bibnamefont
  {Srednicki}},\ }\href {\doibase 10.1103/PhysRevE.50.888} {\bibfield
  {journal} {\bibinfo  {journal} {Phys. Rev. E}\ }\textbf {\bibinfo {volume}
  {50}},\ \bibinfo {pages} {888} (\bibinfo {year} {1994})}\BibitemShut
  {NoStop}%
\bibitem [{\citenamefont {Rigol}\ \emph {et~al.}(2008)\citenamefont {Rigol},
  \citenamefont {Dunjko},\ and\ \citenamefont {Olshanii}}]{Rigol2008}%
  \BibitemOpen
  \bibfield  {author} {\bibinfo {author} {\bibfnamefont {M.}~\bibnamefont
  {Rigol}}, \bibinfo {author} {\bibfnamefont {V.}~\bibnamefont {Dunjko}}, \
  and\ \bibinfo {author} {\bibfnamefont {M.}~\bibnamefont {Olshanii}},\
  }\href@noop {} {\bibfield  {journal} {\bibinfo  {journal} {Nature}\ }\textbf
  {\bibinfo {volume} {452}},\ \bibinfo {pages} {854 EP } (\bibinfo {year}
  {2008})}\BibitemShut {NoStop}%
\bibitem [{\citenamefont {Yan}\ \emph {et~al.}(2020)\citenamefont {Yan},
  \citenamefont {Cincio},\ and\ \citenamefont {Zurek}}]{Yan2020}%
  \BibitemOpen
  \bibfield  {author} {\bibinfo {author} {\bibfnamefont {B.}~\bibnamefont
  {Yan}}, \bibinfo {author} {\bibfnamefont {L.}~\bibnamefont {Cincio}}, \ and\
  \bibinfo {author} {\bibfnamefont {W.~H.}\ \bibnamefont {Zurek}},\ }\href
  {\doibase 10.1103/PhysRevLett.124.160603} {\bibfield  {journal} {\bibinfo
  {journal} {Phys. Rev. Lett.}\ }\textbf {\bibinfo {volume} {124}},\ \bibinfo
  {pages} {160603} (\bibinfo {year} {2020})}\BibitemShut {NoStop}%
\bibitem [{\citenamefont {Roberts}\ and\ \citenamefont
  {Yoshida}(2017)}]{Roberts2017-en}%
  \BibitemOpen
  \bibfield  {author} {\bibinfo {author} {\bibfnamefont {D.~A.}\ \bibnamefont
  {Roberts}}\ and\ \bibinfo {author} {\bibfnamefont {B.}~\bibnamefont
  {Yoshida}},\ }\href {\doibase 10.1007/JHEP04(2017)121} {\bibfield  {journal}
  {\bibinfo  {journal} {J. High Energy Phys.}\ }\textbf {\bibinfo {volume}
  {2017}},\ \bibinfo {pages} {121} (\bibinfo {year} {2017})}\BibitemShut
  {NoStop}%
\bibitem [{\citenamefont {Leone}\ \emph
  {et~al.}(2021{\natexlab{a}})\citenamefont {Leone}, \citenamefont {Oliviero},\
  and\ \citenamefont {Hamma}}]{Leone2021-ov}%
  \BibitemOpen
  \bibfield  {author} {\bibinfo {author} {\bibfnamefont {L.}~\bibnamefont
  {Leone}}, \bibinfo {author} {\bibfnamefont {S.~F.~E.}\ \bibnamefont
  {Oliviero}}, \ and\ \bibinfo {author} {\bibfnamefont {A.}~\bibnamefont
  {Hamma}},\ }\href {\doibase 10.3390/e23081073} {\bibfield  {journal}
  {\bibinfo  {journal} {Entropy}\ }\textbf {\bibinfo {volume} {23}} (\bibinfo
  {year} {2021}{\natexlab{a}}),\ 10.3390/e23081073}\BibitemShut {NoStop}%
\bibitem [{\citenamefont {Roberts}\ \emph {et~al.}(2015)\citenamefont
  {Roberts}, \citenamefont {Stanford},\ and\ \citenamefont
  {Susskind}}]{Roberts2015shocks}%
  \BibitemOpen
  \bibfield  {author} {\bibinfo {author} {\bibfnamefont {D.~A.}\ \bibnamefont
  {Roberts}}, \bibinfo {author} {\bibfnamefont {D.}~\bibnamefont {Stanford}}, \
  and\ \bibinfo {author} {\bibfnamefont {L.}~\bibnamefont {Susskind}},\ }\href
  {\doibase 10.1007/JHEP03(2015)051} {\bibfield  {journal} {\bibinfo  {journal}
  {Journal of High Energy Physics}\ }\textbf {\bibinfo {volume} {2015}},\
  \bibinfo {pages} {51} (\bibinfo {year} {2015})}\BibitemShut {NoStop}%
\bibitem [{\citenamefont {Pilatowsky-Cameo}\ \emph {et~al.}(2020)\citenamefont
  {Pilatowsky-Cameo}, \citenamefont {Ch\'avez-Carlos}, \citenamefont
  {Bastarrachea-Magnani}, \citenamefont {Str\'ansk\'y}, \citenamefont
  {Lerma-Hern\'andez}, \citenamefont {Santos},\ and\ \citenamefont
  {Hirsch}}]{Pilatowsky2020}%
  \BibitemOpen
  \bibfield  {author} {\bibinfo {author} {\bibfnamefont {S.}~\bibnamefont
  {Pilatowsky-Cameo}}, \bibinfo {author} {\bibfnamefont {J.}~\bibnamefont
  {Ch\'avez-Carlos}}, \bibinfo {author} {\bibfnamefont {M.~A.}\ \bibnamefont
  {Bastarrachea-Magnani}}, \bibinfo {author} {\bibfnamefont {P.}~\bibnamefont
  {Str\'ansk\'y}}, \bibinfo {author} {\bibfnamefont {S.}~\bibnamefont
  {Lerma-Hern\'andez}}, \bibinfo {author} {\bibfnamefont {L.~F.}\ \bibnamefont
  {Santos}}, \ and\ \bibinfo {author} {\bibfnamefont {J.~G.}\ \bibnamefont
  {Hirsch}},\ }\href {\doibase 10.1103/PhysRevE.101.010202} {\bibfield
  {journal} {\bibinfo  {journal} {Phys. Rev. E}\ }\textbf {\bibinfo {volume}
  {101}},\ \bibinfo {pages} {010202(R)} (\bibinfo {year} {2020})}\BibitemShut
  {NoStop}%
\bibitem [{\citenamefont {Xu}\ \emph {et~al.}(2020)\citenamefont {Xu},
  \citenamefont {Scaffidi},\ and\ \citenamefont {Cao}}]{Xu2020}%
  \BibitemOpen
  \bibfield  {author} {\bibinfo {author} {\bibfnamefont {T.}~\bibnamefont
  {Xu}}, \bibinfo {author} {\bibfnamefont {T.}~\bibnamefont {Scaffidi}}, \ and\
  \bibinfo {author} {\bibfnamefont {X.}~\bibnamefont {Cao}},\ }\href {\doibase
  10.1103/PhysRevLett.124.140602} {\bibfield  {journal} {\bibinfo  {journal}
  {Phys. Rev. Lett.}\ }\textbf {\bibinfo {volume} {124}},\ \bibinfo {pages}
  {140602} (\bibinfo {year} {2020})}\BibitemShut {NoStop}%
\bibitem [{\citenamefont {Dowling}\ \emph
  {et~al.}(2023{\natexlab{a}})\citenamefont {Dowling}, \citenamefont {Kos},\
  and\ \citenamefont {Modi}}]{dowling2023scrambling}%
  \BibitemOpen
  \bibfield  {author} {\bibinfo {author} {\bibfnamefont {N.}~\bibnamefont
  {Dowling}}, \bibinfo {author} {\bibfnamefont {P.}~\bibnamefont {Kos}}, \ and\
  \bibinfo {author} {\bibfnamefont {K.}~\bibnamefont {Modi}},\ }\href {\doibase
  10.1103/PhysRevLett.131.180403} {\bibfield  {journal} {\bibinfo  {journal}
  {Phys. Rev. Lett.}\ }\textbf {\bibinfo {volume} {131}},\ \bibinfo {pages}
  {180403} (\bibinfo {year} {2023}{\natexlab{a}})}\BibitemShut {NoStop}%
\bibitem [{\citenamefont {Pollock}\ \emph
  {et~al.}(2018{\natexlab{a}})\citenamefont {Pollock}, \citenamefont
  {Rodr\'{\i}guez-Rosario}, \citenamefont {Frauenheim}, \citenamefont
  {Paternostro},\ and\ \citenamefont {Modi}}]{processtensor}%
  \BibitemOpen
  \bibfield  {author} {\bibinfo {author} {\bibfnamefont {F.~A.}\ \bibnamefont
  {Pollock}}, \bibinfo {author} {\bibfnamefont {C.}~\bibnamefont
  {Rodr\'{\i}guez-Rosario}}, \bibinfo {author} {\bibfnamefont {T.}~\bibnamefont
  {Frauenheim}}, \bibinfo {author} {\bibfnamefont {M.}~\bibnamefont
  {Paternostro}}, \ and\ \bibinfo {author} {\bibfnamefont {K.}~\bibnamefont
  {Modi}},\ }\href {\doibase 10.1103/PhysRevA.97.012127} {\bibfield  {journal}
  {\bibinfo  {journal} {Phys. Rev. A}\ }\textbf {\bibinfo {volume} {97}},\
  \bibinfo {pages} {012127} (\bibinfo {year} {2018}{\natexlab{a}})}\BibitemShut
  {NoStop}%
\bibitem [{\citenamefont {Pollock}\ \emph
  {et~al.}(2018{\natexlab{b}})\citenamefont {Pollock}, \citenamefont
  {Rodr\'{\i}guez-Rosario}, \citenamefont {Frauenheim}, \citenamefont
  {Paternostro},\ and\ \citenamefont {Modi}}]{processtensor2}%
  \BibitemOpen
  \bibfield  {author} {\bibinfo {author} {\bibfnamefont {F.~A.}\ \bibnamefont
  {Pollock}}, \bibinfo {author} {\bibfnamefont {C.}~\bibnamefont
  {Rodr\'{\i}guez-Rosario}}, \bibinfo {author} {\bibfnamefont {T.}~\bibnamefont
  {Frauenheim}}, \bibinfo {author} {\bibfnamefont {M.}~\bibnamefont
  {Paternostro}}, \ and\ \bibinfo {author} {\bibfnamefont {K.}~\bibnamefont
  {Modi}},\ }\href {\doibase 10.1103/PhysRevLett.120.040405} {\bibfield
  {journal} {\bibinfo  {journal} {Phys. Rev. Lett.}\ }\textbf {\bibinfo
  {volume} {120}},\ \bibinfo {pages} {040405} (\bibinfo {year}
  {2018}{\natexlab{b}})}\BibitemShut {NoStop}%
\bibitem [{\citenamefont {Milz}\ and\ \citenamefont
  {Modi}(2021)}]{milz2020quantum}%
  \BibitemOpen
  \bibfield  {author} {\bibinfo {author} {\bibfnamefont {S.}~\bibnamefont
  {Milz}}\ and\ \bibinfo {author} {\bibfnamefont {K.}~\bibnamefont {Modi}},\
  }\href {\doibase 10.1103/PRXQuantum.2.030201} {\bibfield  {journal} {\bibinfo
   {journal} {PRX Quantum}\ }\textbf {\bibinfo {volume} {2}},\ \bibinfo {pages}
  {030201} (\bibinfo {year} {2021})}\BibitemShut {NoStop}%
\bibitem [{\citenamefont {Milz}\ \emph {et~al.}(2017)\citenamefont {Milz},
  \citenamefont {Pollock},\ and\ \citenamefont {Modi}}]{OperationalQDynamics}%
  \BibitemOpen
  \bibfield  {author} {\bibinfo {author} {\bibfnamefont {S.}~\bibnamefont
  {Milz}}, \bibinfo {author} {\bibfnamefont {F.~A.}\ \bibnamefont {Pollock}}, \
  and\ \bibinfo {author} {\bibfnamefont {K.}~\bibnamefont {Modi}},\ }\href
  {\doibase 10.1142/S1230161217400169} {\bibfield  {journal} {\bibinfo
  {journal} {Open Syst. Inf. Dyn.}\ }\textbf {\bibinfo {volume} {24}},\
  \bibinfo {pages} {1740016} (\bibinfo {year} {2017})}\BibitemShut {NoStop}%
\bibitem [{\citenamefont {Chiribella}\ \emph {et~al.}(2008)\citenamefont
  {Chiribella}, \citenamefont {D'Ariano},\ and\ \citenamefont
  {Perinotti}}]{Chiribella2008}%
  \BibitemOpen
  \bibfield  {author} {\bibinfo {author} {\bibfnamefont {G.}~\bibnamefont
  {Chiribella}}, \bibinfo {author} {\bibfnamefont {G.~M.}\ \bibnamefont
  {D'Ariano}}, \ and\ \bibinfo {author} {\bibfnamefont {P.}~\bibnamefont
  {Perinotti}},\ }\href {\doibase 10.1103/PhysRevLett.101.060401} {\bibfield
  {journal} {\bibinfo  {journal} {Phys. Rev. Lett.}\ }\textbf {\bibinfo
  {volume} {101}},\ \bibinfo {pages} {060401} (\bibinfo {year}
  {2008})}\BibitemShut {NoStop}%
\bibitem [{\citenamefont {Watrous}(2018)}]{watrous_2018}%
  \BibitemOpen
  \bibfield  {author} {\bibinfo {author} {\bibfnamefont {J.}~\bibnamefont
  {Watrous}},\ }\href {\doibase 10.1017/9781316848142} {\emph {\bibinfo {title}
  {The Theory of Quantum Information}}}\ (\bibinfo  {publisher} {Cambridge
  University Press},\ \bibinfo {year} {2018})\BibitemShut {NoStop}%
\bibitem [{\citenamefont {Costa}\ and\ \citenamefont
  {Shrapnel}(2016)}]{Costa2016}%
  \BibitemOpen
  \bibfield  {author} {\bibinfo {author} {\bibfnamefont {F.}~\bibnamefont
  {Costa}}\ and\ \bibinfo {author} {\bibfnamefont {S.}~\bibnamefont
  {Shrapnel}},\ }\href {\doibase 10.1088/1367-2630/18/6/063032} {\bibfield
  {journal} {\bibinfo  {journal} {New J. Phys.}\ }\textbf {\bibinfo {volume}
  {18}},\ \bibinfo {pages} {063032} (\bibinfo {year} {2016})}\BibitemShut
  {NoStop}%
\bibitem [{\citenamefont {Taranto}\ \emph
  {et~al.}(2019{\natexlab{a}})\citenamefont {Taranto}, \citenamefont {Pollock},
  \citenamefont {Milz}, \citenamefont {Tomamichel},\ and\ \citenamefont
  {Modi}}]{Markovorder1}%
  \BibitemOpen
  \bibfield  {author} {\bibinfo {author} {\bibfnamefont {P.}~\bibnamefont
  {Taranto}}, \bibinfo {author} {\bibfnamefont {F.~A.}\ \bibnamefont
  {Pollock}}, \bibinfo {author} {\bibfnamefont {S.}~\bibnamefont {Milz}},
  \bibinfo {author} {\bibfnamefont {M.}~\bibnamefont {Tomamichel}}, \ and\
  \bibinfo {author} {\bibfnamefont {K.}~\bibnamefont {Modi}},\ }\href {\doibase
  10.1103/PhysRevLett.122.140401} {\bibfield  {journal} {\bibinfo  {journal}
  {Phys. Rev. Lett.}\ }\textbf {\bibinfo {volume} {122}},\ \bibinfo {pages}
  {140401} (\bibinfo {year} {2019}{\natexlab{a}})}\BibitemShut {NoStop}%
\bibitem [{\citenamefont {Taranto}\ \emph
  {et~al.}(2019{\natexlab{b}})\citenamefont {Taranto}, \citenamefont {Milz},
  \citenamefont {Pollock},\ and\ \citenamefont
  {Modi}}]{Taranto2019FiniteMarkov}%
  \BibitemOpen
  \bibfield  {author} {\bibinfo {author} {\bibfnamefont {P.}~\bibnamefont
  {Taranto}}, \bibinfo {author} {\bibfnamefont {S.}~\bibnamefont {Milz}},
  \bibinfo {author} {\bibfnamefont {F.~A.}\ \bibnamefont {Pollock}}, \ and\
  \bibinfo {author} {\bibfnamefont {K.}~\bibnamefont {Modi}},\ }\href {\doibase
  10.1103/PhysRevA.99.042108} {\bibfield  {journal} {\bibinfo  {journal} {Phys.
  Rev. A}\ }\textbf {\bibinfo {volume} {99}},\ \bibinfo {pages} {042108}
  (\bibinfo {year} {2019}{\natexlab{b}})}\BibitemShut {NoStop}%
\bibitem [{\citenamefont {Dowling}\ \emph
  {et~al.}(2023{\natexlab{b}})\citenamefont {Dowling}, \citenamefont
  {Figueroa-Romero}, \citenamefont {Pollock}, \citenamefont {Strasberg},\ and\
  \citenamefont {Modi}}]{Dowling2021}%
  \BibitemOpen
  \bibfield  {author} {\bibinfo {author} {\bibfnamefont {N.}~\bibnamefont
  {Dowling}}, \bibinfo {author} {\bibfnamefont {P.}~\bibnamefont
  {Figueroa-Romero}}, \bibinfo {author} {\bibfnamefont {F.~A.}\ \bibnamefont
  {Pollock}}, \bibinfo {author} {\bibfnamefont {P.}~\bibnamefont {Strasberg}},
  \ and\ \bibinfo {author} {\bibfnamefont {K.}~\bibnamefont {Modi}},\ }\href
  {\doibase 10.22331/q-2023-06-01-1027} {\bibfield  {journal} {\bibinfo
  {journal} {{Quantum}}\ }\textbf {\bibinfo {volume} {7}},\ \bibinfo {pages}
  {1027} (\bibinfo {year} {2023}{\natexlab{b}})}\BibitemShut {NoStop}%
\bibitem [{\citenamefont {Dowling}\ \emph
  {et~al.}(2023{\natexlab{c}})\citenamefont {Dowling}, \citenamefont
  {Figueroa-Romero}, \citenamefont {Pollock}, \citenamefont {Strasberg},\ and\
  \citenamefont {Modi}}]{finitetime}%
  \BibitemOpen
  \bibfield  {author} {\bibinfo {author} {\bibfnamefont {N.}~\bibnamefont
  {Dowling}}, \bibinfo {author} {\bibfnamefont {P.}~\bibnamefont
  {Figueroa-Romero}}, \bibinfo {author} {\bibfnamefont {F.~A.}\ \bibnamefont
  {Pollock}}, \bibinfo {author} {\bibfnamefont {P.}~\bibnamefont {Strasberg}},
  \ and\ \bibinfo {author} {\bibfnamefont {K.}~\bibnamefont {Modi}},\ }\href
  {\doibase 10.21468/SciPostPhysCore.6.2.043} {\bibfield  {journal} {\bibinfo
  {journal} {SciPost Phys. Core}\ }\textbf {\bibinfo {volume} {6}},\ \bibinfo
  {pages} {043} (\bibinfo {year} {2023}{\natexlab{c}})}\BibitemShut {NoStop}%
\bibitem [{\citenamefont {Strasberg}\ and\ \citenamefont
  {D\'{\i}az}(2019)}]{strasberg2019}%
  \BibitemOpen
  \bibfield  {author} {\bibinfo {author} {\bibfnamefont {P.}~\bibnamefont
  {Strasberg}}\ and\ \bibinfo {author} {\bibfnamefont {M.~G.}\ \bibnamefont
  {D\'{\i}az}},\ }\href {\doibase 10.1103/PhysRevA.100.022120} {\bibfield
  {journal} {\bibinfo  {journal} {Phys. Rev. A}\ }\textbf {\bibinfo {volume}
  {100}},\ \bibinfo {pages} {022120} (\bibinfo {year} {2019})}\BibitemShut
  {NoStop}%
\bibitem [{\citenamefont {Milz}\ \emph
  {et~al.}(2020{\natexlab{a}})\citenamefont {Milz}, \citenamefont {Egloff},
  \citenamefont {Taranto}, \citenamefont {Theurer}, \citenamefont {Plenio},
  \citenamefont {Smirne},\ and\ \citenamefont {Huelga}}]{Milz2020prx}%
  \BibitemOpen
  \bibfield  {author} {\bibinfo {author} {\bibfnamefont {S.}~\bibnamefont
  {Milz}}, \bibinfo {author} {\bibfnamefont {D.}~\bibnamefont {Egloff}},
  \bibinfo {author} {\bibfnamefont {P.}~\bibnamefont {Taranto}}, \bibinfo
  {author} {\bibfnamefont {T.}~\bibnamefont {Theurer}}, \bibinfo {author}
  {\bibfnamefont {M.~B.}\ \bibnamefont {Plenio}}, \bibinfo {author}
  {\bibfnamefont {A.}~\bibnamefont {Smirne}}, \ and\ \bibinfo {author}
  {\bibfnamefont {S.~F.}\ \bibnamefont {Huelga}},\ }\href {\doibase
  10.1103/PhysRevX.10.041049} {\bibfield  {journal} {\bibinfo  {journal} {Phys.
  Rev. X}\ }\textbf {\bibinfo {volume} {10}},\ \bibinfo {pages} {041049}
  (\bibinfo {year} {2020}{\natexlab{a}})}\BibitemShut {NoStop}%
\bibitem [{\citenamefont {Strasberg}\ \emph {et~al.}(2022)\citenamefont
  {Strasberg}, \citenamefont {Winter}, \citenamefont {Gemmer},\ and\
  \citenamefont {Wang}}]{Strasberg2022}%
  \BibitemOpen
  \bibfield  {author} {\bibinfo {author} {\bibfnamefont {P.}~\bibnamefont
  {Strasberg}}, \bibinfo {author} {\bibfnamefont {A.}~\bibnamefont {Winter}},
  \bibinfo {author} {\bibfnamefont {J.}~\bibnamefont {Gemmer}}, \ and\ \bibinfo
  {author} {\bibfnamefont {J.}~\bibnamefont {Wang}},\ }\href {\doibase
  10.48550/ARXIV.2209.07977} {\enquote {\bibinfo {title} {Classicality,
  markovianity and local detailed balance from pure state dynamics},}\ }
  (\bibinfo {year} {2022})\BibitemShut {NoStop}%
\bibitem [{\citenamefont {Milz}\ \emph {et~al.}(2021)\citenamefont {Milz},
  \citenamefont {Spee}, \citenamefont {Xu}, \citenamefont {Pollock},
  \citenamefont {Modi},\ and\ \citenamefont {Gühne}}]{Milz2021GME}%
  \BibitemOpen
  \bibfield  {author} {\bibinfo {author} {\bibfnamefont {S.}~\bibnamefont
  {Milz}}, \bibinfo {author} {\bibfnamefont {C.}~\bibnamefont {Spee}}, \bibinfo
  {author} {\bibfnamefont {Z.-P.}\ \bibnamefont {Xu}}, \bibinfo {author}
  {\bibfnamefont {F.~A.}\ \bibnamefont {Pollock}}, \bibinfo {author}
  {\bibfnamefont {K.}~\bibnamefont {Modi}}, \ and\ \bibinfo {author}
  {\bibfnamefont {O.}~\bibnamefont {Gühne}},\ }\href {\doibase
  10.21468/SciPostPhys.10.6.141} {\bibfield  {journal} {\bibinfo  {journal}
  {SciPost Phys.}\ }\textbf {\bibinfo {volume} {10}},\ \bibinfo {pages} {141}
  (\bibinfo {year} {2021})}\BibitemShut {NoStop}%
\bibitem [{\citenamefont {White}\ \emph {et~al.}(2021)\citenamefont {White},
  \citenamefont {Pollock}, \citenamefont {Hollenberg}, \citenamefont {Hill},\
  and\ \citenamefont {Modi}}]{White2021}%
  \BibitemOpen
  \bibfield  {author} {\bibinfo {author} {\bibfnamefont {G.~A.~L.}\
  \bibnamefont {White}}, \bibinfo {author} {\bibfnamefont {F.~A.}\ \bibnamefont
  {Pollock}}, \bibinfo {author} {\bibfnamefont {L.~C.~L.}\ \bibnamefont
  {Hollenberg}}, \bibinfo {author} {\bibfnamefont {C.~D.}\ \bibnamefont
  {Hill}}, \ and\ \bibinfo {author} {\bibfnamefont {K.}~\bibnamefont {Modi}},\
  }\href@noop {} {\enquote {\bibinfo {title} {From many-body to many-time
  physics},}\ } (\bibinfo {year} {2021}),\ \Eprint
  {http://arxiv.org/abs/2107.13934} {arXiv:2107.13934 [quant-ph]} \BibitemShut
  {NoStop}%
\bibitem [{\citenamefont {Eisert}\ \emph {et~al.}(2010)\citenamefont {Eisert},
  \citenamefont {Cramer},\ and\ \citenamefont
  {Plenio}}]{Eisert_Cramer_Plenio_2010}%
  \BibitemOpen
  \bibfield  {author} {\bibinfo {author} {\bibfnamefont {J.}~\bibnamefont
  {Eisert}}, \bibinfo {author} {\bibfnamefont {M.}~\bibnamefont {Cramer}}, \
  and\ \bibinfo {author} {\bibfnamefont {M.~B.}\ \bibnamefont {Plenio}},\
  }\href {\doibase 10.1103/RevModPhys.82.277} {\bibfield  {journal} {\bibinfo
  {journal} {Rev. Mod. Phys.}\ }\textbf {\bibinfo {volume} {82}},\ \bibinfo
  {pages} {277} (\bibinfo {year} {2010})}\BibitemShut {NoStop}%
\bibitem [{\citenamefont {Fannes}\ \emph {et~al.}(1992)\citenamefont {Fannes},
  \citenamefont {Nachtergaele},\ and\ \citenamefont {Werner}}]{Fannes1992}%
  \BibitemOpen
  \bibfield  {author} {\bibinfo {author} {\bibfnamefont {M.}~\bibnamefont
  {Fannes}}, \bibinfo {author} {\bibfnamefont {B.}~\bibnamefont
  {Nachtergaele}}, \ and\ \bibinfo {author} {\bibfnamefont {R.~F.}\
  \bibnamefont {Werner}},\ }\href {\doibase 10.1007/BF02099178} {\bibfield
  {journal} {\bibinfo  {journal} {Communications in Mathematical Physics}\
  }\textbf {\bibinfo {volume} {144}},\ \bibinfo {pages} {443} (\bibinfo {year}
  {1992})}\BibitemShut {NoStop}%
\bibitem [{\citenamefont {Schuch}\ \emph {et~al.}(2008)\citenamefont {Schuch},
  \citenamefont {Wolf}, \citenamefont {Verstraete},\ and\ \citenamefont
  {Cirac}}]{Schuch2008MPS}%
  \BibitemOpen
  \bibfield  {author} {\bibinfo {author} {\bibfnamefont {N.}~\bibnamefont
  {Schuch}}, \bibinfo {author} {\bibfnamefont {M.~M.}\ \bibnamefont {Wolf}},
  \bibinfo {author} {\bibfnamefont {F.}~\bibnamefont {Verstraete}}, \ and\
  \bibinfo {author} {\bibfnamefont {J.~I.}\ \bibnamefont {Cirac}},\ }\href
  {\doibase 10.1103/PhysRevLett.100.030504} {\bibfield  {journal} {\bibinfo
  {journal} {Phys. Rev. Lett.}\ }\textbf {\bibinfo {volume} {100}},\ \bibinfo
  {pages} {030504} (\bibinfo {year} {2008})}\BibitemShut {NoStop}%
\bibitem [{\citenamefont {Bridgeman}\ and\ \citenamefont
  {Chubb}(2017)}]{Bridgeman_2017}%
  \BibitemOpen
  \bibfield  {author} {\bibinfo {author} {\bibfnamefont {J.~C.}\ \bibnamefont
  {Bridgeman}}\ and\ \bibinfo {author} {\bibfnamefont {C.~T.}\ \bibnamefont
  {Chubb}},\ }\href {\doibase 10.1088/1751-8121/aa6dc3} {\bibfield  {journal}
  {\bibinfo  {journal} {Journal of Physics A: Mathematical and Theoretical}\
  }\textbf {\bibinfo {volume} {50}},\ \bibinfo {pages} {223001} (\bibinfo
  {year} {2017})}\BibitemShut {NoStop}%
\bibitem [{\citenamefont {Vidal}\ \emph {et~al.}(2003)\citenamefont {Vidal},
  \citenamefont {Latorre}, \citenamefont {Rico},\ and\ \citenamefont
  {Kitaev}}]{Vidal20003}%
  \BibitemOpen
  \bibfield  {author} {\bibinfo {author} {\bibfnamefont {G.}~\bibnamefont
  {Vidal}}, \bibinfo {author} {\bibfnamefont {J.~I.}\ \bibnamefont {Latorre}},
  \bibinfo {author} {\bibfnamefont {E.}~\bibnamefont {Rico}}, \ and\ \bibinfo
  {author} {\bibfnamefont {A.}~\bibnamefont {Kitaev}},\ }\href {\doibase
  10.1103/PhysRevLett.90.227902} {\bibfield  {journal} {\bibinfo  {journal}
  {Phys. Rev. Lett.}\ }\textbf {\bibinfo {volume} {90}},\ \bibinfo {pages}
  {227902} (\bibinfo {year} {2003})}\BibitemShut {NoStop}%
\bibitem [{Note2()}]{Note2}%
  \BibitemOpen
  \bibinfo {note} {This is in reference to its introduction in Ref.~\cite
  {Lindblad1986chaos}, and the chaotic classical analogy of the Bernoulli
  shift.}\BibitemShut {Stop}%
\bibitem [{\citenamefont {Cotler}\ \emph {et~al.}(2018)\citenamefont {Cotler},
  \citenamefont {Jian}, \citenamefont {Qi},\ and\ \citenamefont
  {Wilczek}}]{Cotler2018}%
  \BibitemOpen
  \bibfield  {author} {\bibinfo {author} {\bibfnamefont {J.}~\bibnamefont
  {Cotler}}, \bibinfo {author} {\bibfnamefont {C.-M.}\ \bibnamefont {Jian}},
  \bibinfo {author} {\bibfnamefont {X.-L.}\ \bibnamefont {Qi}}, \ and\ \bibinfo
  {author} {\bibfnamefont {F.}~\bibnamefont {Wilczek}},\ }\href {\doibase
  10.1007/JHEP09(2018)093} {\bibfield  {journal} {\bibinfo  {journal} {J. High
  Energy Phys.}\ }\textbf {\bibinfo {volume} {2018}},\ \bibinfo {pages} {93}
  (\bibinfo {year} {2018})}\BibitemShut {NoStop}%
\bibitem [{\citenamefont {Lieb}\ and\ \citenamefont
  {Robinson}(1972)}]{lieb_finite_1972}%
  \BibitemOpen
  \bibfield  {author} {\bibinfo {author} {\bibfnamefont {E.~H.}\ \bibnamefont
  {Lieb}}\ and\ \bibinfo {author} {\bibfnamefont {D.~W.}\ \bibnamefont
  {Robinson}},\ }\href {\doibase 10.1007/BF01645779} {\bibfield  {journal}
  {\bibinfo  {journal} {Communications in Mathematical Physics}\ }\textbf
  {\bibinfo {volume} {28}},\ \bibinfo {pages} {251} (\bibinfo {year}
  {1972})}\BibitemShut {NoStop}%
\bibitem [{\citenamefont {Chen}\ \emph {et~al.}(2023)\citenamefont {Chen},
  \citenamefont {Lucas},\ and\ \citenamefont {Yin}}]{chen2023speed}%
  \BibitemOpen
  \bibfield  {author} {\bibinfo {author} {\bibfnamefont {C.-F.}\ \bibnamefont
  {Chen}}, \bibinfo {author} {\bibfnamefont {A.}~\bibnamefont {Lucas}}, \ and\
  \bibinfo {author} {\bibfnamefont {C.}~\bibnamefont {Yin}},\ }\href {\doibase
  10.48550/arXiv.2303.07386} {\enquote {\bibinfo {title} {Speed limits and
  locality in many-body quantum dynamics},}\ } (\bibinfo {year} {2023}),\
  \Eprint {http://arxiv.org/abs/2303.07386} {arXiv:2303.07386 [quant-ph]}
  \BibitemShut {NoStop}%
\bibitem [{\citenamefont {von Keyserlingk}\ \emph {et~al.}(2018)\citenamefont
  {von Keyserlingk}, \citenamefont {Rakovszky}, \citenamefont {Pollmann},\ and\
  \citenamefont {Sondhi}}]{PhysRevX.8.021013}%
  \BibitemOpen
  \bibfield  {author} {\bibinfo {author} {\bibfnamefont {C.~W.}\ \bibnamefont
  {von Keyserlingk}}, \bibinfo {author} {\bibfnamefont {T.}~\bibnamefont
  {Rakovszky}}, \bibinfo {author} {\bibfnamefont {F.}~\bibnamefont {Pollmann}},
  \ and\ \bibinfo {author} {\bibfnamefont {S.~L.}\ \bibnamefont {Sondhi}},\
  }\href {\doibase 10.1103/PhysRevX.8.021013} {\bibfield  {journal} {\bibinfo
  {journal} {Phys. Rev. X}\ }\textbf {\bibinfo {volume} {8}},\ \bibinfo {pages}
  {021013} (\bibinfo {year} {2018})}\BibitemShut {NoStop}%
\bibitem [{\citenamefont {Nahum}\ \emph {et~al.}(2017)\citenamefont {Nahum},
  \citenamefont {Ruhman}, \citenamefont {Vijay},\ and\ \citenamefont
  {Haah}}]{PhysRevX.7.031016}%
  \BibitemOpen
  \bibfield  {author} {\bibinfo {author} {\bibfnamefont {A.}~\bibnamefont
  {Nahum}}, \bibinfo {author} {\bibfnamefont {J.}~\bibnamefont {Ruhman}},
  \bibinfo {author} {\bibfnamefont {S.}~\bibnamefont {Vijay}}, \ and\ \bibinfo
  {author} {\bibfnamefont {J.}~\bibnamefont {Haah}},\ }\href {\doibase
  10.1103/PhysRevX.7.031016} {\bibfield  {journal} {\bibinfo  {journal} {Phys.
  Rev. X}\ }\textbf {\bibinfo {volume} {7}},\ \bibinfo {pages} {031016}
  (\bibinfo {year} {2017})}\BibitemShut {NoStop}%
\bibitem [{\citenamefont {Roberts}\ and\ \citenamefont
  {Swingle}(2016)}]{Swingle2016LR}%
  \BibitemOpen
  \bibfield  {author} {\bibinfo {author} {\bibfnamefont {D.~A.}\ \bibnamefont
  {Roberts}}\ and\ \bibinfo {author} {\bibfnamefont {B.}~\bibnamefont
  {Swingle}},\ }\href {\doibase 10.1103/PhysRevLett.117.091602} {\bibfield
  {journal} {\bibinfo  {journal} {Phys. Rev. Lett.}\ }\textbf {\bibinfo
  {volume} {117}},\ \bibinfo {pages} {091602} (\bibinfo {year}
  {2016})}\BibitemShut {NoStop}%
\bibitem [{\citenamefont {Figueroa-Romero}\ \emph {et~al.}(2019)\citenamefont
  {Figueroa-Romero}, \citenamefont {Modi},\ and\ \citenamefont
  {Pollock}}]{FigueroaRomero_Modi_Pollock_2019}%
  \BibitemOpen
  \bibfield  {author} {\bibinfo {author} {\bibfnamefont {P.}~\bibnamefont
  {Figueroa-Romero}}, \bibinfo {author} {\bibfnamefont {K.}~\bibnamefont
  {Modi}}, \ and\ \bibinfo {author} {\bibfnamefont {F.~A.}\ \bibnamefont
  {Pollock}},\ }\href {\doibase 10.22331/q-2019-04-30-136} {\bibfield
  {journal} {\bibinfo  {journal} {Quantum}\ }\textbf {\bibinfo {volume} {3}},\
  \bibinfo {pages} {136} (\bibinfo {year} {2019})}\BibitemShut {NoStop}%
\bibitem [{\citenamefont {Figueroa-Romero}\ \emph {et~al.}(2021)\citenamefont
  {Figueroa-Romero}, \citenamefont {Pollock},\ and\ \citenamefont
  {Modi}}]{FigueroaRomero_Pollock_Modi_2021}%
  \BibitemOpen
  \bibfield  {author} {\bibinfo {author} {\bibfnamefont {P.}~\bibnamefont
  {Figueroa-Romero}}, \bibinfo {author} {\bibfnamefont {F.~A.}\ \bibnamefont
  {Pollock}}, \ and\ \bibinfo {author} {\bibfnamefont {K.}~\bibnamefont
  {Modi}},\ }\href {\doibase 10.1038/s42005-021-00629-w} {\bibfield  {journal}
  {\bibinfo  {journal} {Communications Physics}\ }\textbf {\bibinfo {volume}
  {4}},\ \bibinfo {pages} {127} (\bibinfo {year} {2021})}\BibitemShut {NoStop}%
\bibitem [{\citenamefont {Nielsen}\ \emph {et~al.}(2006)\citenamefont
  {Nielsen}, \citenamefont {Dowling}, \citenamefont {Gu},\ and\ \citenamefont
  {Doherty}}]{Nielsen2006}%
  \BibitemOpen
  \bibfield  {author} {\bibinfo {author} {\bibfnamefont {M.~A.}\ \bibnamefont
  {Nielsen}}, \bibinfo {author} {\bibfnamefont {M.~R.}\ \bibnamefont
  {Dowling}}, \bibinfo {author} {\bibfnamefont {M.}~\bibnamefont {Gu}}, \ and\
  \bibinfo {author} {\bibfnamefont {A.~C.}\ \bibnamefont {Doherty}},\ }\href
  {\doibase 10.1126/science.1121541} {\bibfield  {journal} {\bibinfo  {journal}
  {Science}\ }\textbf {\bibinfo {volume} {311}},\ \bibinfo {pages} {1133}
  (\bibinfo {year} {2006})},\ \Eprint
  {http://arxiv.org/abs/https://www.science.org/doi/pdf/10.1126/science.1121541}
  {https://www.science.org/doi/pdf/10.1126/science.1121541} \BibitemShut
  {NoStop}%
\bibitem [{\citenamefont {Susskind}(2014)}]{susskind}%
  \BibitemOpen
  \bibfield  {author} {\bibinfo {author} {\bibfnamefont {L.}~\bibnamefont
  {Susskind}},\ }\href {\doibase 10.48550/arXiv.1402.5674} {\enquote {\bibinfo
  {title} {Computational complexity and black hole horizons},}\ } (\bibinfo
  {year} {2014}),\ \Eprint {http://arxiv.org/abs/1402.5674} {arXiv:1402.5674
  [quant-ph]} \BibitemShut {NoStop}%
\bibitem [{\citenamefont {Bouland}\ \emph {et~al.}(2019)\citenamefont
  {Bouland}, \citenamefont {Fefferman},\ and\ \citenamefont
  {Vazirani}}]{bouland}%
  \BibitemOpen
  \bibfield  {author} {\bibinfo {author} {\bibfnamefont {A.}~\bibnamefont
  {Bouland}}, \bibinfo {author} {\bibfnamefont {B.}~\bibnamefont {Fefferman}},
  \ and\ \bibinfo {author} {\bibfnamefont {U.}~\bibnamefont {Vazirani}},\
  }\href {\doibase 10.48550/arXiv.1910.14646} {\enquote {\bibinfo {title}
  {Computational pseudorandomness, the wormhole growth paradox, and constraints
  on the ads/cft duality},}\ } (\bibinfo {year} {2019}),\ \Eprint
  {http://arxiv.org/abs/1910.14646} {arXiv:1910.14646 [quant-ph]} \BibitemShut
  {NoStop}%
\bibitem [{\citenamefont {Sekino}\ and\ \citenamefont
  {Susskind}(2008)}]{Sekino_Susskind_2008}%
  \BibitemOpen
  \bibfield  {author} {\bibinfo {author} {\bibfnamefont {Y.}~\bibnamefont
  {Sekino}}\ and\ \bibinfo {author} {\bibfnamefont {L.}~\bibnamefont
  {Susskind}},\ }\href {\doibase 10.1088/1126-6708/2008/10/065} {\bibfield
  {journal} {\bibinfo  {journal} {Journal of High Energy Physics}\ }\textbf
  {\bibinfo {volume} {2008}},\ \bibinfo {pages} {065} (\bibinfo {year}
  {2008})}\BibitemShut {NoStop}%
\bibitem [{\citenamefont {Shenker}\ and\ \citenamefont
  {Stanford}(2014{\natexlab{b}})}]{Shenker2014}%
  \BibitemOpen
  \bibfield  {author} {\bibinfo {author} {\bibfnamefont {S.~H.}\ \bibnamefont
  {Shenker}}\ and\ \bibinfo {author} {\bibfnamefont {D.}~\bibnamefont
  {Stanford}},\ }\href {\doibase 10.1007/JHEP12(2014)046} {\bibfield  {journal}
  {\bibinfo  {journal} {Journal of High Energy Physics}\ }\textbf {\bibinfo
  {volume} {2014}},\ \bibinfo {pages} {46} (\bibinfo {year}
  {2014}{\natexlab{b}})}\BibitemShut {NoStop}%
\bibitem [{Note3()}]{Note3}%
  \BibitemOpen
  \bibinfo {note} {This quantity is alternatively called fidelity decay or
  Loschmidt Echo in the literature. We take the middle ground here, to give
  credit to the historical role of Peres~\cite {Peres1984stab}, while remaining
  familiar to those who recognize this quantity as the latter.}\BibitemShut
  {Stop}%
\bibitem [{\citenamefont {Poulin}\ \emph {et~al.}(2004)\citenamefont {Poulin},
  \citenamefont {Blume-Kohout}, \citenamefont {Laflamme},\ and\ \citenamefont
  {Ollivier}}]{Poulin2004}%
  \BibitemOpen
  \bibfield  {author} {\bibinfo {author} {\bibfnamefont {D.}~\bibnamefont
  {Poulin}}, \bibinfo {author} {\bibfnamefont {R.}~\bibnamefont
  {Blume-Kohout}}, \bibinfo {author} {\bibfnamefont {R.}~\bibnamefont
  {Laflamme}}, \ and\ \bibinfo {author} {\bibfnamefont {H.}~\bibnamefont
  {Ollivier}},\ }\href {\doibase 10.1103/PhysRevLett.92.177906} {\bibfield
  {journal} {\bibinfo  {journal} {Physical review letters}\ }\textbf {\bibinfo
  {volume} {92}},\ \bibinfo {pages} {177906} (\bibinfo {year}
  {2004})}\BibitemShut {NoStop}%
\bibitem [{\citenamefont {Childs}\ \emph {et~al.}(2021)\citenamefont {Childs},
  \citenamefont {Su}, \citenamefont {Tran}, \citenamefont {Wiebe},\ and\
  \citenamefont {Zhu}}]{Childs2021}%
  \BibitemOpen
  \bibfield  {author} {\bibinfo {author} {\bibfnamefont {A.~M.}\ \bibnamefont
  {Childs}}, \bibinfo {author} {\bibfnamefont {Y.}~\bibnamefont {Su}}, \bibinfo
  {author} {\bibfnamefont {M.~C.}\ \bibnamefont {Tran}}, \bibinfo {author}
  {\bibfnamefont {N.}~\bibnamefont {Wiebe}}, \ and\ \bibinfo {author}
  {\bibfnamefont {S.}~\bibnamefont {Zhu}},\ }\href {\doibase
  10.1103/PhysRevX.11.011020} {\bibfield  {journal} {\bibinfo  {journal} {Phys.
  Rev. X}\ }\textbf {\bibinfo {volume} {11}},\ \bibinfo {pages} {011020}
  (\bibinfo {year} {2021})}\BibitemShut {NoStop}%
\bibitem [{\citenamefont {Lindblad}(1979)}]{Lindblad1979-tm}%
  \BibitemOpen
  \bibfield  {author} {\bibinfo {author} {\bibfnamefont {G.}~\bibnamefont
  {Lindblad}},\ }\href {\doibase 10.1007/BF01197883} {\bibfield  {journal}
  {\bibinfo  {journal} {Commun. Math. Phys.}\ }\textbf {\bibinfo {volume}
  {65}},\ \bibinfo {pages} {281} (\bibinfo {year} {1979})}\BibitemShut
  {NoStop}%
\bibitem [{\citenamefont {Milz}\ \emph
  {et~al.}(2020{\natexlab{b}})\citenamefont {Milz}, \citenamefont {Sakuldee},
  \citenamefont {Pollock},\ and\ \citenamefont
  {Modi}}]{Milz2020kolmogorovextension}%
  \BibitemOpen
  \bibfield  {author} {\bibinfo {author} {\bibfnamefont {S.}~\bibnamefont
  {Milz}}, \bibinfo {author} {\bibfnamefont {F.}~\bibnamefont {Sakuldee}},
  \bibinfo {author} {\bibfnamefont {F.~A.}\ \bibnamefont {Pollock}}, \ and\
  \bibinfo {author} {\bibfnamefont {K.}~\bibnamefont {Modi}},\ }\href {\doibase
  10.22331/q-2020-04-20-255} {\bibfield  {journal} {\bibinfo  {journal}
  {{Quantum}}\ }\textbf {\bibinfo {volume} {4}},\ \bibinfo {pages} {255}
  (\bibinfo {year} {2020}{\natexlab{b}})}\BibitemShut {NoStop}%
\bibitem [{\citenamefont {Jalabert}\ and\ \citenamefont
  {Pastawski}(2001)}]{Jalabert2001}%
  \BibitemOpen
  \bibfield  {author} {\bibinfo {author} {\bibfnamefont {R.~A.}\ \bibnamefont
  {Jalabert}}\ and\ \bibinfo {author} {\bibfnamefont {H.~M.}\ \bibnamefont
  {Pastawski}},\ }\href {\doibase 10.1103/PhysRevLett.86.2490} {\bibfield
  {journal} {\bibinfo  {journal} {Phys. Rev. Lett.}\ }\textbf {\bibinfo
  {volume} {86}},\ \bibinfo {pages} {2490} (\bibinfo {year}
  {2001})}\BibitemShut {NoStop}%
\bibitem [{\citenamefont {Cucchietti}\ \emph {et~al.}(2002)\citenamefont
  {Cucchietti}, \citenamefont {Lewenkopf}, \citenamefont {Mucciolo},
  \citenamefont {Pastawski},\ and\ \citenamefont {Vallejos}}]{Cucchietti2002}%
  \BibitemOpen
  \bibfield  {author} {\bibinfo {author} {\bibfnamefont {F.~M.}\ \bibnamefont
  {Cucchietti}}, \bibinfo {author} {\bibfnamefont {C.~H.}\ \bibnamefont
  {Lewenkopf}}, \bibinfo {author} {\bibfnamefont {E.~R.}\ \bibnamefont
  {Mucciolo}}, \bibinfo {author} {\bibfnamefont {H.~M.}\ \bibnamefont
  {Pastawski}}, \ and\ \bibinfo {author} {\bibfnamefont {R.~O.}\ \bibnamefont
  {Vallejos}},\ }\href {\doibase 10.1103/PhysRevE.65.046209} {\bibfield
  {journal} {\bibinfo  {journal} {Phys. Rev. E}\ }\textbf {\bibinfo {volume}
  {65}},\ \bibinfo {pages} {046209} (\bibinfo {year} {2002})}\BibitemShut
  {NoStop}%
\bibitem [{\citenamefont {Maldacena}\ \emph {et~al.}(2016)\citenamefont
  {Maldacena}, \citenamefont {Shenker},\ and\ \citenamefont
  {Stanford}}]{Maldacena_Shenker_Stanford_2016}%
  \BibitemOpen
  \bibfield  {author} {\bibinfo {author} {\bibfnamefont {J.}~\bibnamefont
  {Maldacena}}, \bibinfo {author} {\bibfnamefont {S.~H.}\ \bibnamefont
  {Shenker}}, \ and\ \bibinfo {author} {\bibfnamefont {D.}~\bibnamefont
  {Stanford}},\ }\href {\doibase 10.1007/JHEP08(2016)106} {\bibfield  {journal}
  {\bibinfo  {journal} {Journal of High Energy Physics}\ }\textbf {\bibinfo
  {volume} {2016}},\ \bibinfo {pages} {106} (\bibinfo {year}
  {2016})}\BibitemShut {NoStop}%
\bibitem [{\citenamefont {Foini}\ and\ \citenamefont
  {Kurchan}(2019)}]{Foini2019}%
  \BibitemOpen
  \bibfield  {author} {\bibinfo {author} {\bibfnamefont {L.}~\bibnamefont
  {Foini}}\ and\ \bibinfo {author} {\bibfnamefont {J.}~\bibnamefont
  {Kurchan}},\ }\href {\doibase 10.1103/PhysRevE.99.042139} {\bibfield
  {journal} {\bibinfo  {journal} {Phys. Rev. E}\ }\textbf {\bibinfo {volume}
  {99}},\ \bibinfo {pages} {042139} (\bibinfo {year} {2019})}\BibitemShut
  {NoStop}%
\bibitem [{\citenamefont {Parker}\ \emph {et~al.}(2019)\citenamefont {Parker},
  \citenamefont {Cao}, \citenamefont {Avdoshkin}, \citenamefont {Scaffidi},\
  and\ \citenamefont {Altman}}]{Parker2019}%
  \BibitemOpen
  \bibfield  {author} {\bibinfo {author} {\bibfnamefont {D.~E.}\ \bibnamefont
  {Parker}}, \bibinfo {author} {\bibfnamefont {X.}~\bibnamefont {Cao}},
  \bibinfo {author} {\bibfnamefont {A.}~\bibnamefont {Avdoshkin}}, \bibinfo
  {author} {\bibfnamefont {T.}~\bibnamefont {Scaffidi}}, \ and\ \bibinfo
  {author} {\bibfnamefont {E.}~\bibnamefont {Altman}},\ }\href {\doibase
  10.1103/PhysRevX.9.041017} {\bibfield  {journal} {\bibinfo  {journal} {Phys.
  Rev. X}\ }\textbf {\bibinfo {volume} {9}},\ \bibinfo {pages} {041017}
  (\bibinfo {year} {2019})}\BibitemShut {NoStop}%
\bibitem [{Note4()}]{Note4}%
  \BibitemOpen
  \bibinfo {note} {The classical Pesin's theorem states that the
  Kolmogorov-Sinai entropy is a lower bound of the sum of the positive Lyapunov
  exponents of a classical dynamical system~\cite {Pesin1977}.}\BibitemShut
  {Stop}%
\bibitem [{\citenamefont {Pi\ifmmode~\check{z}\else \v{z}\fi{}orn}\ and\
  \citenamefont {Prosen}(2009)}]{Prosen2009}%
  \BibitemOpen
  \bibfield  {author} {\bibinfo {author} {\bibfnamefont {I.}~\bibnamefont
  {Pi\ifmmode~\check{z}\else \v{z}\fi{}orn}}\ and\ \bibinfo {author}
  {\bibfnamefont {T.~c.~v.}\ \bibnamefont {Prosen}},\ }\href {\doibase
  10.1103/PhysRevB.79.184416} {\bibfield  {journal} {\bibinfo  {journal} {Phys.
  Rev. B}\ }\textbf {\bibinfo {volume} {79}},\ \bibinfo {pages} {184416}
  (\bibinfo {year} {2009})}\BibitemShut {NoStop}%
\bibitem [{\citenamefont {Muth}\ \emph {et~al.}(2011)\citenamefont {Muth},
  \citenamefont {Unanyan},\ and\ \citenamefont {Fleischhauer}}]{Muth2011}%
  \BibitemOpen
  \bibfield  {author} {\bibinfo {author} {\bibfnamefont {D.}~\bibnamefont
  {Muth}}, \bibinfo {author} {\bibfnamefont {R.~G.}\ \bibnamefont {Unanyan}}, \
  and\ \bibinfo {author} {\bibfnamefont {M.}~\bibnamefont {Fleischhauer}},\
  }\href {\doibase 10.1103/PhysRevLett.106.077202} {\bibfield  {journal}
  {\bibinfo  {journal} {Phys. Rev. Lett.}\ }\textbf {\bibinfo {volume} {106}},\
  \bibinfo {pages} {077202} (\bibinfo {year} {2011})}\BibitemShut {NoStop}%
\bibitem [{\citenamefont {Dubail}(2017)}]{Dubail_2017}%
  \BibitemOpen
  \bibfield  {author} {\bibinfo {author} {\bibfnamefont {J.}~\bibnamefont
  {Dubail}},\ }\href {\doibase 10.1088/1751-8121/aa6f38} {\bibfield  {journal}
  {\bibinfo  {journal} {Journal of Physics A: Mathematical and Theoretical}\
  }\textbf {\bibinfo {volume} {50}},\ \bibinfo {pages} {234001} (\bibinfo
  {year} {2017})}\BibitemShut {NoStop}%
\bibitem [{\citenamefont {Jonay}\ \emph {et~al.}(2018)\citenamefont {Jonay},
  \citenamefont {Huse},\ and\ \citenamefont {Nahum}}]{Jonay2018}%
  \BibitemOpen
  \bibfield  {author} {\bibinfo {author} {\bibfnamefont {C.}~\bibnamefont
  {Jonay}}, \bibinfo {author} {\bibfnamefont {D.~A.}\ \bibnamefont {Huse}}, \
  and\ \bibinfo {author} {\bibfnamefont {A.}~\bibnamefont {Nahum}},\ }\href
  {\doibase 10.48550/ARXIV.1803.00089} {\enquote {\bibinfo {title}
  {{Coarse-grained dynamics of operator and state entanglement}},}\ } (\bibinfo
  {year} {2018}),\ \Eprint {http://arxiv.org/abs/1803.00089} {arXiv:1803.00089
  [cond-mat.stat-mech]} \BibitemShut {NoStop}%
\bibitem [{\citenamefont {Alba}\ \emph {et~al.}(2019)\citenamefont {Alba},
  \citenamefont {Dubail},\ and\ \citenamefont {Medenjak}}]{Alba2019}%
  \BibitemOpen
  \bibfield  {author} {\bibinfo {author} {\bibfnamefont {V.}~\bibnamefont
  {Alba}}, \bibinfo {author} {\bibfnamefont {J.}~\bibnamefont {Dubail}}, \ and\
  \bibinfo {author} {\bibfnamefont {M.}~\bibnamefont {Medenjak}},\ }\href
  {\doibase 10.1103/PhysRevLett.122.250603} {\bibfield  {journal} {\bibinfo
  {journal} {Phys. Rev. Lett.}\ }\textbf {\bibinfo {volume} {122}},\ \bibinfo
  {pages} {250603} (\bibinfo {year} {2019})}\BibitemShut {NoStop}%
\bibitem [{\citenamefont {Alba}(2021)}]{Alba2021}%
  \BibitemOpen
  \bibfield  {author} {\bibinfo {author} {\bibfnamefont {V.}~\bibnamefont
  {Alba}},\ }\href {\doibase 10.1103/PhysRevB.104.094410} {\bibfield  {journal}
  {\bibinfo  {journal} {Phys. Rev. B}\ }\textbf {\bibinfo {volume} {104}},\
  \bibinfo {pages} {094410} (\bibinfo {year} {2021})}\BibitemShut {NoStop}%
\bibitem [{\citenamefont {Heyl}\ \emph {et~al.}(2019)\citenamefont {Heyl},
  \citenamefont {Hauke},\ and\ \citenamefont
  {Zoller}}]{Heyl_Hauke_Zoller_2019}%
  \BibitemOpen
  \bibfield  {author} {\bibinfo {author} {\bibfnamefont {M.}~\bibnamefont
  {Heyl}}, \bibinfo {author} {\bibfnamefont {P.}~\bibnamefont {Hauke}}, \ and\
  \bibinfo {author} {\bibfnamefont {P.}~\bibnamefont {Zoller}},\ }\href
  {\doibase 10.1126/sciadv.aau8342} {\bibfield  {journal} {\bibinfo  {journal}
  {Science Advances}\ }\textbf {\bibinfo {volume} {5}},\ \bibinfo {pages}
  {eaau8342} (\bibinfo {year} {2019})}\BibitemShut {NoStop}%
\bibitem [{\citenamefont {Santhanam}\ \emph {et~al.}(2022)\citenamefont
  {Santhanam}, \citenamefont {Paul},\ and\ \citenamefont
  {Kannan}}]{SANTHANAM20221}%
  \BibitemOpen
  \bibfield  {author} {\bibinfo {author} {\bibfnamefont {M.}~\bibnamefont
  {Santhanam}}, \bibinfo {author} {\bibfnamefont {S.}~\bibnamefont {Paul}}, \
  and\ \bibinfo {author} {\bibfnamefont {J.~B.}\ \bibnamefont {Kannan}},\
  }\href {\doibase https://doi.org/10.1016/j.physrep.2022.01.002} {\bibfield
  {journal} {\bibinfo  {journal} {Physics Reports}\ }\textbf {\bibinfo {volume}
  {956}},\ \bibinfo {pages} {1} (\bibinfo {year} {2022})},\ \bibinfo {note}
  {quantum kicked rotor and its variants: Chaos, localization and
  beyond}\BibitemShut {NoStop}%
\bibitem [{\citenamefont {Alet}\ and\ \citenamefont
  {Laflorencie}(2018)}]{ALET2018}%
  \BibitemOpen
  \bibfield  {author} {\bibinfo {author} {\bibfnamefont {F.}~\bibnamefont
  {Alet}}\ and\ \bibinfo {author} {\bibfnamefont {N.}~\bibnamefont
  {Laflorencie}},\ }\href {\doibase https://doi.org/10.1016/j.crhy.2018.03.003}
  {\bibfield  {journal} {\bibinfo  {journal} {Comptes Rendus Physique}\
  }\textbf {\bibinfo {volume} {19}},\ \bibinfo {pages} {498} (\bibinfo {year}
  {2018})},\ \bibinfo {note} {quantum simulation / Simulation
  quantique}\BibitemShut {NoStop}%
\bibitem [{\citenamefont {Geraedts}\ \emph {et~al.}(2016)\citenamefont
  {Geraedts}, \citenamefont {Nandkishore},\ and\ \citenamefont
  {Regnault}}]{Geraedts2016}%
  \BibitemOpen
  \bibfield  {author} {\bibinfo {author} {\bibfnamefont {S.~D.}\ \bibnamefont
  {Geraedts}}, \bibinfo {author} {\bibfnamefont {R.}~\bibnamefont
  {Nandkishore}}, \ and\ \bibinfo {author} {\bibfnamefont {N.}~\bibnamefont
  {Regnault}},\ }\href {\doibase 10.1103/PhysRevB.93.174202} {\bibfield
  {journal} {\bibinfo  {journal} {Physical review. B, Condensed matter}\
  }\textbf {\bibinfo {volume} {93}},\ \bibinfo {pages} {174202} (\bibinfo
  {year} {2016})}\BibitemShut {NoStop}%
\bibitem [{\citenamefont {Skinner}\ \emph {et~al.}(2019)\citenamefont
  {Skinner}, \citenamefont {Ruhman},\ and\ \citenamefont
  {Nahum}}]{Skinner2019}%
  \BibitemOpen
  \bibfield  {author} {\bibinfo {author} {\bibfnamefont {B.}~\bibnamefont
  {Skinner}}, \bibinfo {author} {\bibfnamefont {J.}~\bibnamefont {Ruhman}}, \
  and\ \bibinfo {author} {\bibfnamefont {A.}~\bibnamefont {Nahum}},\ }\href
  {\doibase 10.1103/PhysRevX.9.031009} {\bibfield  {journal} {\bibinfo
  {journal} {Phys. Rev. X}\ }\textbf {\bibinfo {volume} {9}},\ \bibinfo {pages}
  {031009} (\bibinfo {year} {2019})}\BibitemShut {NoStop}%
\bibitem [{\citenamefont {Zhang}\ \emph {et~al.}(2020)\citenamefont {Zhang},
  \citenamefont {Reyes}, \citenamefont {Kourtis}, \citenamefont {Chamon},
  \citenamefont {Mucciolo},\ and\ \citenamefont {Ruckenstein}}]{Zhang2020}%
  \BibitemOpen
  \bibfield  {author} {\bibinfo {author} {\bibfnamefont {L.}~\bibnamefont
  {Zhang}}, \bibinfo {author} {\bibfnamefont {J.~A.}\ \bibnamefont {Reyes}},
  \bibinfo {author} {\bibfnamefont {S.}~\bibnamefont {Kourtis}}, \bibinfo
  {author} {\bibfnamefont {C.}~\bibnamefont {Chamon}}, \bibinfo {author}
  {\bibfnamefont {E.~R.}\ \bibnamefont {Mucciolo}}, \ and\ \bibinfo {author}
  {\bibfnamefont {A.~E.}\ \bibnamefont {Ruckenstein}},\ }\href {\doibase
  10.1103/PhysRevB.101.235104} {\bibfield  {journal} {\bibinfo  {journal}
  {Phys. Rev. B}\ }\textbf {\bibinfo {volume} {101}},\ \bibinfo {pages}
  {235104} (\bibinfo {year} {2020})}\BibitemShut {NoStop}%
\bibitem [{\citenamefont {Poulin}\ \emph {et~al.}(2011)\citenamefont {Poulin},
  \citenamefont {Qarry}, \citenamefont {Somma},\ and\ \citenamefont
  {Verstraete}}]{Poulin2011}%
  \BibitemOpen
  \bibfield  {author} {\bibinfo {author} {\bibfnamefont {D.}~\bibnamefont
  {Poulin}}, \bibinfo {author} {\bibfnamefont {A.}~\bibnamefont {Qarry}},
  \bibinfo {author} {\bibfnamefont {R.}~\bibnamefont {Somma}}, \ and\ \bibinfo
  {author} {\bibfnamefont {F.}~\bibnamefont {Verstraete}},\ }\href {\doibase
  10.1103/PhysRevLett.106.170501} {\bibfield  {journal} {\bibinfo  {journal}
  {Phys. Rev. Lett.}\ }\textbf {\bibinfo {volume} {106}},\ \bibinfo {pages}
  {170501} (\bibinfo {year} {2011})}\BibitemShut {NoStop}%
\bibitem [{\citenamefont {Nakata}\ \emph {et~al.}(2017)\citenamefont {Nakata},
  \citenamefont {Hirche}, \citenamefont {Koashi},\ and\ \citenamefont
  {Winter}}]{Winter2017}%
  \BibitemOpen
  \bibfield  {author} {\bibinfo {author} {\bibfnamefont {Y.}~\bibnamefont
  {Nakata}}, \bibinfo {author} {\bibfnamefont {C.}~\bibnamefont {Hirche}},
  \bibinfo {author} {\bibfnamefont {M.}~\bibnamefont {Koashi}}, \ and\ \bibinfo
  {author} {\bibfnamefont {A.}~\bibnamefont {Winter}},\ }\href {\doibase
  10.1103/PhysRevX.7.021006} {\bibfield  {journal} {\bibinfo  {journal} {Phys.
  Rev. X}\ }\textbf {\bibinfo {volume} {7}},\ \bibinfo {pages} {021006}
  (\bibinfo {year} {2017})}\BibitemShut {NoStop}%
\bibitem [{\citenamefont {Chamon}\ \emph {et~al.}(2014)\citenamefont {Chamon},
  \citenamefont {Hamma},\ and\ \citenamefont {Mucciolo}}]{Chamon2014-cb}%
  \BibitemOpen
  \bibfield  {author} {\bibinfo {author} {\bibfnamefont {C.}~\bibnamefont
  {Chamon}}, \bibinfo {author} {\bibfnamefont {A.}~\bibnamefont {Hamma}}, \
  and\ \bibinfo {author} {\bibfnamefont {E.~R.}\ \bibnamefont {Mucciolo}},\
  }\href {\doibase 10.1103/PhysRevLett.112.240501} {\bibfield  {journal}
  {\bibinfo  {journal} {Phys. Rev. Lett.}\ }\textbf {\bibinfo {volume} {112}},\
  \bibinfo {pages} {240501} (\bibinfo {year} {2014})}\BibitemShut {NoStop}%
\bibitem [{\citenamefont {Bianchi}\ \emph {et~al.}(2022)\citenamefont
  {Bianchi}, \citenamefont {Hackl}, \citenamefont {Kieburg}, \citenamefont
  {Rigol},\ and\ \citenamefont
  {Vidmar}}]{Bianchi_Hackl_Kieburg_Rigol_Vidmar_2022}%
  \BibitemOpen
  \bibfield  {author} {\bibinfo {author} {\bibfnamefont {E.}~\bibnamefont
  {Bianchi}}, \bibinfo {author} {\bibfnamefont {L.}~\bibnamefont {Hackl}},
  \bibinfo {author} {\bibfnamefont {M.}~\bibnamefont {Kieburg}}, \bibinfo
  {author} {\bibfnamefont {M.}~\bibnamefont {Rigol}}, \ and\ \bibinfo {author}
  {\bibfnamefont {L.}~\bibnamefont {Vidmar}},\ }\href {\doibase
  10.1103/PRXQuantum.3.030201} {\bibfield  {journal} {\bibinfo  {journal} {PRX
  Quantum}\ }\textbf {\bibinfo {volume} {3}},\ \bibinfo {pages} {030201}
  (\bibinfo {year} {2022})}\BibitemShut {NoStop}%
\bibitem [{\citenamefont {Serbyn}\ \emph {et~al.}(2021)\citenamefont {Serbyn},
  \citenamefont {Abanin},\ and\ \citenamefont {Papić}}]{Serbyn_2021}%
  \BibitemOpen
  \bibfield  {author} {\bibinfo {author} {\bibfnamefont {M.}~\bibnamefont
  {Serbyn}}, \bibinfo {author} {\bibfnamefont {D.~A.}\ \bibnamefont {Abanin}},
  \ and\ \bibinfo {author} {\bibfnamefont {Z.}~\bibnamefont {Papić}},\ }\href
  {\doibase 10.1038/s41567-021-01230-2} {\bibfield  {journal} {\bibinfo
  {journal} {Nature physics}\ }\textbf {\bibinfo {volume} {17}},\ \bibinfo
  {pages} {675–685} (\bibinfo {year} {2021})}\BibitemShut {NoStop}%
\bibitem [{\citenamefont {Moudgalya}\ \emph {et~al.}(2022)\citenamefont
  {Moudgalya}, \citenamefont {Bernevig},\ and\ \citenamefont
  {Regnault}}]{Moudgalya_2022}%
  \BibitemOpen
  \bibfield  {author} {\bibinfo {author} {\bibfnamefont {S.}~\bibnamefont
  {Moudgalya}}, \bibinfo {author} {\bibfnamefont {B.~A.}\ \bibnamefont
  {Bernevig}}, \ and\ \bibinfo {author} {\bibfnamefont {N.}~\bibnamefont
  {Regnault}},\ }\href {\doibase 10.1088/1361-6633/ac73a0} {\bibfield
  {journal} {\bibinfo  {journal} {Reports on Progress in Physics}\ }\textbf
  {\bibinfo {volume} {85}} (\bibinfo {year} {2022}),\
  10.1088/1361-6633/ac73a0}\BibitemShut {NoStop}%
\bibitem [{\citenamefont {Bremner}\ \emph {et~al.}(2009)\citenamefont
  {Bremner}, \citenamefont {Mora},\ and\ \citenamefont
  {Winter}}]{Bremner_Mora_Winter_2009}%
  \BibitemOpen
  \bibfield  {author} {\bibinfo {author} {\bibfnamefont {M.~J.}\ \bibnamefont
  {Bremner}}, \bibinfo {author} {\bibfnamefont {C.}~\bibnamefont {Mora}}, \
  and\ \bibinfo {author} {\bibfnamefont {A.}~\bibnamefont {Winter}},\ }\href
  {\doibase 10.1103/PhysRevLett.102.190502} {\bibfield  {journal} {\bibinfo
  {journal} {Physical review letters}\ }\textbf {\bibinfo {volume} {102}},\
  \bibinfo {pages} {190502} (\bibinfo {year} {2009})}\BibitemShut {NoStop}%
\bibitem [{\citenamefont {Gross}\ \emph {et~al.}(2009)\citenamefont {Gross},
  \citenamefont {Flammia},\ and\ \citenamefont
  {Eisert}}]{Gross_Flammia_Eisert_2009}%
  \BibitemOpen
  \bibfield  {author} {\bibinfo {author} {\bibfnamefont {D.}~\bibnamefont
  {Gross}}, \bibinfo {author} {\bibfnamefont {S.~T.}\ \bibnamefont {Flammia}},
  \ and\ \bibinfo {author} {\bibfnamefont {J.}~\bibnamefont {Eisert}},\ }\href
  {\doibase 10.1103/PhysRevLett.102.190501} {\bibfield  {journal} {\bibinfo
  {journal} {Physical review letters}\ }\textbf {\bibinfo {volume} {102}},\
  \bibinfo {pages} {190501} (\bibinfo {year} {2009})}\BibitemShut {NoStop}%
\bibitem [{\citenamefont {Crutchfield}(2012)}]{Crutchfield2012}%
  \BibitemOpen
  \bibfield  {author} {\bibinfo {author} {\bibfnamefont {J.~P.}\ \bibnamefont
  {Crutchfield}},\ }\href {\doibase 10.1038/nphys2190} {\bibfield  {journal}
  {\bibinfo  {journal} {Nature Physics}\ }\textbf {\bibinfo {volume} {8}},\
  \bibinfo {pages} {17} (\bibinfo {year} {2012})}\BibitemShut {NoStop}%
\bibitem [{\citenamefont {Vidal}(2007)}]{Vidal2007}%
  \BibitemOpen
  \bibfield  {author} {\bibinfo {author} {\bibfnamefont {G.}~\bibnamefont
  {Vidal}},\ }\href {\doibase 10.1103/PhysRevLett.99.220405} {\bibfield
  {journal} {\bibinfo  {journal} {Phys. Rev. Lett.}\ }\textbf {\bibinfo
  {volume} {99}},\ \bibinfo {pages} {220405} (\bibinfo {year}
  {2007})}\BibitemShut {NoStop}%
\bibitem [{\citenamefont {Dowling}\ \emph
  {et~al.}(2023{\natexlab{d}})\citenamefont {Dowling}, \citenamefont {Modi},
  \citenamefont {Muñoz}, \citenamefont {Singh},\ and\ \citenamefont
  {White}}]{TeMERA2023}%
  \BibitemOpen
  \bibfield  {author} {\bibinfo {author} {\bibfnamefont {N.}~\bibnamefont
  {Dowling}}, \bibinfo {author} {\bibfnamefont {K.}~\bibnamefont {Modi}},
  \bibinfo {author} {\bibfnamefont {R.~N.}\ \bibnamefont {Muñoz}}, \bibinfo
  {author} {\bibfnamefont {S.}~\bibnamefont {Singh}}, \ and\ \bibinfo {author}
  {\bibfnamefont {G.~A.~L.}\ \bibnamefont {White}},\ }\href@noop {} {\enquote
  {\bibinfo {title} {Process tree: Efficient representation of quantum
  processes with complex long-range memory},}\ } (\bibinfo {year}
  {2023}{\natexlab{d}}),\ \Eprint {http://arxiv.org/abs/2312.04624}
  {arXiv:2312.04624 [quant-ph]} \BibitemShut {NoStop}%
\bibitem [{\citenamefont {Zurek}(2003)}]{Zurek2003}%
  \BibitemOpen
  \bibfield  {author} {\bibinfo {author} {\bibfnamefont {W.~H.}\ \bibnamefont
  {Zurek}},\ }\href {\doibase 10.1103/RevModPhys.75.715} {\bibfield  {journal}
  {\bibinfo  {journal} {Rev. Mod. Phys.}\ }\textbf {\bibinfo {volume} {75}},\
  \bibinfo {pages} {715} (\bibinfo {year} {2003})}\BibitemShut {NoStop}%
\bibitem [{\citenamefont {Leone}\ \emph
  {et~al.}(2021{\natexlab{b}})\citenamefont {Leone}, \citenamefont {Oliviero},
  \citenamefont {Zhou},\ and\ \citenamefont {Hamma}}]{Leone2021quantumchaosis}%
  \BibitemOpen
  \bibfield  {author} {\bibinfo {author} {\bibfnamefont {L.}~\bibnamefont
  {Leone}}, \bibinfo {author} {\bibfnamefont {S.~F.~E.}\ \bibnamefont
  {Oliviero}}, \bibinfo {author} {\bibfnamefont {Y.}~\bibnamefont {Zhou}}, \
  and\ \bibinfo {author} {\bibfnamefont {A.}~\bibnamefont {Hamma}},\ }\href
  {\doibase 10.22331/q-2021-05-04-453} {\bibfield  {journal} {\bibinfo
  {journal} {{Quantum}}\ }\textbf {\bibinfo {volume} {5}},\ \bibinfo {pages}
  {453} (\bibinfo {year} {2021}{\natexlab{b}})}\BibitemShut {NoStop}%
\bibitem [{\citenamefont {Anand}\ \emph {et~al.}(2021)\citenamefont {Anand},
  \citenamefont {Styliaris}, \citenamefont {Kumari},\ and\ \citenamefont
  {Zanardi}}]{Anand2021-yi}%
  \BibitemOpen
  \bibfield  {author} {\bibinfo {author} {\bibfnamefont {N.}~\bibnamefont
  {Anand}}, \bibinfo {author} {\bibfnamefont {G.}~\bibnamefont {Styliaris}},
  \bibinfo {author} {\bibfnamefont {M.}~\bibnamefont {Kumari}}, \ and\ \bibinfo
  {author} {\bibfnamefont {P.}~\bibnamefont {Zanardi}},\ }\href {\doibase
  10.1103/PhysRevResearch.3.023214} {\bibfield  {journal} {\bibinfo  {journal}
  {Phys. Rev. Res.}\ }\textbf {\bibinfo {volume} {3}},\ \bibinfo {pages}
  {023214} (\bibinfo {year} {2021})}\BibitemShut {NoStop}%
\bibitem [{\citenamefont {Hunter-Jones}(2018)}]{HunterJones2018ChaosAR}%
  \BibitemOpen
  \bibfield  {author} {\bibinfo {author} {\bibfnamefont {N.}~\bibnamefont
  {Hunter-Jones}},\ }\emph {\bibinfo {title} {Chaos and Randomness in
  Strongly-Interacting Quantum Systems}},\ \href {\doibase 10.7907/BHZ5-HV76}
  {Ph.D. thesis},\ \bibinfo  {school} {California Institute of Technology}
  (\bibinfo {year} {2018})\BibitemShut {NoStop}%
\bibitem [{\citenamefont {Pesin}(1977)}]{Pesin1977}%
  \BibitemOpen
  \bibfield  {author} {\bibinfo {author} {\bibfnamefont {Y.}~\bibnamefont
  {Pesin}},\ }\href {\doibase 10.1070/RM1977v032n04ABEH001639} {\bibfield
  {journal} {\bibinfo  {journal} {Russian Mathematical Surveys}\ }\textbf
  {\bibinfo {volume} {32}},\ \bibinfo {pages} {55} (\bibinfo {year}
  {1977})}\BibitemShut {NoStop}%
\bibitem [{\citenamefont {Chiribella}\ \emph {et~al.}(2009)\citenamefont
  {Chiribella}, \citenamefont {D'Ariano},\ and\ \citenamefont
  {Perinotti}}]{Chiribella2009}%
  \BibitemOpen
  \bibfield  {author} {\bibinfo {author} {\bibfnamefont {G.}~\bibnamefont
  {Chiribella}}, \bibinfo {author} {\bibfnamefont {G.~M.}\ \bibnamefont
  {D'Ariano}}, \ and\ \bibinfo {author} {\bibfnamefont {P.}~\bibnamefont
  {Perinotti}},\ }\href {\doibase 10.1103/PhysRevA.80.022339} {\bibfield
  {journal} {\bibinfo  {journal} {Phys. Rev. A}\ }\textbf {\bibinfo {volume}
  {80}},\ \bibinfo {pages} {022339} (\bibinfo {year} {2009})}\BibitemShut
  {NoStop}%
\bibitem [{\citenamefont {Oreshkov}\ and\ \citenamefont
  {Giarmatzi}(2016)}]{Oreshkov2016}%
  \BibitemOpen
  \bibfield  {author} {\bibinfo {author} {\bibfnamefont {O.}~\bibnamefont
  {Oreshkov}}\ and\ \bibinfo {author} {\bibfnamefont {C.}~\bibnamefont
  {Giarmatzi}},\ }\href {\doibase 10.1088/1367-2630/18/9/093020} {\bibfield
  {journal} {\bibinfo  {journal} {New J. Phys.}\ }\textbf {\bibinfo {volume}
  {18}},\ \bibinfo {pages} {093020} (\bibinfo {year} {2016})}\BibitemShut
  {NoStop}%
\bibitem [{\citenamefont {Shrapnel}\ \emph {et~al.}(2018)\citenamefont
  {Shrapnel}, \citenamefont {Costa},\ and\ \citenamefont
  {Milburn}}]{CostaBornRule}%
  \BibitemOpen
  \bibfield  {author} {\bibinfo {author} {\bibfnamefont {S.}~\bibnamefont
  {Shrapnel}}, \bibinfo {author} {\bibfnamefont {F.}~\bibnamefont {Costa}}, \
  and\ \bibinfo {author} {\bibfnamefont {G.}~\bibnamefont {Milburn}},\ }\href
  {\doibase 10.1088/1367-2630/aabe12} {\bibfield  {journal} {\bibinfo
  {journal} {New J. Phys.}\ }\textbf {\bibinfo {volume} {20}},\ \bibinfo
  {pages} {053010} (\bibinfo {year} {2018})}\BibitemShut {NoStop}%
\bibitem [{Note5()}]{Note5}%
  \BibitemOpen
  \bibinfo {note} {Markov's inequality states that for any non-negative random
  variable $X$ with mean $\protect \mathbb {E}(X)$, and for any $\delta > 0 $,
  $\protect \mathbb {P}\{X\geq \delta \} \leq \protect \mathbb {E}(X) / \delta
  $.}\BibitemShut {Stop}%
\bibitem [{Note6()}]{Note6}%
  \BibitemOpen
  \bibinfo {note} {By this we mean that the process $\mathinner {|{\Upsilon
  }\rangle }$ has a freedom in choosing what the ``actual'' initial state is.
  One has a gauge freedom in choosing any initial state that would
  tomographically lead to the same state $\mathinner {|{\Upsilon }\rangle
  }$.}\BibitemShut {Stop}%
\bibitem [{Note7()}]{Note7}%
  \BibitemOpen
  \bibinfo {note} {Generalised Pauli operators on a $d-$dimensional space are
  defined by the generators~\cite {Roberts2017-en} \begin {equation} X
  \mathinner {|{n}\rangle } = \mathinner {|{n+1}\rangle }, \hskip 1em\relax
  \protect \text {and} \hskip 1em\relax Z \mathinner {|{n}\rangle }= \protect
  \qopname \relax o{exp}{[2\pi i n/d]}\mathinner {|{n}\rangle }. \end
  {equation}}\BibitemShut {NoStop}%
\bibitem [{\citenamefont {Burgarth}\ \emph {et~al.}(2021)\citenamefont
  {Burgarth}, \citenamefont {Facchi}, \citenamefont {Ligab\`o},\ and\
  \citenamefont {Lonigro}}]{Burgarth2021}%
  \BibitemOpen
  \bibfield  {author} {\bibinfo {author} {\bibfnamefont {D.}~\bibnamefont
  {Burgarth}}, \bibinfo {author} {\bibfnamefont {P.}~\bibnamefont {Facchi}},
  \bibinfo {author} {\bibfnamefont {M.}~\bibnamefont {Ligab\`o}}, \ and\
  \bibinfo {author} {\bibfnamefont {D.}~\bibnamefont {Lonigro}},\ }\href
  {\doibase 10.1103/PhysRevA.103.012203} {\bibfield  {journal} {\bibinfo
  {journal} {Phys. Rev. A}\ }\textbf {\bibinfo {volume} {103}},\ \bibinfo
  {pages} {012203} (\bibinfo {year} {2021})}\BibitemShut {NoStop}%
\bibitem [{\citenamefont {Low}(2009)}]{Low_2009}%
  \BibitemOpen
  \bibfield  {author} {\bibinfo {author} {\bibfnamefont {R.~A.}\ \bibnamefont
  {Low}},\ }\href {\doibase 10.1098/rspa.2009.0232} {\bibfield  {journal}
  {\bibinfo  {journal} {Proceedings of the Royal Society A: Mathematical,
  Physical and Engineering Sciences}\ }\textbf {\bibinfo {volume} {465}},\
  \bibinfo {pages} {3289–3308} (\bibinfo {year} {2009})}\BibitemShut
  {NoStop}%
\end{thebibliography}


% \onecolumngrid
\appendix
\renewcommand{\thesubsection}{\Roman{subsection}}

\section{The Process Tensor} \label{ap:processes}
Here we supplement the details of the background Section~\ref{sec:processes}, in order to describe how the process tensor, a familiar object in open quantum systems, can be derived from the pure process tensor $\ket{\ups}$. In summary, the process tensor $\ups_B$ corresponds to the reduced, generally mixed state description of the pure process $\ket{\ups}$, when the final state on the space $\mc{H}_R$ is traced over at the end.

Measurement is necessarily invasive in quantum mechanics. Therefore, to construct such a multitime description we need to represent quantum measurements in a way which includes the resultant state. Arbitrary interventions are defined by the action on $\mc{H}_S$ by external \emph{instruments}, which mathematically are trace-non-increasing, completely positive (CP) and time independent maps, $ \mc{A}$. If a instrument is also trace-preserving (TP), then it is deterministic; e.g. a unitary map or a complete measurement (POVM). If an instrument is trace-decreasing, then it is non-deterministic; e.g. a particular measurement result. The trace of the outgoing state corresponds to the probability of this outcome occurring out of a complete measurement described by the set $\mc{J}=\{\mc{A}_{x_i}\}$,
\begin{equation} \label{eq:single_time_instrument}
    \mathbb{P}(x_j|\mc{J})=\tr[\rho^\prime]:=\tr[\mc{A}_{x_j}(\rho)],
\end{equation}
Note that in this work, calligraphic font Latin letters will generally be used for such superoperators - that is a map of a (density) operator - while standard font uppercase Latin or lowercase Greek letters will be used for operators (matrices).

Similarly, for multiple, consecutive interactions of a single quantum system at different times, the (in general subnormalized) outgoing state is 
\begin{equation}
    \rho^\prime= \tr_E[\mc{U}_k \mc{A}_{x_k} \mc{U}_{k-1} \dots \mc{U}_1 \mc{A}_{x_1} (\rho(t_0))  ] \label{eq:PT}
\end{equation}
 where $\mc{U}_j(\sigma) := \ex^{-i H (t_j-t_{j-1}) } \sigma \ex^{i H (t_j-t_{j-1}) }$ is the unitary superoperator describing the dilated system-environment ($\mc{H}_S \otimes \mc{H}_E$) evolution, and the trace is the partial trace over the environment ($\mc{H}_E$). For non-deterministic instruments, the trace of this final state gives the probability of measuring a sequence of outcomes $x_1,x_2,\dots,x_k$.

For rank-one instruments as considered throughout this work, $\mc{A}_{x_i}(\cdot) = A_{x_i} (\cdot) A_{x_i}^\dg $, where $A_{x_i}$ are defined as in Eq.~\eqref{eq:pure_process}. That is, by rank-one we mean that there is only a single Kraus operator for the CP map.

From $\ket{\ups}$ in Eq.~\eqref{eq:PurePT}, one can define the reduced state only on the `butterfly space' of interventions $B$,
\begin{equation} \label{eq:upsB}
    \ups_B := \tr_{R}[\ket{\ups}\bra{\ups}].
\end{equation}
 $\ups_B$ is called a process tensor~\cite{processtensor,processtensor2,milz2020quantum} (also called a `quantum comb'~\cite{Chiribella2008,Chiribella2009} or `process matrix'~\cite{Costa2016,Oreshkov2016}), and is used for example in determining exact, unambiguous multitime properties of open quantum systems. For one, it admits a multitime Born rule~\cite{Chiribella2009,CostaBornRule}.,
 \begin{equation}\label{eq:born}
    \mathbb{P}(x_k,\dots,x_1|\mc{J}_k,\dots,\mc{J}_1 ) = \tr[\rho^\prime] =: \tr[\ups_B \mathbf{A}_{\vec{x}}^\mathrm{T} ].
\end{equation}
where we have recalled the definition of the outgoing reduced state $\rho^\prime$ in Eq.~\eqref{eq:PT}, and defined the instrument tensor $\mathbf{A}_{\vec{x}}^\mathrm{T}$, equal to $ \ket{\vec{x}}\bra{\vec{x}}$ for time-local rank-one instruments. It is not relevant here to define  $\mathbf{A}_{\vec{x}}^\mathrm{T}$ in full generality; see e.g. Ref.~\cite{milz2020quantum} for more details.

 Regarding normalization of the process tensor, and of instruments, note that to get well defined probabilities in the generalized Born rule \eqref{eq:born}, we take the instruments to be supernormalized and the process to be normalized. For example, a sequence of unitary maps should give unit probability for any process tensor $\ups_B$. This is immediate to check for example for the maximally noisy process which has a uniformly mixed reduced Choi state on $\mc{H}_B$, 
\begin{equation}
    \tr[\mathbf{U}_k^\mathrm{T} \ups^{(\mathbb{H})}_B  ] = \frac{1}{d_S^{2k}}\tr[\mathbf{U}_k ] \overset{!}{=} 1.
\end{equation} 
This locally noisy process is relevant to this work, see the discussion around Eq.~\eqref{eq:typ_process}, and Section~\ref{sec:typ_process}.

\section{Proof of Results from Section~\ref{sec:mainresult}} \label{ap:proofs1}
Here we restate all formal results from Section~\ref{sec:mainresult}, together with a proof for each. 
\BRent*
\begin{proof}
    Assume $\ket{\vec{x}}$ and $\ket{y}$ are the Choi states of two multitime butterflies, with $\braket{\vec{x}|\vec{y}} = 0 $. First assume that $|\braket{\ups_{R|\vec{x}} | \ups_{R|\vec{y}}}|^2 =\epsilon \approx 0$. Then
    \begin{align}
        |\braket{\ups_{R|\vec{x}} | \ups_{R|\vec{y}}}|^2 &= |\tr_R[ \braket{\vec{y}| \ups_{BR} } \braket{\ups_{BR}| \vec{x}}]|^2\\
        &= |\braket{\vec{y} | \ups_B | \vec{x}}|^2=\epsilon \label{eq:orthog}
    \end{align}
    given that $\ket{\vec{x}}$ and $\ket{\vec{y}}$ are projections on the butterfly space alone. Given that $\ket{\ups}$ is a pure state, here we have used its Schmidt decomposition as in Eq.~\eqref{eq:schmidt},
    \begin{equation}
        \begin{split} \label{eq:schmidt1}
            \langle \ups | \vec{x}  \rangle \langle \vec{y} | \ups \rangle=& \sum_i \lambda_i \bra{\ups_B^{(\alpha_i)}}\bra{\ups_R^{(\beta_i)}} \left(| \vec{x}  \rangle_B \langle \vec{y} |_B \right)  \\
            & \quad \times \sum_j \lambda_j^* \ket{\ups_B^{(\alpha_j)}} \ket{\ups_R^{(\beta_j)}}\\
            =&\sum_i |\lambda_i|^2   \langle \vec{y}  \ket{\ups_B^{(\alpha_i)}} \bra{\ups_B^{(\alpha_i)}} \vec{x}  \rangle \\
            =& \bra{\vec{y}} \ups_B \ket{\vec{x}}.
        \end{split}
    \end{equation}
    Now, if this is true for any orthogonal butterflies $\ket{x}$ and $\ket{y}$, the only solution to Eq.~\eqref{eq:orthog} is if $\ups_B \propto\id + \epsilon \Omega$, where $\Omega$ is traceless with bounded operator norm, such that $\| \epsilon \Omega\|\leq \epsilon$. In the other direction, if $\ket{\ups_{BR}}$ is approximately maximally entangled, then $\ups_B = \id + \epsilon \Omega$ and so 
    \begin{equation}
        |\braket{\vec{y} | \ups_B | \vec{x}}|^2 = |\braket{\vec{y} | \id + \epsilon \Omega| \vec{x}}|^2 = \mc{O}(\epsilon^2) \approx 0.
    \end{equation}
\end{proof}

\randButt*

\begin{proof}
From the Schmidt decomposition in the splitting $BR_1:R_2$, the analogue of Eq.~\eqref{eq:schmidt1}, one can equivalently write 
\begin{equation}
    |\langle \vec{y} | \ups_{B   R_1} | \vec{x} \rangle |^2 =|\braket{\ups_{R_2 | \vec{x}} |\ups_{R_2 | \vec{y}}}|^2. 
\end{equation}
We will then use the following result, which we prove below using an application of Weingarten Calculus. For the Haar random sampling of two orthogonal projections $\{\ket{\vec{x}}, \ket{\vec{y}} \}$, the expectation value is
    \begin{align} 
        \mathbb{E}_{\mc{X} \sim \mathbb{H}}\left\{  {|\langle \vec{y} | \ups_{B   R_1} | \vec{x} \rangle |^2} \right\} &\underset{\vec{x} \perp \vec{y}}{=}\frac{d_{B    {R_1}}^2(\pur  -{1}/{d_{B    {R_1}}})}{d_{B    {R_1}}^2-1} \nn \\
        & \approx \pur -{1}/{ d_{B    {R_1}}}, \label{eq:2folda}
    \end{align}
    where we take $d_{B    {R_1}}^2-1 \approx d_{B    {R_1}}^2 $ in the second line.

    Two Haar random, orthogonal states $\{ \ket{\vec{x}},\ket{\vec{y}} \}_\mathbb{H}$ can be generated from any other, e.g. computational, orthogonal states $\{ \ket{0},\ket{1} \}$, given a random unitary matrix $U \in \mc{H}$, by identifying $\ket{\vec{x}} = U \ket{0}$ and $\ket{\vec{y}} = U \ket{1}$. Define $\Phi^{(2)}_\mathbb{H}(A)$ to be the $2-$fold average of the tensor $A \in \mc{H} \otimes \mc{H}$. We may use Weingarten Calculus to compute it explicitly~\cite{Roberts2017-en}
     \begin{align}
            &\Phi^{(2)}_\mathbb{H}(A):= \int dU U \otimes U (A) U^\dg \otimes U^\dg  \\
            &\,= \frac{1}{d^2-1} \Big(\id \tr[A] + \mathds{S} \tr[\mathds{S} A] -\frac{1}{d}\mathds{S} \tr[A] -\frac{1}{d} \id \tr[\mathds{S} A] \Big),\nn
    \end{align}
    where $\mathds{S} $ is the swap operation. By choosing $A \equiv \ups_{B    R_1} \otimes \ups_{B   R_1}$, we can rewrite the left hand side of Eq.~\eqref{eq:2folda} as 
    \begin{align} \nonumber
        LHS =& \int dU_{\mathbb{H}}\braket{\vec{x}| U \ups_{B    R_1} U^\dg |\vec{y}} \braket{\vec{y}| U \ups_{B    R_1} U^\dg |\vec{x}}|\nonumber\\
        =& \bra{\vec{x}} \bra{\vec{y}} \left(\int dU_{\mathbb{H}} U^{\otimes 2} (\ups_{B    R_1} \otimes \ups_{B    R_1}) U^{\dg \otimes 2} \right) \ket{\vec{y}} \ket{\vec{x}} \nonumber\\
        =&\bra{\vec{x}} \bra{\vec{y}} \Phi^{(2)}_\mathbb{H}(\ups_{B    R_1} \otimes \ups_{B    R_1})\ket{\vec{y}} \ket{\vec{x}} \label{eq:WeinCalc}\\
        =&\bra{\vec{x}} \bra{\vec{y}}\Big(\frac{1}{d_{B    {R_1}}^2-1} \Big(\id \tr[\ups_{B    R_1} \otimes \ups_{B    R_1}] \nn \\
        +& \mathds{S} \tr[\mathds{S} (\ups_{B    R_1} \otimes \ups_{B    R_1})]  -\frac{1}{d_{B    {R_1}}}\mathds{S} \tr[\ups_{B    R_1} \otimes \ups_{B    R_1}]  \nonumber\\
        -&\frac{1}{d_{B    {R_1}}} \id \tr[\mathds{S} (\ups_{B    R_1} \otimes \ups_{B    R_1})] \Big)  \ket{\vec{y}} \ket{\vec{x}}.\nonumber
    \end{align}
    For the first trace in this equation, we can directly evaluate
    \begin{equation}
            \tr[\ups_{B    R_1} \otimes \ups_{B    R_1} ]=\tr[\ups_{B    R_1} ]^2=1
    \end{equation}
    For the second trace, 
    \begin{equation}
            \tr[\mathds{S} (\ups_{B    R_1} \otimes \ups_{B    R_1})] = \tr(\ups_{B    R_1} ^2 ).
    \end{equation}
    Then, by the orthogonality of the butterflies, $\bra{\vec{x}} \bra{\vec{y}} \id \ket{\vec{y}} \ket{\vec{x}} = 0$ while $\bra{\vec{x}} \bra{\vec{y}} \mathds{S} \ket{\vec{y}} \ket{\vec{x}} = d_{B    {R_1}}^2 $, and so only the second and third terms in the final line of \eqref{eq:WeinCalc} survive. Using this, we arrive at Eq.~\eqref{eq:2folda}.

    We will now use this to prove Theorem~\ref{thm:ran_butterfly}. We can directly apply Eq.~\eqref{eq:2folda} to Markov's inequality,\textsuperscript{\footnote{Markov's inequality states that for any non-negative random variable $X$ with mean $\mathbb{E}(X)$, and for any $\delta > 0 $, $\mathbb{P}\{X\geq \delta \} \leq  \mathbb{E}(X) / \delta $.}}
    \begin{align}
        \mathbb{P}_{\mc{X} \sim \mathbb{H}}\Big\{ |\braket{\ups_{R_2 | \vec{x}} |\ups_{R_2 | \vec{y}}}&|^2 \geq \delta \Big\} =\mathbb{P}\left\{ {|\langle \vec{y} | \ups_{B   R_1} | \vec{x} \rangle |^2}  \geq \delta \right\}\nonumber\\
        &\leq \frac{ \mathbb{E}_{\mc{X} \sim \mathbb{H}}\left\{ {|\langle \vec{y} | \ups_{B   R_1} | \vec{x} \rangle |^2}  \right\}}{\delta} \\
        &\approx \frac{\pur -1/(d_{B    {R_1}}) }{\delta  }. \nn
    \end{align}
    \end{proof}

    \zetaST*

    \begin{proof}
    We will prove this via the contrapositive statement. Assume $\ket{\ups}$ has area-law spatiotemporal entanglement. Then the conditional states $\ket{\ups_{R|\vec{x}}}$ and $\ket{\ups_{R|\vec{y}}}$ can be represented efficiently by a MPS. Then they can be prepared from an auxiliary product state $\ket{\psi_0}$ using efficient unitaries $V_{\vec{x}}$ and $V_{\vec{y}}$, 
    \begin{equation}
        \ket{\ups_{R|\vec{x}}} = V_{\vec{x}} \ket{\psi_0} \quad \text{and }\ket{\ups_{R|\vec{y}}} = V_{\vec{y}} \ket{\psi_0}
    \end{equation}
    i.e. both $V_{\vec{x}}$ and $V_{\vec{y}}$ have an MPO representation with a constant bond dimension. It directly follows that for $V=V_{\vec{x}}^\dg V_{\vec{y}} \in \mc{R}_{\text{MPO}}$,
    \begin{equation}
        \zeta(\ups)=1.
    \end{equation}
\end{proof}

\stateChaos*
\begin{proof}
    Consider that the Butterfly Flutter Fidelity is $\zeta \approx 1$, for any two butterflies with Choi states $\ket{\vec{x}_i}$ and $\ket{\vec{x}_j}$ from some basis of butterflies $\{ \ket{\vec{x}_i}\}_{i=1}^{d_B^2}$. This means a simple (low depth) unitary $V_{ij}$ in Eq.~\eqref{eq:state_chaos} approximately `corrects' the final states, 
    \begin{equation}
         |\bra{\ups_{R|\vec{x}_i}} V_{ij}  \ket{\ups_{R|\vec{x}_j}}|^2 \approx 1,
    \end{equation}
    where as usual we define a simple unitary as one with an efficient MPO representation, such that it cannot create volume-law entanglement from an area-law state. Now assume that $\ket{\ups}$ is volume-law spatiotemporally entangled. In particular, this means that the final states $\ket{\ups_{R|\vec{x}_i}} $ and $\ket{\ups_{R|\vec{x}_j}} $ are both volume-law entangled quantum states. As $\ups_{R|\vec{x}_i}$ and $\ups_{R|\vec{x}_j}$ are connected via a simple circuit, we can write each of them in terms of some intermediate state
    \begin{equation}
        \ket{\ups_{R|\vec{x}_i}} = V_i \ket{R_0}
    \end{equation}
    where $V_i$ is a simple unitary, but $\ket{R_0}$ is volume-law entangled. As this is true for any $\ket{\vec{x}} \in \{ \ket{\vec{x}_i}\}_{i=1}^{d_B^2}$, this means that the full purified process can be written as 
    \begin{align}
         \ket{\ups}&= \sum_m \lambda_m (\id_B \otimes V_m) \ket{B_m R_0}\\
         &=\sum_m \lambda_m (\ket{B_m}\bra{B_m} \otimes V_m) \ket{B_0 R_0}
    \end{align}
    where by gauge freedom, $\ket{B_0 R_0}$ is the initial state of the process.\textsuperscript{\footnote{By this we mean that the process $\ket{\ups}$ has a freedom in choosing what the ``actual'' initial state is. One has a gauge freedom in choosing any initial state that would tomographically lead to the same state $\ket{\ups}$.}} However, $\sum_m \lambda_m (\ket{B_m}\bra{B_m} \otimes V_m)$ is a simple dynamics, in that it can be simulated efficiently with an MPO.   
\end{proof}

\section{Proofs from Section~\ref{sec:prev_chaos}} \label{ap:proofs2}

\Sdy*
\begin{proof}
    For a given bond dimension $\chi$ of a process $\ket{\ups}$, in the bipartition $B:R$, the entropy of a reduced state is upper bounded by that of a uniform eigenvalue distribution,
    \begin{equation}
        S(\ups_B) \leq -\sum_{i=1}^\chi \frac{1}{\chi} \log (\frac{1}{\chi})=\log(\chi).
    \end{equation}
     Then, recalling the characteristic scaling of bond dimension for different entanglement classes (see Section~\ref{sec:processes}), for large $k$, $S(\ups)$ is bounded by
     \begin{equation}
         S_{\mathrm{Dy}}(\Upsilon_{B}^\mathrm{(vol)}) \leq \underset{k\to \infty}{\lim} \frac{\log(d_S^{2k})}{k} = 2\log(d_S)
     \end{equation}
     for volume-law entanglement scaling.
    Otherwise, for area-law and sub-volume-law respectively
    \begin{align}
        &S(\Upsilon_{B}^{\mathrm{(area)}}) \leq \underset{k\to \infty}{\lim} \frac{\log(D)}{k} = 0, \text{ and} \\
        &S(\Upsilon_{B}^{\mathrm{(sub-vol)}}) \leq \underset{k\to \infty}{\lim} \frac{\log(d_S^{\log(k)})}{k} = 0,
    \end{align}
    where the second limit is computed via L'H\^{o}pital's rule. Hence $S(\Upsilon_{B})$ can only be non-zero for a volume-aw entangled process.
\end{proof}

\pesin*
\begin{proof}
    From the definition of the quantum $2-$R\'enyi entropy, 
    \begin{align}
        S^{(2)}(\ups_B)&= - \log \big( \tr[\ups_B^2] \big) \nn \\
        &=-\log \Big( \frac{1}{d_B^2}\sum_{\vec{x},\vec{y}}^{d_B^{2}}  |\bra{\vec{x}} \ups_B \ket{\vec{y}} |^2 \Big) 
    \end{align}
    where we have taken into account the supernormalization of instruments as in Eq.~\eqref{eq:normalization}. Assuming that the butterfly flutters are unitary, $ |\bra{\vec{x}} \ups_B \ket{\vec{y}} |^2 = 1$, and so
    \begin{align}
         S^{(2)}&(\ups_B) \nn\\
         &=-\log \Big( \frac{1}{d_B^2}\sum_{\vec{x}\neq \vec{y}}^{d_B^{2}-d_B}  |\bra{\vec{x}} \ups_B \ket{\vec{y}} |^2  - \sum_{\vec{x}}^{d_B}  |\bra{\vec{x}} \ups_B \ket{\vec{x}} |^2 \Big)  \nn \\
         &=-\log \Big( \frac{1}{d_B^2} \big(\sum_{\vec{x}\neq \vec{y}}^{d_B^{2}-d_B}  |\bra{\vec{x}} \ups_B \ket{\vec{y}} |^2  - d_B \big) \Big)) 
    \end{align}
\end{proof}

\TMI*
\begin{proof}
    Assuming that $\ket{\ups}$ is volume-law entangled, then $S(\ups_B) \approx \log(d_B)$, $ S(\ups_{R_1}) \approx d_{R_1}$, and in particular $I(B : R_1) = I(B : R_2) \approx 0$. Then from the definition of quantum mutual information, we have that
    \begin{align}
         I_3(B:R_1:R_2) =& I(B : R_1) + I(B : R_2) - I(B : R)\nn \\
         \approx&- I(B : R) \\
         =& - S(\ups_B) - S(\ups_{R}) + S(\ups_{B R}) \nn\\
         \approx& - 2 \log(d_B) \nn 
    \end{align}
    where we have used that $S(\ups_{B R}) = 0$ as it is an isolated system, and that $S(\ups_{B}) = S(\ups_{R}) = \log(d_B)$, given that $d_B < d_R$ and that for unitary dynamics, information is preserved for a single step process.
\end{proof}

\LOEbff*
\begin{proof}
    For the two butterfly flutters $X$ and $\id$, the Butterfly Flutter Fidelity \eqref{eq:state_chaos} is equal to
    \begin{equation}
        \zeta(\ups) := \tr[V U \ket{\psi}\bra{\psi} X^{\dg} U^\dg].
    \end{equation}
    If we enforce that $\zeta(\ups ) \overset{!}{=} 1$ for any initial state, then this means that
     \begin{equation}
        X^\dg U^\dg V U \overset{!}{=} \id  \label{eq:align}        
    \end{equation}
    which directly implies that the correction unitary is equal to
    \begin{equation}
        V = U X U^\dg= X_{-t}. \nn
    \end{equation}
    Now, as we have assumed that $V \in \mc{R}_{\mathrm{MPO}}$, this means that as an efficient MPO, $X_{-t}$ has bond dimensions that scale at most logarithmically, for any $t$. This means also that the MPO representation of $X_{t}$ has restricted bonds. Finally, this equivalently implies that the Choi state of the operator has a restricted bond dimension, for any $t$,
    \begin{equation}
        S(\ket{X_t}) \sim \mc{O}(\log(t)).
    \end{equation}
\end{proof}

\section{Construction of a Local Basis of Multitime Unitary Instruments} \label{ap:basis}
To construct a basis of unitary butterflies, one can do the following procedure:
\begin{enumerate}
    \item Choose any orthonormal basis of unitary matrices $\{\sigma_{i}^{(\ell)} \}_{i=1}^{d_S^2}$ for each time $t_\ell$, 
    \begin{equation} \label{eq:orthonormal}
        \tr[\sigma_i^{(\ell)} \sigma_j^{(\ell)} ] =d_S \delta_{ij}
    \end{equation}
    For example, this could be the generalized Pauli matrices.\textsuperscript{\footnote{Generalised Pauli operators on a $d-$dimensional space are defined by the generators~\cite{Roberts2017-en}
    \begin{equation}
        X \ket{n} = \ket{n+1}, \quad \text{and} \quad Z \ket{n}= \exp{[2\pi i n/d]}\ket{n}.
    \end{equation}}} These are taken to act on the system Hilbert space $\mc{H}_{S(t_\ell)}$.
    \item Using single-time CJI (see Section~\ref{sec:processes} and Fig.~\ref{fig:CJI}), map these operators to states by having them act on half of a maximally entangled state on the doubled space $\mc{H}_{S(t_\ell)} \otimes \mc{H}_{S(t_\ell)^\prime}$,
    \begin{equation}
        \ket{{x}_i^{(\ell)}}:=(\sigma_{i}^{(\ell) } \otimes \id) \ket{\phi^+}_{S(t_\ell) S^\prime(t_\ell)}.
    \end{equation}
    the orthonormality condition \eqref{eq:orthonormal} carries over to this representation,
    \begin{equation}
        \braket{x_i^{(\ell)}|x_j^{(\ell)}}=d_S\delta_{ij}.
    \end{equation}
    \item Do this for every time to arrive at a full basis for the butterfly space $\mc{H}_B \equiv  \mc{H}^\mathrm{o}_{S(t_k)} \otimes \mc{H}^\mathrm{io}_{S(t_{k-1})} \dots \otimes \mc{H}^\mathrm{io}_{S(t_2)} \otimes \mc{H}^\mathrm{io}_{S(t_1)}$, together with some portion of the final state, $\mc{H}_{  R_1}$,
    \begin{equation}
        \{ \ket{\vec{x}_i^\prime}\}_i^{d_{B    R_1}^2}:= \left\{\ket{{x}_{i_0}^{(0)}} \otimes \ket{{x}_{i_1}^{(1)}}  \otimes \dots \otimes \ket{{x}_{i_k}^{(k)}} \right\}_{i_0,i_1,\dots,i_k =1}^{d_{   R_1},d_S,\dots,d_S}.\nn
    \end{equation}
    This basis is local in time. 
\end{enumerate}

\section{The Lindblad-Bernoulli Shift} \label{sec:LB}
We here consider a somewhat pathological example which is not chaotic, yet looks so for many of the usual diagnostics. First proposed by Lindblad~\cite{Lindblad1986chaos} as a quantum counterpart to the Bernoulli shift classical stochastic process, the \emph{Lindblad-Bernoulli shift} describes a discrete dynamics which cyclically permutes an $n-$body $\mc{H}_S \otimes \mc{H}_E$ state, together with some local unitary $L$ on the $\mc{H}_S$ state,
\begin{equation}
    \mc{U}(\phi_1 \otimes \phi_2  \otimes \dots \otimes \phi_n) = (L\phi_2 L^\dg) \otimes \dots \otimes \phi_n \otimes \phi_1.
\end{equation}
We take the total system size $n$ to be large, compared to the number of time steps $k$ which we will consider, such that states fed into the process are essentially `lost' to the environment. Additionally, we take the total initial $\mc{H}_S \otimes \mc{H}_E$ state to be in a product state, $\ket{\psi}_{SE} = \phi_1 \otimes \dots \otimes \phi_n $.

Intuitively this system is highly regular. It simply cycles through different states in a straightforward manner, without any scrambling of local information. Any information put in to the process on the system level will never return from the large environment $\mc{H}_E$. Indeed, while this process has a maximal $B:R$ entanglement \ref{c1}, it is not volume-law spatiotemporally entangled \ref{c2} as the Choi state is a product state, which we will now show. 

For a single time step, the reduced map $\Lambda$ is uncorrelated with any other time step, where 
\begin{align}
    \Lambda( \rho) := \tr_{E} (\mc{U}_i(\ket{\psi_{SE}} \bra{\psi_{SE}}) ).
\end{align}
Then by the CJI, this channel acting on one-half of a maximally entangled state gives its Choi representation 
\begin{align}
    \Lambda_i \otimes \id \left( \ket{\phi^+} \bra{\phi^+} \right) = L \phi_i L^\dg \otimes \frac{\id}{d_S} ,
\end{align}
Then the total process Choi state, for $k<n$ time steps, is
\begin{equation} \label{eq:LB_Choi}
    \ups^{(\mathrm{LB})} = \bigotimes_{i=1}^k \left(L \phi_i L^\dg \otimes \frac{\id}{d_S} \right).
\end{equation}
Looking at the Peres-Loschmidt Echo for this example,
\begin{align}
     \zeta^\prime(\ups^{(\mathrm{LB})},\mc{X}_{\mathrm{LE}})|&= |\langle W_{\epsilon}^{\otimes k} | \ups^{(\mathrm{LB})} | \id^{\otimes k} \rangle|^2 \nn \\
     &= \prod_i^k |\langle W_{\epsilon} | L \phi_i L^\dg \otimes \frac{\id}{d_S} | \id \rangle|^2\nn \\
     &=\prod_i^k |\bra{\phi^+} (W_{\epsilon}^\dg \otimes \id) ( L \phi_i L^\dg  \otimes \frac{\id}{d_S}) \ket{\phi^+} |^2  \nn \\
     &=\prod_i^k |\tr[L^\dg W_{\epsilon}^\dg L \phi_i   ]|^2  \label{eq:LB_LE}\\
     &\leq \prod_i^k |\tr[L^\dg W_{\epsilon}^\dg L]  \tr[\phi_i]   ]|^2 \nn \\
     &= |\tr[W_{\epsilon}]  |^{2k} =(1-\epsilon )^{2k} d_B \approx 0. \nn
\end{align}
where $\ket{\phi^+}$ is the unnormalized, $d_S^{2}$ dimensional maximally entangled state on the space $\mc{H}^\mathrm{o}_{S} \otimes \mc{H}^\mathrm{i}_{S}$, which is equal to the vectorised identity matrix. In the final line we have used that for positive operators, $\tr(XY) \leq \tr(X) \tr(Y) $, that $\phi_i$ is a density matrix, and the weakness of the unitary perturbation $ |\langle W_{\epsilon} | \id \rangle| = \tr[W_{\epsilon}] = (1-\epsilon )d_S $. Recall that the prime on $\zeta$ means we neglect any correction unitary $\mc{H}_R$, setting it to $V = \id$. From the smallness of Eq.~\eqref{eq:LB_LE}, we see that $\zeta^\prime$ for the Peres-Loschmidt Echo misclassifies this example as chaotic. It detects that the butterfly orthogonalizes the entire final state on $\mc{H}_R$, but not that this effect scrambles throughout $\mc{H}_R$. 

Adding a correction unitary $V$ of consisting of a single layer of $k$ entirely local gates, the full Butterfly Flutter Fidelity given by $\zeta$ in Eq.~\eqref{eq:state_chaos} will correctly detect the Lindblad-Bernoulli shift as a regular dynamics, with $\zeta = 1$. This is because $V$ will act to align the part of the environment where the perturbation effect resides. For simplicity, we consider $\zeta$ where the butterflies $\ket{\vec{x}}$ and $\ket{\vec{y}}$ correspond to complete measurements and independent preparations at each time step, inputting pure orthogonal states $\psi_i^{x}$ or $\psi_i^{y}$ at time $t_i$. An analogous computation to Eq.~\eqref{eq:LB_LE}, similarly reveals the apparent chaoticity of the Lindblad-Bernoulli shift according to $\zeta$ if the correction unitary is chosen to trivially be the identity matrix. However, if we allow $V$ to be a unit depth circuit, from Eq.~\eqref{eq:state_chaos} we arrive at
\begin{align}
     \zeta(&\ups^{(\mathrm{LB})})= \underset{{V\in \mc{R}_{\mathrm{MPO}} }}{\mathrm{sup}}  \left(| \langle \ups_{R|\vec{x}}^{(\mathrm{LB})} | V | \ups_{R|\vec{y}}^{(\mathrm{LB})} \rangle|^2 \right) \nn \\
     &= \underset{{V\in \mc{R}_{\mathrm{MPO}}}}{\mathrm{sup}}  \left(| \langle \psi^\prime_{SE|\vec{x}} | V | \psi^\prime_{SE | \vec{y}} \rangle|^2 \right) \\ 
     &= \underset{{V\in \mc{R}_{\mathrm{MPO}}}}{\mathrm{sup}} |\bra{\psi_1^{x} ... \psi_k^{x} \phi_{k+1} ... \phi_n } V \ket{\psi_1^{y} ... \psi_k^{y} \phi_{k+1} ... \phi_n }|^2 \nn \\
     &= 1,\nn
\end{align}
where in the final line $V$ is chosen in the supremum to align the $k$ orthogonal sites of the final state with local unitaries, rendering $\ket{\psi_i^x}=\ket{\psi_i^y}$. This value being equal to one is indicative of the regularity of the dynamical system according to \ref{c3}. This example illustrates the advantage of the complete quantum chaos diagnostic defined in Eq.~\eqref{eq:state_chaos}. If the butterfly flutter's effect spreads in a complex way in the final state, then no simple $V$ would be able to completely align the final states. This is what the scrambling criterion \ref{c2} distinguishes. 

Much like the Peres-Loschmidt Echo, Dynamical Entropy cannot tell that the Lindblad-Bernoulli shift is regular. From the Choi representation, Eq.~\eqref{eq:LB_Choi}, we arrive immediately at a non-zero value,
\begin{equation}
    S_{\mathrm{Dy}}^{(k)} (\ups^{(\mathrm{LB})}) = \frac{k S(V \phi_i \otimes \frac{\id}{d_S} )}{k} = \log(d_S).
\end{equation}
In fact, Lindblad originally introduced this as a simple example of a chaotic process, given that it is unpredictable under repeated measurements~\cite{Lindblad1986chaos}. The stronger conditions \ref{c2}-\ref{c3} show directly that he was mistaken, that this dynamics is highly regular.

We note that in Ref.~\cite{Burgarth2021} a similar example can be found to the Lindblad-Bernoulli shift, with dynamics generated by Chebotarev-Gregoratti Hamiltonians. These models would exhibit similar properties of apparent chaoticity according to the Peres-Loschmidt Echo, but for an infinite environment and valid for continuous time evolution. We have examined the simpler Lindblad-Bernoulli Shift here as it is comparably instructive, and to remain within the paradigm of finite-dimensional isolated quantum systems.

\section{Typical Quantum Processes and Their Entanglement Structure} \label{ap:pedro}
Here we explain and utilize the main results from Refs.~\cite{FigueroaRomero_Modi_Pollock_2019,FigueroaRomero_Pollock_Modi_2021}, to prove concentration of measure bounds for typical processes. These results state that processes generated from random evolution -- be it fully Haar random or from an $\epsilon-$approximate $t-$design -- are likely to look Markovian when one only has access to repeated measurements on a small subsystem compared to the full unitarily evolving isolated system. In order to prove this result, Refs.~\cite{FigueroaRomero_Modi_Pollock_2019,FigueroaRomero_Pollock_Modi_2021} argue that in processes generated from random evolution, the Choi state is typically close to being maximally mixed, and therefore has little memory. For us, this means directly that the purified process (as described in Section~\ref{sec:processes}) must be volume-law entangled. The following result will therefore follow rather directly from the proofs of Refs.~\cite{FigueroaRomero_Modi_Pollock_2019,FigueroaRomero_Pollock_Modi_2021}. 

\buttTyp*

\begin{proof}
    We will utilize Levy's Lemma, which states that for some probability measure $\sigma$, and function $f(x)$ with $\delta>0$
    \begin{equation}
        \mathbb{P}_{x \sim \sigma}\left\{f(x) \geq \mathbb{E}_{\sigma}(f) + \delta  \right\} \leq \alpha_\sigma(\delta/L),
    \end{equation}
    where $L>0$ is the Lipschitz constant of $f$, which dictates how slowly varying $f$ is in the measure space $\sigma$. The function $\alpha_\sigma$ is the concentration rate, which we require to be vanishing in increasing $\delta$ to describe a concentration of measure~\cite{Popescu2006}.

    For our purposes, $\sigma$ will either be the full Haar measure $\mathbb{H}$, or an $\epsilon-$approximate $t-$design $\mu_{\epsilon-t}$, and $f(x)$ will be the deviation from maximum spatiotemporal entanglement, $\log(d_{B R_1}) - S^{(2) }(\ups_{B R_1})$. 

    We first review and modify the results of Ref.~\cite{FigueroaRomero_Modi_Pollock_2019} and Ref.~\cite{FigueroaRomero_Pollock_Modi_2021}, from which we define a concentration of measure result for $U_i \sim \mathbb{H}$ and $U_i \sim \mu_{\epsilon-t}$ respectively.

    In both cases, we will arrive at the general form of the concentration of measure,  
    \begin{equation} \label{eq:pedro}
         \mathbb{P}_{U_i \sim \mu}\Big\{ \|\ups_{B    R_1} -\frac{\id}{d_{B    R_1}}\|_{2}^2 \geq {\mc{J}_{\mu}^\prime(\delta)}\Big\} \leq \mc{G}_{\mu}(\delta).
    \end{equation}
    Theorem~\ref{thm:rand_process_zeta} will then follow from Eq.~\eqref{eq:pedro}, which we show at the end. 

    \subsection{Process From Haar Random Evolution} \label{sec:perfect_Haar}
    For dynamics generated by independent Haar random unitaries, given that $d_{B    R_1} = d_S^{2k}d_{   R_1} \approx d_S^{2k+1} < d_E \approx d_R$, the Haar average of the left hand side of the inequality within the brackets of Eq.~\eqref{eq:pedro} is
    \begin{align}
        &\mathbb{E}_{\mathbb{H}}(\| \ups_{B    R_1} -\frac{ \id_{B    R_1}}{d_{B     R_1} } \|_2^2) = \mathbb{E}_{\mathbb{H}}(\tr(\ups_{B    R_1}^2)) - \frac{1}{d_{B    R_1}} \label{eq:betadef}\\
        &\quad =\frac{d_E^2-1}{d_E(d_{SE}+1)} \left( \frac{d_E^2-1}{d_{SE}^2-1} \right)^k + \frac{1}{d_E}-\frac{1}{d_{S}^{2k+1}}\nn \\
        &\quad =:\mc{B} \nn
    \end{align}
    This is proved in the Appendix F in Ref.~\cite{FigueroaRomero_Modi_Pollock_2019}. Our setup is, however, a slightly modified version of this, as our final intervention is size $d_{   R_1}$ rather than $d_S$. In detail, as $\mc{H}_{R_1}$ is the final intervention space, one can modify the result Eq.~\eqref{eq:betadef} by changing the factors $\mc{A}$ and $\mc{B}$ in Eq.~(74) in Ref.~\cite{FigueroaRomero_Modi_Pollock_2019} to,
    \begin{align}
        &\mc{A}\to d_{SE}d_{   R_2}(d_{   R_2}^2+d_{   R_1 }^2 -2 ),\text{ and}\nn\\
        &\mc{B}\to d_{SE} d_{   R_2}(d_{SE}^2-1)
    \end{align}
    where we recall that $   \mc{H}_{R_2}$ is the complement to $\mc{H}_{R_1}$, such that $   \mc{H}_{R_1} \otimes      \mc{H}_{R_2} = \mc{H}_{R}$. Using this to simplify Eq.~(74) in \cite{FigueroaRomero_Modi_Pollock_2019}, we arrive a long, messy expression. However, it has same asymptotic behavior as Eq.~\eqref{eq:betadef}. More precisely, we make the approximation that $ d_{   R_1} \approx d_S \ll d_E$  -- which is valid for the assumption that $d_{   R_1} \ll d_R$. 
    
    In particular, if $d_E \gg 1$, such that $d_E-1 \approx d_E$, then in both cases 
    \begin{equation}
        \mathbb{E}_{\mathbb{H}}(\| \ups_{B    R_1} -\frac{ \id_{B    R_1}}{d_{BR} } \|_2^2) \approx \frac{1}{d_E}.
    \end{equation}

    Now, the concentration rate is the exponential function
    \begin{equation}
        \exp(\frac{-\delta^2(k+1) d}{4 L^2} )
    \end{equation}
    which is proved in the Appendices of Ref.~\cite{FigueroaRomero_Modi_Pollock_2019}. Now, the Lipschitz constants can also be bounded by almost the same quantity as in \cite{FigueroaRomero_Modi_Pollock_2019}, despite here having $f:=\|\rho-\id/d \|_2^2$ compared to $({1}/{2})\|\rho-\id/d \|_1$. This is because we can use that $\| X\|_2^2 \leq \| X\|_2 \leq \| X\|_1$, where we also have an additional factor of $2$ in $L$, given the $(1/2)$ factor in the definition of non-Markovianity $\mc{N}$ in \cite{FigueroaRomero_Modi_Pollock_2019}. We therefore arrive at 
    \begin{equation} \label{eq:concentration}
        \mathbb{P}\left\{\|\ups_{B    R_1} -\frac{\id}{d_{B    R_1}} \|_2^2  >  \mc{B}+\delta \right\} \leq  \exp [-\mc{C} \delta^2]
    \end{equation}
    for $\mc{B}$ defined in Eq.~\eqref{eq:betadef}, and 
    \begin{align}
        \mc{C} &:= \frac{(k+1) d_E d_S (d_S-1)^2}{8 (d_S^{k+1} -1)^2} \nn \\
        &=\frac{(k+1) d_E d_S}{8 (d_S^{k}+d_S^{k-1}+\dots +1)^2} \label{eq:cdef} 
    \end{align}
    defined from considerations above. Again, we have taken the approximation that $d_{   R_1}\approx d_S$.

    \subsection{Process From approximate unitary $t-$designs}
     If evolution is instead sampled from an $\epsilon-$approximate $t-$design, we can adapt the results from Ref.~\cite{FigueroaRomero_Pollock_Modi_2021}, which in turn builds on the deviation bounds for $k-$design results of Ref.~\cite{Low_2009}. The concentration of measure bound takes the form
    \begin{align}
        \mc{F}:= \Bigg\{& \left[\frac{16m}{(k+1)d_{SE} }\left( \frac{d_S^{k+1} -1}{d_S-1}\right)^2\right]^m +(\mc{B})^{m}   \\
        &+\frac{\epsilon }{16^m d_{SE}^t}\left(d_E^4 d_S^{2(k+2)}+ \frac{1}{d_S^{2k+1}} \right)\Bigg\}, \label{eq:fdef}
    \end{align}
    For any $0<m\leq t/4$, and $\delta>0$. $m$ can be chosen to optimize this bound, and overall $\mc{F}$ is small for $d_E \gg d_S^{2k+1}=d_{B    R_1}$, for a high $t$. See Ref.~\cite{FigueroaRomero_Pollock_Modi_2021} for details; a proof for this can be found in the methods section. Here we have slightly modified that result, as our object of interest is $\|\ups_{B    R_1} -\frac{\id}{d_{B    R_1}} \|_2^2 $ rather than a non-Markovianity measure. This means we replace $\mc{N}_\smallblackdiamond \to \|\ups_{B    R_1} -\frac{\id}{d_{B    R_1}} \|_2^2 $ and $\delta \to \sqrt{\delta}/2$, and do not have the $d_S^{3(2k+1)}$ factor on the right hand side of Theorem 1 of Ref.~\cite{FigueroaRomero_Pollock_Modi_2021}. In addition, we have made the same approximation as that considered above in Section~\ref{sec:perfect_Haar}, in that we take $d_{   R_1} \approx d_S$, valid for the asymptotic case where $d_E \gg d_{   R_1}$.

    Now, to complete the proof we note that 
    \begin{align}
        &\|\ups_{B R_1} -\frac{\id}{d_{BR_1}} \|_2^2 \geq \mc{J}_\mu^\prime (\delta )\nn \\
        \iff& \tr(\ups_{B R_1}^2) - \frac{1}{d_{BR_1}} \geq \mc{J}_\mu^\prime (\delta )\nn \\
        \iff& S^{(2)}(\ups_{B R_1})  \leq -\log(\frac{d_{B R_1} \mc{J}_\mu^\prime(\delta )}{d_{B R_1}} + \frac{1}{d_{B R_1 }})\\
         \iff& \log(d_{B R_1}) - S^{(2)}(\ups_{B R_1})  \geq \log({d_{B R_1} \mc{J}_\mu^\prime(\delta )} + 1 )\nn
    \end{align}
    where $S^{(2)}$ is the 2-R\'enyi entropy. Subbing this into the probability brackets of Eq.~\eqref{eq:pedro}, we arrive at Eq.~\eqref{eq:statechaos_typicality}, with 
    \begin{equation}
       \mc{J}_\mu(\delta ):=\log({d_{B R_1} \mc{J}_\mu^\prime(\delta )} + 1 )
    \end{equation}
\end{proof}

\end{document}